%% file: main.tex
\documentclass[sigconf, nonacm]{acmart}

\usepackage{amsmath, mathtools}
\usepackage{siunitx}
\usepackage{upgreek}
\usepackage{amsfonts}
\usepackage{enumitem}
\usepackage[normalem]{ulem}
\usepackage[linesnumbered,vlined,ruled,commentsnumbered]{algorithm2e}
\usepackage[para,online,flushleft]{threeparttable}
\usepackage{multirow}
\usepackage{multicol}
\usepackage{hyperref}
\usepackage{url}
\usepackage{bm}
\usepackage[labelfont=bf,skip=0pt]{caption,subfig}
\usepackage{flushend}

\usepackage{tabularx}
\usepackage{float}

\usepackage[english]{babel}
\usepackage{graphicx}
\usepackage{subcaption}

\usepackage{listings}
\usepackage{tablefootnote}
\usepackage{comment}
\usepackage{tikz}

\DeclareMathOperator*{\argmin}{arg\,min}

\newcommand{\G}{\boldsymbol{G}}
\newcommand{\E}{\boldsymbol{E}}
\newcommand{\Vnodes}{\mathcal{V}}
\newcommand{\Unodes}{\mathcal{U}}
\newcommand{\nodes}{\mathcal{D}}
\newcommand{\edges}{\mathcal{E}}
\newcommand{\pone}{\mathsf{WBC}}
\newcommand{\ptwo}{\mathsf{GWBC}}
\newcommand{\dpone}{\mathsf{dWBC}}
\newcommand{\dptwo}{\mathsf{dGWBC}}
\newcommand{\V}{\boldsymbol{V}}

\newcommand{\U}{\boldsymbol{U}}
\newcommand{\eps}{\varepsilon}
\renewcommand{\epsilon}{\varepsilon}
\renewcommand{\Re}{\mathbb{R}}

\newcommand{\cover}{{Q_\textsf{cover}}}

\newcommand{\dist}{\psi}

\newcommand{\node}{u}

\newcommand{\rects}{\mathcal{R}}
\newcommand{\rect}{r}

\newcommand{\edgecover}{\mathcal{S}}
\renewcommand{\cover}{S}
\newcommand{\base}{\alpha}

\allowdisplaybreaks

\newcommand{\baseone}{\textsf{Role}}
\newcommand{\basetwo}{\textsf{AMBEA}}
\newcommand{\basethree}{CPGR}
\newcommand{\ours}{\textsf{RT-BC}}

\newcommand{\ratio}{\hat{\phi}}
\newcommand{\bestRatio}{\phi^*}
\newcommand{\bestS}{S'}

\begin{document}

\title{Compact Representations of Geometric Bipartite Graphs via Weighted Biclique Covers}

\author{Aryan Esmailpour}
\orcid{0009-0000-3798-9578}
\affiliation{%
  \institution{Department of Computer Science University of Illinois Chicago}
  \country{USA}}

\author{Khoi Le}
\orcid{0009-0004-1953-8043}
\affiliation{%
  \institution{Department of Computer Science University of Illinois Chicago}
  \country{USA}}

\author{Stavros Sintos}
\orcid{0000-0002-2114-8886}
\affiliation{%
  \institution{Department of Computer Science University of Illinois Chicago}
  \country{USA}}






\input{abstract}
\maketitle

\input{intro}

\input{hardness}

\input{approx}

\input{exp2}

\input{relwork}
\input{conclusion}



\bibliographystyle{plainurl}
\balance
\bibliography{ref}

\newpage
\input{appendix}

\end{document}

%% file: abstract.tex
\begin{abstract}

Bipartite graphs are a fundamental representation for relational data arising in recommendation systems, market-basket analysis, social networks, and communication graphs. A key challenge in these settings is to store and transmit large bipartite graphs compactly while preserving exact structural and path information.
We study biclique-based representations of bipartite graphs $\G=(\V,\U,\E)$, where the edge set is encoded using a collection of complete bipartite subgraphs. We focus on the Weighted Biclique Covering problem, which minimizes the total number of vertices used across all bicliques, and introduce a generalized variant that additionally penalizes the number of bicliques, capturing practical overheads in storage, transmission, and model complexity.

While the weighted biclique covering problem is known to be $\mathsf{NP}$-Complete, we show that the generalized variant is also $\mathsf{NP}$-Complete.
Despite this hardness, many real-world bipartite graphs admit low-dimensional geometric embeddings or can be well approximated by them. Leveraging this observation, we develop the first approximation algorithms with provable guarantees for the (generalized) weighted biclique covering problem on geometric bipartite graphs.
Specifically, for $\delta$-disk bipartite graphs in low-dimensional $\ell_\infty^d$ spaces, we design a polynomial-time algorithm that achieves an $O(\log |\U| \cdot \log^d |\V|)$-approximation, combining ideas from greedy set cover, geometric range searching, and densest subgraph optimization. We also show how our algorithms extend to $\ell_\base^d$ metrics for any $\base\geq 1$.
Finally, we evaluate our algorithms on real-world bipartite datasets and show that they efficiently compute significantly smaller biclique-based representations than natural baselines, while scaling to large graphs.
\end{abstract}

%% file: intro.tex
\section{Introduction} \label{sec:intro}

Bipartite graphs provide a natural and expressive framework for modeling relational data arising in a wide range of applications, including social and biological networks, job matching, dating platforms, recommendation systems, and resource allocation~\cite{wellman1997structural}. As modern datasets continue to grow in scale, the corresponding graphs often contain massive numbers of vertices and edges, making both storage and algorithmic processing increasingly challenging. In particular, the cost of running classical graph algorithms can become prohibitive as graph size and density increase.

Graph representation and compression techniques have therefore emerged as essential tools for mitigating these challenges. A common goal of such methods is to reduce the size of the graph, typically by decreasing the number of explicitly stored edges, while preserving its essential structural properties. For many downstream tasks, however, it is crucial that the compressed representation preserves path information, namely reachability and shortest-path relationships between vertices. When this information is retained, the compressed graph can be used directly by path-sensitive algorithms
such as matching, flow, and shortest-path computations, leading to substantial improvements in efficiency for large and dense graphs.

\input{intro_figure}

\paragraph{Motivating scenario}
Consider a cloud computing environment consisting of multiple computing nodes connected by a network. A bipartite graph $\G(\V,\U,\E)$ is initially stored on a machine~$A$, but an expensive downstream graph algorithm, such as shortest paths, matching, or flow computation, must be executed on a computing node $B$. A natural approach is to transfer the graph from $A$ to $B$ and run the algorithm on the explicit edge representation. However, transferring the entire graph requires sending at least $|\E|$ objects over the network, and the subsequent running time of many classical graph algorithms also depends directly on the number of edges. For large and dense graphs, both the communication cost and the computational cost on $B$ can therefore become prohibitively expensive. This motivates the need for a compact representation that can be transmitted more efficiently and, at the same time, supports faster execution of downstream graph algorithms by replacing the original graph with a smaller path-preserving representation.

A variety of graph compression and representation techniques~\cite{bannach2024faster,buehrer2008scalable,dhulipala2016compressing,rossi2018graphzip, fan2021making} can be applied prior to transfer. However, most existing approaches incur substantial overhead when reconstructing complete path information, limiting their usefulness for downstream graph analytics. To address this limitation, prior work~\cite{FEDERcompress} and more recently~\cite{chavan2026speeding} proposed a biclique-based representation that preserves exact path information. Specifically, the edge set of the bipartite graph $\G$ is \emph{partitioned} into complete bipartite subgraphs (bicliques). For a biclique $\cover \subseteq \G$ with vertex sets $\V_\cover \subseteq \V$ and $\U_\cover \subseteq \U$, the $|\V_\cover|\cdot|\U_\cover|$ edges are replaced by a single auxiliary node $w_\cover$ connected to all vertices in $\V_\cover \cup \U_\cover$. This transformation replaces $|\V_\cover|\cdot|\U_\cover|$ edges by $|\V_\cover|+|\U_\cover|$ edges while preserving path information in the resulting graph. The size of biclique $\cover$ in the new representation is defined by $|\V_\cover|+|\U_\cover|$.

\paragraph{Limitations of biclique partitions}
While biclique partitions can yield compact graph representations, this line of work suffers from two important limitations. First, the requirement that the bicliques form a \emph{partition} of the edge set is overly restrictive. Under this requirement, each edge of $\G$ must belong to exactly one biclique, severely limiting the flexibility of the representation. Second, prior work does not explicitly study the optimization problem of constructing a representation of minimum size.

In contrast, we observe that requiring only a \emph{cover} of the edges by bicliques (allowing edges to appear in multiple bicliques) is sufficient to preserve path information using the same node-replacement construction. 
Furthermore, covering the edges by bicliques can lead to strictly smaller representations because every partition is a cover, however not every cover is a partition. 
Consequently, allowing coverings can only improve the quality of the resulting representation. Figure~\ref{fig:biclique_compression_tb_final} presents an example illustrating that a biclique cover yields a smaller representation.


\paragraph{Weighted biclique covering}
Motivated by these observations, we study the problem of covering the edges of a bipartite graph $\G$ with bicliques so as to minimize the size of the resulting representation, which is known  as the \emph{Weighted Biclique Cover} problem ($\pone$)~\cite{tarjan1975complexity}. Prior work on $\pone$ has primarily only focused on upper and lower bounds on the optimal solution size \cite{tarjan1975complexity,chung1983decomposition,davoodi2016edge,tuza1984covering}. 
The problem has also been studied in the role-mining literature under the name \emph{edge-RMP}~\cite{edge-RMP-hardness}, where it is shown to be $\mathsf{NP}$-complete. Existing approaches for $\pone$ (and equivalently edge-RMP)~\cite{role-mining-survey,newbery1989edge,edge-RMP-hardness} rely on heuristics and do not provide theoretical guarantees.

\paragraph{A generalized cost model}
While we provide the first approximation algorithms for the $\pone$ problem on certain classes of graphs, minimizing only the total size of the selected bicliques does not fully capture several realistic deployment scenarios. In many practical settings, an additional fixed cost is incurred for each selected biclique, independent of its size.
Returning to the cloud computing scenario, suppose that the edges of each biclique are transferred as a single packet. Each packet incurs a fixed overhead that is independent of its payload (e.g., a network header). Let $c$ denote this per-biclique overhead. To better capture such costs, we introduce a generalized objective that accounts for both the total size of the selected bicliques and their number. We formalize this extension as the \emph{Generalized Weighted Biclique Cover} problem ($\ptwo$), whose goal is to minimize the sum of the sizes of the selected bicliques plus $c$ times their number.
Although the $\ptwo$ problem has not been explicitly defined in prior work, it naturally interpolates between two known problems. When $c=0$, $\ptwo$ reduces to the $\pone$ problem. In contrast, for sufficiently large values of $c$, $\ptwo$ reduces to the \emph{bipartite dimension problem}, which seeks to minimize the number of bicliques needed to cover all edges of the input graph.
While inapproximability results and exact algorithms for specific graph classes are known for the bipartite dimension problem on specific graph classes~\cite{simon1990approximate,gruber2007inapproximability,chalermsook2014nearly,feige1998zero,khot2006better,zuckerman2006linear,amilhastre1998complexity,lubiw1991weighted,muller1996edge,chandran2017parameterized}, these results do not extend to $\ptwo$.



\paragraph{Geometric structure and our approach}
Despite the long history of $\pone$ and its relevance to data compression and distributed computation, no approximation algorithms with theoretical guarantees are known, either for $\pone$ or for its generalized variant $\ptwo$.
In this paper, we build on the observation that many graphs arising in practice are naturally \emph{disk graphs}~\cite{atminas2018forbidden, lahn2021faster,lokshtanov2024bipartizing} (also referred to as \emph{proximity graphs}~\cite{agarwal2024reporting}) or \emph{intersection graphs}~\cite{har2025approximating, har2022approximation, jana2023maximum} or can be well approximated as such. 
In these graphs, nodes are embedded as points in a geometric space, and edges connect pairs of nodes that are sufficiently close. 
For example, social networks and co-authorship graphs can often be embedded, with small error, into the space of user profiles with low intrinsic dimension~\cite{verbeek2014metric, yang2012defining}. 
More generally, for many networks across diverse domains (e.g., social, transportation, Internet), suitable embeddings exist that approximately preserve structural properties and shortest paths of the original graphs~\cite{verbeek2014metric, zhao2010orion, zhao2011efficient}. 
Although these graphs can be represented implicitly via their geometric embeddings, most classical graph algorithms require an explicit edge-based representation. Consequently, even geometrically embedded graphs must often be materialized and transmitted as conventional graphs.

In this paper, we leverage the geometric structure of bipartite disk and intersection graphs to design the first approximation algorithms with provable guarantees for both $\pone$ and $\ptwo$.

\vspace{-1em}
\subsection{Problem Definitions and our Contributions}
 Let $\G(\V, \U, \E)$ be a bipartite graph with two sets of vertices $\V$ and $\U$, with $n=|\V|+|\U|$, and a set of edges $\E\subseteq \V\times \U$, with $m=|\E|$.
 For a subgraph $G'\subseteq G$, let $\Vnodes(G')$ (or $\Unodes(G')$) be the vertices of $\V$ (or $\U$) that are included in $G'$. We also define $\nodes(G')=\Vnodes(G')\cup\Unodes(G')$, and $|\nodes(G')|=|\Vnodes(G')|+|\Unodes(G')|$.
 Let $\edges(G')$ be the set of edges in the subgraph $G'$.

\begin{definition}(Biclique)
    A subgraph $\cover\subseteq \G$ of a given bipartite graph $\G$ is a biclique (also called complete bipartite subgraph) if for every pair $v \in \Vnodes(\cover)$ and $u \in \Unodes(\cover)$, $(v,u)\in \edges(\cover)$.
\end{definition}

\begin{definition}(Biclique Edge Cover)
    A collection of subgraphs $\edgecover = \{\cover_1, \dots, \cover_\kappa\}$, is called a \textit{biclique edge covering} of a given graph $\G$, if for every $i=1,\ldots,\kappa$, $\cover_i$ is a biclique of $\G$, and for each $e \in \edges(G)$ there exists at least one $\cover_j \in \edgecover$, such that $e \in \edges(\cover_j)$.
\end{definition}

Next, we define two optimization problems studied in this paper.

\begin{definition}\label{def:wbc} (Weighted Biclique Covering $(\pone)$)
Given a bipartite graph $\G(\V, \U, \E)$ the goal is to construct a biclique edge cover $\edgecover = \{\cover_1, \dots, \cover_\kappa\}$ for $\G$, with minimum cost $\displaystyle \mu(\edgecover) =\!\!\! \sum_{\cover_i \in \edgecover}|\nodes(\cover_i)|$.
\end{definition}
\begin{definition}\label{def:gwbc}
    (Generalized Weighted Biclique Covering$(\ptwo)$) Given a bipartite graph $\G$ and a rational parameter $c \geq 0$ the goal is to construct a biclique edge cover $\edgecover = \{\cover_1, \dots, \cover_\kappa\}$ for $\G$, with minimum cost $\sigma(\edgecover) = c\cdot \kappa + \sum_{\cover_i \in \edgecover}|\nodes(\cover_i)|$. 
\end{definition}
In order to study the hardness we also define the decision version of the problems: Given a parameter $\tau>0$ the goal is to decide whether there exists a biclique edge cover $\edgecover$ such that $\mu(\edgecover)\leq \tau$ for the $\pone$ problem and $\sigma(\edgecover)\leq \tau$ for the $\ptwo$ problem. We refer to these problems as $\dpone$ and $\dptwo$, respectively.

For a real number $\beta>1$, an algorithm is a $\beta$-approximation algorithm for the $\ptwo$ (resp. $\pone$) problem if it returns a biclique edge cover $\edgecover$ such that $\sigma(\edgecover)\leq \beta\cdot \sigma(\edgecover^*)$ (resp. $\mu(\edgecover)\leq \beta\cdot \mu(\edgecover^*)$), where $\edgecover^*$ is the biclique edge cover with the minimum cost.

For any $\base\geq 1$, let $\dist_\base:\Re^d\times \Re^d\rightarrow \Re$ be the distance function in $\ell^d_\base$ metric, defined as
follows. For every pair of points 
$v,u\in \Re^d$, $\dist_\base(v,u)=\left(\sum_{j=\{1,\ldots,d\}}|v(j)-u(j)|^\base\right)^{1/\base}$ where $v(j)$ (resp. $u(j)$) is the $j$-th coordinate of point $v$ (resp. $u$). For $\base=\infty$ we have, $\dist(v,u)=\max_{j=\{1,\ldots, d\}}|v_j-u_j|$.


Next, we formally define the $\delta$-disk bipartite graph and the intersection bipartite graph. 
\begin{definition}($\delta$-disk bipartite graph)
A bipartite graph $\G(\V,\U,\E)$ is a $\delta$-disk bipartite graph in $\ell^d_\base$ metric if $\V$ and $\U$ are two finite sets of points in $\Re^d$, and for any pair of points $v\in \V$ and $u\in\U$ there exists an edge $(v,u)\in \E$ if and only if $\dist_\base(v,u)\leq \delta$.
\end{definition}

\newcommand{\ball}{\bigcirc}

A \emph{hyper-rectangle} $r$ in $\Re^d$ is defined as $r=\times_{j=1}^{d} I_j$, where each $I_j$ is an interval specifying the range along the $j$-th coordinate. A hyper-rectangle is called a \emph{box} if all its side lengths are equal, that is, if all intervals $I_j$ have the same length.

\begin{definition}(Intersection bipartite graph)
Graph $\G(\V\!,\!\U\!,\!\E\!,\! \rects_V\!,\! \rects_U\!)$ is an intersection bipartite graph in $\Re^d$ if every node $v\in \V$ ($u\in \U$) is associated with hyper-rectangle $r_v\in \rects_V$ ($r_u\in \rects_U$) in $\Re^d$,
and for any pair of points $v\in \V$ and $u\in \U$ there exists an edge $(v,u)\in \E$ if and only if $r_v\cap r_u\neq\emptyset$.
\end{definition}

\vspace{-0.5em}
\paragraph{Our Contributions}
Our results are summarized as follows.

$\bullet$ In Section~\ref{sec:hardnessnew}, we show that $\dptwo$ is $\mathsf{NP}$-complete.

$\bullet$ To the best of our knowledge, no approximation algorithms with provable guarantees have previously been studied for the $\pone$ problem, nor for its generalization $\ptwo$, either on general graphs or on restricted graph classes. 
In Section~\ref{sec:algs}, we present the first polynomial time $O(\log^{d+1} n)$-approximation algorithm for the $\ptwo$ (and $\pone)$ problem on $\delta$-disk graphs in the $\ell_{\infty}^d$ metric.

$\bullet$ In Section~\ref{sec:ext}, we extend our algorithmic framework to the more general class of intersection bipartite graphs. This extension yields also the first approximation algorithms for the $\pone$ and $\ptwo$ problems on interval graphs and $k$-nearest-neighbor bipartite graphs. Moreover, we discuss how our techniques for the $\pone$ and $\ptwo$ problems on $\delta$-disk graphs generalize to the $\ell_{\base}^d$ metric for any $\base \geq 1$. In particular, we present a practical $O(\lambda \log n)$-approximation algorithm that applies to any $\ell_\base^d$ metric, where $\lambda$ is the \emph{range cover size} parameter defined in Subsection~\ref{subsubsec:practical}.

$\bullet$ We implement our practical algorithm for $\delta$-disk graphs under both the $\ell_{\infty}^d$ and $\ell_2^d$ metrics, for both $\pone$ and $\ptwo$. In Section~\ref{sec:Newexp}, we evaluate our methods on real-world datasets and show that they consistently outperform all baselines, producing (generalized) weighted biclique covers that are significantly smaller (often by up to a factor of four) while achieving comparable running times and substantially lower memory usage.


%% file: intro_figure.tex
\usetikzlibrary{positioning, calc, shapes.symbols, shapes.geometric, arrows.meta, backgrounds, fit}

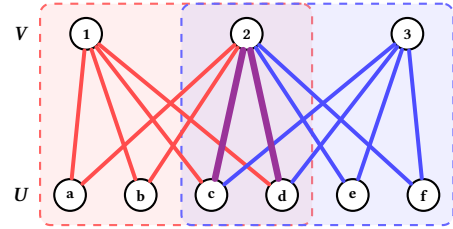
\begin{figure}[t]
    \centering
    \begin{tikzpicture}[
        scale=0.85, transform shape,
        vertex/.style={circle, draw, fill=white, thick, minimum size=14pt, inner sep=0pt, font=\footnotesize\bfseries},
        rededge/.style={line width=1.5pt, red!70, cap=round},
        blueedge/.style={line width=1.5pt, blue!70, cap=round},
        purpleedge/.style={line width=2.5pt, violet!80, cap=round},
        redgroup/.style={draw=red!60, fill=red!30, thick, dashed, rounded corners, inner sep=5pt, fill opacity=0.2},
        bluegroup/.style={draw=blue!60, fill=blue!30, thick, dashed, rounded corners, inner sep=5pt, fill opacity=0.2}
    ]

    \node[vertex] (u1) at (0.5, 2.5) {1};
    \node[vertex] (u2) at (3.0, 2.5) {2};
    \node[vertex] (u3) at (5.5, 2.5) {3};
    \node[font=\bfseries] at (-0.5, 2.5) {$\V$};

    \foreach \i/\l in {0/a, 1/b, 2/c, 3/d, 4/e, 5/f} {
        \node[vertex] (v\l) at (\i*1.1 + 0.25, 0) {\l};
    }
    \node[font=\bfseries] at (-0.5, 0) {$\U$};

    \begin{pgfonlayer}{background}
        \node[redgroup, fit=(u1) (u2) (va) (vd)] {};
        \node[bluegroup, fit=(u2) (u3) (vc) (vf)] {};
    \end{pgfonlayer}

    \foreach \target in {va, vb, vc, vd} \draw[rededge] (u1) -- (\target);
    \foreach \target in {va, vb} \draw[rededge] (u2) -- (\target);

    \foreach \target in {vc, vd, ve, vf} \draw[blueedge] (u3) -- (\target);
    \foreach \target in {ve, vf} \draw[blueedge] (u2) -- (\target);

    \draw[purpleedge] (u2) -- (vc);
    \draw[purpleedge] (u2) -- (vd);

    \end{tikzpicture}
    \caption{Example of graph compression using biclique edge covers. The red edges belong to the first biclique ($K_{2,4}$ defined by $\{1,2\} \times \{a,b,c,d\}$), and the blue edges belong to the second biclique ($K_{2,4}$ defined by $\{2,3\} \times \{c,d,e,f\}$). The purple edges belong to both bicliques. The total size of this cover is $6+6 = \mathbf{12}$. An optimal biclique partition (disallowing overlaps) is given by $\{1, 2\} \times \{a, b, c, d\}$, $\{2, 3\} \times \{e, f\}$, and $\{3\} \times \{c, d\}$, which results in a total size of $6 + 4 + 3 = \mathbf{13}$. 
    }
    \label{fig:biclique_compression_tb_final}
    \vspace{-1.5em}
\end{figure}

%% file: hardness.tex
\section{Hardness of GWBC}
\label{sec:hardnessnew}
In Appendix~\ref{sec:gwbc-hardold}, we show a simpler proof that holds for any constant $c$ reducing from the $\dpone$ problem. In this section we show a more involved proof that holds for any $c\geq 0$.
It is known that covering the edges of a bipartite graph with the
minimum number of bicliques (the bipartite dimension problem) is
NP-complete~\cite{orlin1977contentment}. Similarly, minimizing the
total size of the bicliques in a cover (the dWBC problem) is also
NP-complete~\cite{edge-RMP-hardness}. In this section, we show that
dGWBC is NP-complete for any rational value $c\geq 0$.

We use the reduction from Vertex Cover introduced
in~\cite{edge-RMP-hardness} for proving the hardness of dWBC.
Let $J=(W,F)$ be an instance of Vertex Cover, where
$W=\{w_1,\ldots,w_n\}$, $|F|=m$, together with an integer $K$.
The goal is to decide whether $J$ admits a vertex cover of size at
most $K$.

We construct a bipartite graph $\G=(\V,\U,\E)$ as follows.
For every vertex $w_i\in W$, we introduce two vertices
$x_i,y_i\in\V$ and one vertex $z_i\in\U$, and add the edges
$(x_i,z_i)$, and $(y_i,z_i)$.
Thus,
$N(z_i)=\{x_i,y_i\}$.

Next, for every edge $e=\{w_i,w_j\}\in F$, with $i<j$, we
introduce four new vertices
$p_e^1,p_e^2,p_e^3,p_e^4\in\V$
and five new vertices
$u_e^1,u_e^2,u_e^3,u_e^4,u_e^5\in\U$.
Their neighborhoods are defined as follows:
$N(u_e^1)=\{x_i,p_e^1,p_e^2\},
N(u_e^2)=\{y_j,p_e^2,p_e^3\},
N(u_e^3)=\{y_i,p_e^3,p_e^4\},
N(u_e^4)=\{x_j,p_e^4,p_e^1\},
N(u_e^5)=\{p_e^1,p_e^2,p_e^3,p_e^4\}$.
The vertices $p_e^1,\ldots,p_e^4$ are private to the gadget
corresponding to $e$, i.e., they do not occur in the gadget of any
other edge of $J$. This completes the construction of $\G$.

The above construction is exactly the bipartite graph used in the
reduction of~\cite{edge-RMP-hardness}. We next describe a useful
normal form for biclique covers of $\G$.

Consider any biclique edge cover $\edgecover$ of $\G$. Without loss
of generality, for every $w_i\in W$, the two edges incident to $z_i$
are covered in one of the following two ways:
\begin{enumerate}
    \item one biclique contains $z_i$ together with both $x_i$ and
    $y_i$; or
    \item two bicliques containing $z_i$ are used, one whose
    $\V$-side is $\{x_i\}$ and one whose $\V$-side is $\{y_i\}$.
\end{enumerate}

Indeed, since $N(z_i)=\{x_i,y_i\}$, any biclique containing $z_i$
has $\V$-side equal to $\{x_i\}$, $\{y_i\}$, or
$\{x_i,y_i\}$. Bicliques having the same singleton $\V$-side can
be merged by taking the union of their $\U$-sides. Such a merge
does not increase either the total biclique size or the number of
bicliques.

Furthermore, suppose that the cover contains the biclique
$\{x_i,y_i\}\times\{z_i\}$ and also some biclique whose $\V$-side
is $\{x_i\}$. We may add $z_i$ to the $\U$-side of the latter
biclique and replace $\{x_i,y_i\}\times\{z_i\}$ by
$\{y_i\}\times\{z_i\}$. This preserves all covered edges, does not
change the number of bicliques, and does not increase their total
size. The same argument applies if there exists a biclique whose
$\V$-side is $\{y_i\}$.

Thus, we may additionally assume that whenever the first case
occurs, there is no other biclique in $\edgecover$ whose $\V$-side
is exactly $\{x_i\}$ or exactly $\{y_i\}$. Since $c\geq0$, all of
these transformations do not increase
$\sigma(\edgecover)=\mu(\edgecover)+c|\edgecover|$.

We call $w_i$ \emph{selected} if the second case occurs.

The following observation gives the lower bounds on the number of
bicliques that will be needed for the additional
$c|\edgecover|$ term in the GWBC objective. All missing proofs in this section are shown in Appendix~\ref{appndx:NPhard}.

\begin{lemma}
\label{lem:gadget-biclique-count}
Suppose that $k'$ vertices of $J$ are selected. Then the bicliques
covering the vertex gadgets require at least $n+k'$ bicliques.
Moreover, every edge gadget requires at least four additional
bicliques. If neither endpoint of an edge $e\in F$ is selected,
then its gadget requires at least five additional bicliques.
These bounds can be summed over the edge gadgets.
\end{lemma}

We also use the corresponding lower bounds on the total biclique
size established by the reverse-direction analysis of the reduction
in~\cite{edge-RMP-hardness}. In our terminology, if $k'$ vertices are
selected, the vertex gadgets contribute at least
$3(n - k') + 4k' = 3n+k'$
to the total biclique size. In addition, every edge gadget
contributes at least $18$, and an edge gadget whose two endpoints
are both unselected contributes at least $20$. These are incremental
per-gadget bounds beyond the contribution of the vertex gadgets.

For the dGWBC instance, we set the target value to
\begin{align*}
\tau_c
&=3n+K+18m+c(n+K+4m) \notag\\
&=3n+18m+c(n+4m)+(1+c)K.
\end{align*}

We next establish the correctness of the reduction.

\begin{lemma}
\label{lem:gwbc-hardness}
The graph $J$ has a vertex cover of size at most $K$ if and only if
the constructed graph $\G$ admits a biclique edge cover
$\edgecover$ satisfying
$\sigma(\edgecover)\leq\tau_c$.
\end{lemma}

Finally, we conclude to the next theorem.

\begin{theorem}
\label{thm:gwbc-hardness}
For any rational value $c\geq0$, the dGWBC problem is NP-complete.
\end{theorem}

%% file: approx.tex
\newcommand{\pbg}{$\delta$-disk bipartite graph}
\vspace{-1em}
\section{Approximation Algorithms}
\label{sec:algs}
\newcommand{\rangetree}{\mathcal{T}}
\newcommand{\canonical}{\mathcal{C}}
We design an approximation algorithm for the $\ptwo$ problem for any $c\geq 0$ in \emph{$\delta$-disk bipartite graphs} in $\ell^d_\infty$ metric, where $d=O(1)$.
Since our algorithm also works for $c=0$, it is also valid for the $\pone$ problem.
In the $\ell^d_\infty$ metric a ball with center a point $u\in \Re^d$ and radius $\delta$ is a box $\rho_u$ in $\Re^d$ with center $u$ and equal size side lengths $2\delta$. It is straightforward that if $\G$ is a \pbg\ then for every $v\in \V$ and $u\in \U$, $(v,u)\in \E$ if and only if $v\in \rho_u$.
Throughout the section, we consider that $n_V=|\V|$, $n_U=|\U|$, $n=n_V+n_U$, $m=|\E|$, and without loss of generality $n_V\leq  n_U$.

\paragraph{High level ideas}
Before we proceed with the technical details we give the high level ideas of our algorithm. We first focus on $\pone$ and then extend to $\ptwo$.
We start by reducing $\pone$ to an instance of the weighted set cover problem.
Each edge corresponds to an element, and each biclique to a set including its edges, with weight equal to the number of nodes in the biclique.
We apply the classical greedy $O(\log n)=O(\log n_U)$-approximation algorithm for weighted set cover, which also yields an $O(\log n_U)$-approximation for $\pone$.
Let $\Delta$ denote the set of uncovered edges at the beginning of an iteration of the greedy algorithm.
The algorithm selects the biclique $\cover^*$ that minimizes $\frac{|\nodes(\cover^*)|}{|\edges(\cover^*)\cap\Delta|} $.
However, a fundamental obstacle arises: the number of bicliques in a graph may be exponential in the number of nodes.
As a result, a naive implementation of the greedy algorithm is computationally expensive.

A key observation is that the greedy algorithm does not require selecting the optimal biclique in each iteration.
It suffices to find a biclique $\cover$ such that
$\frac{|\nodes(\cover)|}{|\edges(\cover)\cap\Delta|} 
\leq
\hat{\gamma}\cdot \frac{|\nodes(\cover^*)|}{|\edges(\cover^*)\cap\Delta|}$.
In this case, the greedy algorithm yields an $O(\hat{\gamma}\cdot \log n_U)$-approximation.

A natural first attempt is to compute a \emph{densest subgraph} of $G$ restricted to uncovered edges.
The densest subgraph problem finds a subgraph that maximizes the ratio of the number of edges over the number of nodes in polynomial time~\cite{goldberg1984finding}. Although this objective matches the greedy criterion for our algorithm, the resulting subgraph is not necessarily a biclique.
To address this, we exploit the geometric structure of $\delta$-disk graphs using a range tree~\cite{de2008computational} built over $\V$. 
A range tree creates $O(n_V\log^{d-1}n_V)$ hyper-rectangles $\canonical$, such that for any hyper-rectangle $\rho$ in $\Re^d$ there are always $O(\log^d n_V)$ 
pairwise disjoint hyper-rectangles in $\canonical$ that contains exactly the points in $\rho\cap \V$. 

Using this property we show that there should exist a hyper-rectangle $x\in \canonical$ and a subgraph $G_{x}(V_x,U_x,E_x)$ such that i) $V_x\subseteq \V$, ii) $U_x$ is the set of points in $\U$ within distance $\delta$ from all points in $V_x$, iii) $E_x$ is the set of uncovered edges between $V_x$ and $U_x$, and iv)
the densest subgraph of $G_x$ contains at least $\frac{1}{O(\log^d n_V)}\cdot|\edges(\cover^*)\cap \Delta|$ edges. Using these observations, we show that the procedure of iterating over all $x\in \canonical$ and computing the densest subgraph in $G_x$ will lead to an $O(\log(n_U)\log^{d}(n_V))$-approximation algorithm for the $\pone$ problem in polynomial time.

For the $\ptwo$ problem, we proceed in a similar manner. As before, we reduce $\ptwo$ to an instance of the weighted set cover problem; however, in this case, the weight of each biclique (set) is defined as the number of its vertices plus the parameter $c$. 
Due to the presence of the parameter $c$, the classical densest subgraph formulation is no longer appropriate. Instead, we introduce the \emph{$c$-densest subgraph} problem, in which, given a graph, the objective is to find a subgraph that minimizes the ratio of the number of vertices plus $c$ to the number of edges. We show how known results in graph theory yield efficient exact or approximate solutions for the $c$-densest subgraph problem.


\vspace{-1em}
\subsection{Expensive approx. algorithm for $\ptwo$}
\label{subsec:expApprox}
Given a bipartite graph $\G(\V,\U,\E)$, we map $\ptwo$ to the well-known \emph{weighted set cover problem}. 
In the weighted set cover problem we are given a universe $\mathcal{A}$ of $N$ items and a family of subsets $\boldsymbol{T}=\{T_1,\ldots, T_{M}\}$ of $\mathcal{A}$ such that each subset $T_i$ is associated with a weight $w_i\geq 0$. The goal is to select a family of sets $T'\subseteq \boldsymbol{T}$ such that every element $a\in \mathcal{A}$ is covered by at least one set in $T'$ and $\sum_{T_i\in T'}w_i$ is minimized. It is known that a greedy algorithm that in each iteration selects the set $T_i$ with the smallest ratio $\frac{w_i}{b_i}$, where $b_i$ is the number of uncovered items in $\mathcal{A}$ contained in $T_i$, is an $O(\log N)$-approximation for the weighted set cover problem~\cite{chvatal1979greedy}.

We can straightforwardly map our problem to an instance of the weighted set cover problem. The universe is defined as $\mathcal{A}=\E$, i.e., all edges we must cover. For every biclique $\cover$ in $\G$ we construct a set $T_{\cover}=\edges(\cover)$ with weight $w_{\cover}=c+|\nodes(\cover)|$.
For every set cover in the set system $(\mathcal{A},\boldsymbol{T})$ of weight $W$ there exists a biclique cover $\edgecover$ in $\G$ with cost $\sigma(\edgecover)=W$, and vice versa.
Thus, the greedy algorithm that chooses in each iteration the biclique $\cover$ that minimizes $\frac{w_{\cover}}{b_{\cover}}$, where $b_{\cover}$ is the number of uncovered edges in $\edges(\cover)$, returns an $O(\log n_U)$-approximation for the $\ptwo$ problem. Even more generally, any algorithm that chooses in each iteration a biclique $\cover$ such that $\frac{w_{\cover}}{b_{\cover}}\leq \hat{\gamma}\cdot \min_{\cover'}\frac{w_{\cover'}}{b_{S'}}$, returns an $O(\hat{\gamma}\cdot \log n_U)$-approximation for the $\ptwo$ problem.
In the next section, for \pbg s we show how we can compute in polynomial time an $O(\log^d n_V)$-approximation of the biclique with the minimum ratio of weight over number of newly covered elements.


\vspace{-1em}
\subsection{Fast approx. algorithm for $\ptwo$}
\label{subsec:fastapprox}

Before describing our algorithm, we first introduce two technical tools that are central to our approach.

\subsubsection{Range tree~\cite{de2008computational}}
A range tree is a tree-based geometric data structure that is used for range aggregation queries in $\Re^d$, where $d=O(1)$. While a range tree has various interesting properties, we only mention the main property that we need for our algorithm. Let $P$ be a set of points in $\Re^d$ and let $\rangetree$ be the range tree constructed over $P$. The range tree $\rangetree$ constructs a set $\canonical$ of $O(|P|\cdot \log^{d-1}(|P|))$ canonical axis-aligned hyper-rectangles in $\Re^d$. Let $\rho$ be any hyper-rectangle in $\Re^d$. We know that there exists a subset of $O(\log^d |P|)$
pairwise disjoint hyper-rectangles $C(\rho)\subseteq \canonical$ such that $\bigcup_{x\in C(\rho)}(P\cap x)=P\cap \rho$. 
For every $x\in \canonical$ let $\bar{x}$ be the minimum volume axis-aligned hyper-rectangle such that $x\cap P=\bar{x}\cap P$. Without loss of generality, we assume every $x$ in $\canonical$ is minimal, i.e., $x=\bar{x}$.


\subsubsection{$c$-densest subgraph}
Given a bipartite graph $\G(\V,\U,
\E)$, the problem of computing the densest subgraph asks to compute a subgraph $X\subseteq \G$ such that $\frac{|\edges(X)|}{|\nodes(X)|}$ is maximized. 
Two common techniques to solve the densest subgraph problem are the max-flow based algorithm~\cite{goldberg1984finding}, and an LP-based algorithm~\cite{charikar2000greedy}.
There exists also a greedy $2$-approximation algorithm that runs in time linear to the size of the input graph~\cite{charikar2000greedy}.
For the $\ptwo$ problem we are interested in a variation of the densest subgraph problem that we call \emph{$c$-densest subgraph problem}. 
Given a bipartite graph $\G(\V,\U,\E)$ and a rational number $c\geq 0$, the goal is 
compute a subgraph $X\subseteq \G$ such that $\frac{|\edges(X)|}{|\nodes(X)|+c}$ is maximized.

We defer to 
Appendix~\ref{appndx:cdensest} 
the details of how the known exact max-flow–based algorithm for the standard densest subgraph problem can be extended to obtain an exact polynomial-time algorithm for the $c$-densest subgraph problem. For our purposes even a $O(1)$-approximation algorithm for the $c$-densest subgraph problem is sufficient.
As shown in Appendix~\ref{appndx:cdensest}, Chekuri et al.~\cite{chekuri2022densest} gives an $O(1)$-approximation for the $c$-densest subgraph problem in linear time.
Putting everything together, we get the next result.

\renewcommand{\deg}{\mathsf{deg}}
\vspace{-0.5em}
\begin{lemma}
\label{lem:denssub}
    Given a bipartite graph $\G(\V,\U,\E)$ and a rational number $c\geq 0$, there exists an exact algorithm for the $c$-densest subgraph problem that runs in $O((|\V|+|\U|)\cdot |\E|\cdot \log(\max\{|\V|,|\U|,c\}))$ time. There also exists an $O(1)$-approximation algorithm with running time $O(|\V|+|\U|+|\E|)$.
\end{lemma}
Let $\mathsf{DenSub}(\G,c)$ denote the procedure that returns an $O(1)$-approximation of the $c$-densest subgraph problem on a bipartite graph $\G$.

\begin{algorithm}[t]
\caption{\ours}
\label{alg:gwbc-rangetree}

\KwIn{A $\delta$-disk bipartite graph $\G(\V,\U,\E)$ in $\ell_\infty^d$, and $c\geq 0$}
\KwOut{A biclique edge cover $\edgecover$ of $\G$}

$\edgecover\gets \emptyset,\quad i=1$\;
Initialize hash table $B$ with $B(v,u)\gets 1$ for every $(v,u)\in \E$\;

Build $\rangetree$ on $\V$ and let $\canonical$ be its canonical hyper-rectangles\;

\While{there exists $(v,u)\in \E$ with $B(v,u)=1$}{
    $\bestRatio \gets +\infty,\quad \bestS \gets \emptyset$\;

    \ForEach{$x\in \canonical$}{
        $V_x \gets \V\cap x,\quad U_x \gets \{u\in \U \mid x\subseteq \rho_u\}$\;
        $E_x \gets \{(v,u)\in \E \mid v\in V_x,\ u\in U_x,\ B(v,u)=1\}$\;


        Construct $G_x(V_x,U_x,E_x)$\;
        $S^{(x)} \gets \mathsf{DenSub}(G_x,c),\quad b_x \gets |E(S^{(x)})|$\;

        \lIf{$b_x=0$}{\textbf{continue}}

        $\ratio \gets \frac{|\nodes(S^{(x)})|+c}{b_x}$\;
        \lIf{$\ratio < \bestRatio$}{$\bestRatio \gets \ratio, \quad \bestS \gets S^{(x)}$}
    }

    Define biclique $S_i$ induced by $\nodes(\bestS)$\;
    \ForEach{$v\in \Vnodes(S_i),\, u\in \Unodes(S_i)$}{
        $B(v,u)\gets 0$\;
    }

    $\edgecover\gets \edgecover\cup\{S_i\}$\;
    $i\gets i+1$\;
}

\Return{$\edgecover$}\;
\end{algorithm}
\vspace{-0.5em}
\subsubsection{Approximation algorithm}
\label{subsec:mainalg}
As we mentioned in Subsection~\ref{subsec:expApprox} we will implement an approximate version of the greedy algorithm, where in each iteration $i$  it computes a biclique $S_i$ such that $\frac{w_{S_i}}{b_{S_i}}\leq O(\log^d n_V)\min_{S'_i}\frac{w_{S'_i}}{b_{S'_i}}$, where $w_{S_i}=|\nodes(S_i)|+c$ and $b_{S_i}$ is the number of edges in $\edges(S_i)$ that have not been covered by bicliques already selected in the previous iterations.

We initialize a hash table $B$ such that for each edge $(v,u)\in\E$, we set $B(v,u)=1$ if $(v,u)$ is not yet covered by the current biclique cover $\edgecover$, and $B(v,u)=0$ otherwise.
We also construct a range tree $\rangetree$ over  $\V$ and let $\canonical$ be the set of its canonical hyper-rectangles.

Next, consider any iteration of the greedy algorithm, where $\edgecover$ contains the currently selected bicliques and hash-table $B$ is up to date.
For example, consider the beginning of the $i$-th iteration where $\edgecover=\{S_1,\ldots, S_{i-1}\}$.
We go through every canonical hyper-rectangle $x\in \canonical$. We construct the bipartite graph $G_x(V_x,U_x,E_x)$ where $V_x=\V\cap x$, $U_x=\{u\in \U\mid x\subseteq \rho_u\}$, and $E_x=\{(v,u)\in \E\mid v\in V_x, u\in U_x, B(v,u)=1\}$. Recall that $\rho_u$ is a box in $\Re^d$ with center $u$ and equal side lengths $2\delta$. In other words, we create the bipartite graph that contains all vertices from $\V$ in $x$, all vertices from $\U$ whose boxes completely cover $x$ and all uncovered edges between these vertices.
We call $\mathsf{DenSub}(G_x,c)$ and let $S^{(x)}$ be the returned subgraph of $G_x$. 
We set $x_i=\argmin_{x\in \canonical}\frac{|\nodes(S^{(x)})|+c}{b_{S^{(x)}}}$, and let $S_i$ be the biclique
such that $\nodes(S_i)=\nodes(S^{(x_i)})$ and $\edges(S_i)$ consists of all edges in $\E$ between $\Vnodes(S^{(x_i)})$ and $\Unodes(S^{(x_i)})$.
As we show in the correctness proof, $S_i$ is a valid biclique of $\G$.
For every $v\in\Vnodes(S_i)$ and $u\in\Unodes(S_i)$ we set $B(v,u)=0$. We add $S_i$ in $\edgecover$ and continue with the next iteration of the greedy algorithm. The greedy algorithm terminates when all edges are covered, i.e., $B(v',u')=0$ for every edge $(v',u')\in\E$.
We present the pseudocode of our algorithm in Algorithm~\ref{alg:gwbc-rangetree}.


\paragraph{Correctness and analysis}
Recall that without loss of generality we consider that $n_V\leq n_U$.
\begin{lemma}
    \label{lem:corr}
    $\sigma(\edgecover)\leq O(\log(n_U)\log^{d}(n_V))\sigma(\edgecover^*)$, where $\edgecover^*$ is the biclique edge cover that minimizes $\sigma(\edgecover^*)$.
\end{lemma}
\begin{proof}
    From Subsection~\ref{subsec:expApprox} it is sufficient to show that in each iteration $i$ of our algorithm we select a biclique $S_i$ such that $\frac{|\nodes(S_i)|+c}{b_{S_i}}\leq O(\log^d n_V) \frac{|\nodes(S_i^*)|+c}{b_{S_i^*}}$, where $b_{S_i}=|\edges(S_i)\setminus (\bigcup_{j\in[i-1]} \edges(S_j))|$ and $S^*_i$ is the biclique that minimizes $ \frac{|\nodes(S_i^*)|+c}{b_{S_i^*}}$.

    First, it is straightforward to see that $S_i$ is a valid biclique in $\G$. In fact, for every $x\in \canonical$ if $v\in \Vnodes(S^{(x)})$ and $u\in \Unodes(S^{(x)})$, it holds that $(v,u)\in\E$. Indeed, notice that by the construction of $V_x$ and $U_x$, for every $v\in V_x$ and $u\in U_x$, $\dist(v,u)\leq \delta$. Hence for every vertex $v$ in the subset of $V_x$ that is selected in $S^{(x)}$ and every vertex $u$ in the subset of $U_x$ that is selected in $S^{(x)}$ it holds that $(v,u)\in \E$.

We note that $y=\bigcap_{u\in \Unodes(S_i^*)}\rho_u$ is a hyper-rectangle in $\Re^d$ and since $S_i^*$ is a biclique it holds that $\Vnodes(S_i^*)\subseteq y$.
From the definition of a range tree, let $C(y)$ be the set of $O(\log^d n_V)$ canonical hyper-rectangles of $\rangetree$ that cover all points in $\Vnodes(S_i^*)$. For every hyper-rectangle $z\in C(y)$ let $\xi_z=|((\Vnodes(S_i^*)\cap z)\times \Unodes(S_i^*))\setminus (\bigcup_{j\in[i-1]}\edges(S_j))|$, i.e., the number of uncovered edges among vertices of $\Vnodes(S_i^*)$ in $z$ and vertices of $\Unodes(S_i^*)$. By definition notice that $b_{S_i^*}=\sum_{z\in C(y)}\xi_z$. Let $z'$ be the canonical hyper-rectangle in $C(y)$ such that $\xi_{z'}=\max_{z\in C(y)}\xi_z$. Since there are $O(\log^d n_V)$ canonical hyper-rectangles in $C(y)$, we have that $\xi_{z'}\geq \frac{b_{S_i^*}}{O(\log^d n_V)}$. Furthermore, $|\Vnodes(S_i^*)\cap z'|\leq |\Vnodes(S_i^*)|$. Hence,
    \begin{equation}\label{eq:approx1}
    \frac{|\Vnodes(S_i^*)\cap z'|+|\Unodes(S_i^*)|+c}{\xi_{z'}}\leq O(\log^d n_V)\frac{|\nodes(S_i^*)|+c}{b_{S_i^*}}.
    \end{equation}

    Next, notice that $z'$ is a canonical hyper-rectangle in $\canonical$ that our algorithm visits and computes the subgraph $S^{(z')}$. By definition, 
    \begin{equation}\label{eq:approx2}
        \frac{|\nodes(S_i)|+c}{b_{S_i}}\leq \frac{|\nodes(S^{(z')})|+c}{b_{S^{(z')}}}.
    \end{equation}

Finally, notice that $V_{z'}\supseteq \Vnodes(S_i^*)\cap z'$,  $U_{z'}\supseteq \Unodes(S_i^*)$ and by construction, $E_{z'}$ contains all uncovered edges between the vertices in $V_{z'}$ and $U_{z'}$.
The procedure $\mathsf{DenSub}(G_{z'},c)$ computes the subgraph $S^{(z')}\subseteq G_{z'}$ that approximately maximizes  the ratio $\frac{|\edges_{G_{z'}}(S^{(z')})|}{|\nodes(S^{(z')})|+c}$, up to a constant factor, where $\edges_{G_{z'}}(S^{(z')})$ are the edges in $E_{z'}$ between nodes in $\Vnodes(S^{(z')})$ and $\Unodes(S^{(z')})$. By definition every edge in $E_{z'}$ is an uncovered edge of $\E$, so equivalently $\mathsf{DenSub}(G_{z'},c)$ computes the subgraph $S^{(z')}\subseteq G_{z'}$ that approximately minimizes 
$\frac{|\nodes(S^{(z')})|+c}{|\edges_{G_{z'}}(S^{(z')})|}=\frac{|\nodes(S^{(z')})|+c}{b_{S^{(z')}}}$, up to a constant multiplicative factor.
Hence, by the near-optimality of $\mathsf{DenSub}(G_{z'},c)$ we have that
\begin{equation}\label{eq:approx3}
    \frac{|\nodes(S^{(z')})|+c}{b_{S^{(z')}}}\leq O(1)\cdot\frac{|\Vnodes(S_i^*)\cap z'|+|\Unodes(S_i^*)|+c}{\xi_{z'}}.
\end{equation}

From Equations~\eqref{eq:approx1},~\eqref{eq:approx2}, and~\eqref{eq:approx3}, we get
\begin{align*}
    \frac{|\nodes(S_i)|+c}{b_{S_i}}\leq O(\log^d n_V)\cdot\frac{|\nodes(S_i^*)|+c}{b_{S_i^*}}.
\end{align*}
\vspace{-1em}
\end{proof}

\begin{lemma}
    \label{lem:runtime}
    Our algorithm runs in $O(n_V\cdot(n_U+m)\cdot m\cdot \log^{d-1}(n_V))$ time and uses $O(n_U+m+n_V\cdot\log^{d-1}n_V)$ space.
\end{lemma}
\begin{proof}
    The range tree $\rangetree$ is constructed in $O(n_V\log^{d-1}n_V)$ and there are $O(n_V\log^{d-1}n_V)$ canonical hyper-rectangles. For each $x\in \canonical$ we create the graph $G_x$ in $O(n_V+n_U+m)$ time and the procedure $\mathsf{DenSub}(G_x,c)$ is executed in $O(n_U+m)$ time (Lemma~\ref{lem:denssub}). Hence each iteration of our greedy algorithm runs in $O(n_V\cdot (n_U+ m)\cdot \log^{d-1} (n_V))$ time. In each iteration of the greedy algorithm at least one new edge is covered so the overall time of our algorithm is bounded by $O(n_V\cdot (n_U+ m)\cdot m\cdot \log^{d-1} (n_V))$. Finally, our algorithm constructs the range tree $\rangetree$ with space $O(n_V\log^{d-1}n_V)$. All the other steps of the algorithm including the algorithm for the $c$-densest subgraph problem (Lemma~\ref{lem:denssub}) use linear space with respect to the size of the graph $\G$.
    \end{proof}


Putting everything together we conclude to Theorem~\ref{thm:approx}.

\vspace{-0.5em}
\begin{theorem}
\label{thm:approx}
    Given a \pbg\ $\G(\V,\U,\E)$ in $\ell^d_\infty$, where $d=O(1)$, such that $n_V=|\V|$, $n_U=|\U|$ (and without loss of generality $n_V\leq n_U$), and $|\E|=m$,  there exists an $O(\log(n_U)\cdot\log^{d} (n_V))$-approximation algorithm for the $\ptwo$ (and $\pone$) problem for any $c\geq 0$, that runs in $O(n_V\cdot (n_U+m)\cdot m\cdot \log^{d-1}(n_V))$ time and uses $O(n_U+m+n_V\log^{d-1}n_V)$ space. 
\end{theorem}

\vspace{-0.5em}
\section{Extensions}
\label{sec:ext}
We show how the algorithm from Subsection~\ref{subsec:mainalg} can be extended to intersection bipartite graphs. 
We also describe how to adapt the algorithm to $\delta$-disk bipartite graphs in $\ell_\base^d$ for $\base \geq 1$.

\vspace{-0.5em}
\subsection{Intersection bipartite graphs}
\label{subsec:inbigraph}
Let $\G(\V,\U,\E, \rects_V, \rects_U)$ be an intersection bipartite graph. 
Our algorithm is an extension of the algorithm in Section~\ref{subsec:mainalg}.
The high level idea is that 
for each hyper-rectangle $\rect_v\in \rects_V$ we construct a point $p_v$ in $\Re^{2d}$ and for each $\rect_u\in \rects_U$ a hyper-rectangle $\hat{\rho}_u$ in $\Re^{2d}$, such that $\rect_v\cap\rect_u\neq\emptyset$ if and only if $p_v\in \hat{\rho}_u$.
A hyper-rectangle $\rect$ in $\Re^d$ is defined as $\rect=[\rect^{(1)-}, \rect^{(1)+}]\times\ldots\times [\rect^{(d)-}, \rect^{(d)+}]$, where $[\rect^{(j)-}, \rect^{(j)+}]$ is the interval of $\rect$ in $j$-th dimension.
For every $v\in V$ let $p_v$ be the point in $\Re^{2d}$ defined by the corners of $\rect_v$: 
$p_v=(\rect_v^{(1)-}\!\!,\ldots, \rect_v^{(d)-}\!\!, \rect_v^{(1)+}, \ldots, \rect_v^{(d)+}).$
Let $P=\{p_v\mid v\in \V\}$ and let $\rangetree$ be the range tree constructed on $P$. 
For every $r_u\in \rects_U$, we construct a hyper-rectangle $\hat{\rho}_u=(-\infty,r_u^{(1)+}]\times (-\infty,r_u^{(2)+}]\times\ldots\times (-\infty,r_u^{(d)+}]\times [r_u^{(1)-},\infty)\times [r_u^{(2)-},\infty)\times\ldots\times [r_u^{(d)-},\infty)$.
For every $x\in \canonical$, we define $V_x=\{v\in \V\mid p_v\in x\}$, and $U_x=\{u\in\U\mid x\subseteq \hat{\rho}_u\}$. All the other steps of the algorithm follow verbatim.

The correctness of the method and extensions to interval and $k$-NN bipartite graphs are given in 
Appendix~\ref{appndx:extintersection}.

\vspace{-0.3em}
\begin{theorem}
\label{thm:approxG}
    Given an intersection bipartite graph $\G(\V\!,\!\U\!,\!\E\!,\rects_V\!,\!\rects_U)$ in $\Re^d$, where $d=O(1)$, such that $n_V=|\V|$, $n_U=|\U|$ (and without loss of generality $n_V\leq n_U$), and $|\E|=m$,  there exists an $O(\log(n_U)\cdot\log^{2d} (n_V))$-approximation algorithm for the $\ptwo$ (and $\pone$) problem for any $c\geq 0$, that runs in $O(n_V\cdot (n_U+m)\cdot m\cdot \log^{2d-1}(n_V))$ time and uses $O(n_U+m+n_V\log^{2d-1}n_V)$ space. 
\end{theorem}

\vspace{-1em}
\subsection{$\delta$-disk graph for any $\ell_\base^d$}

We next discuss how the algorithm from Theorem~\ref{thm:approx} extends to $\ell_\base^d$ for any $\base \geq 1$. In the $\ell_\infty^d$ metric, an edge $(v,u)\in\E$ exists if and only if $v$ lies in an axis-aligned box $\rho_u$ of side length $2\delta$ centered at $u$. More generally, for any $\base$, we define the ball $\bigcirc_u=\{p\in\Re^d \mid \dist_\base(u,p)\leq\delta\}$, and observe that $(v,u)\in\E$ if and only if $v\in\bigcirc_u$.
The correctness of our algorithm relies on range trees, which decompose any hyper-rectangle into $O(\log^{d} n_V)$ canonical regions. However, for $\base\neq\infty$, the region $\bigcirc_u$ is not a hyper-rectangle, and therefore range trees cannot be applied.
We show two approaches for handling metrics $\ell_\base^d$ for any  $\base\geq 1$. The first is a theoretical approach based on the notion of \emph{$\eps$-bicliques} that are graph structures that approximate standard bicliques. 
We present the theoretical approach in 
Appendix~\ref{appndx:extellalpha}. 
The second is a practical approach, which we also employ in our experimental evaluation.

\subsubsection{Practical approach}
\label{subsubsec:practical}
The theoretical approach constructs an $\eps$-biclique edge cover rather than an exact biclique cover and relies on an advanced geometric data structure (BBD tree). 
In practice, however, when our goal is to compute a biclique edge cover for general metrics $\ell_\base^d$ with $\base \geq 1$, we instead employ an \emph{$R$-tree}~\cite{Guttman1984} or a \emph{k-d tree}~\cite{bentley1975multidimensional}. The $R$-tree and k-d tree are the most widely used geometric data structures for range aggregation queries, where the query range can be any convex region.
The $R$-tree (or k-d tree) constructed over $\V$ induces a collection $\canonical$ of $O(|\V|)$ canonical axis-aligned hyper-rectangles in $\Re^d$. For any convex region $\rho \subseteq \Re^d$, there exists a subset $C(\rho)\subseteq \canonical$ of $\lambda$ pairwise disjoint hyper-rectangles such that
$\bigcup_{x\in C(\rho)}(\V\cap x)=\V\cap \rho$.
While for range trees and BBD-trees the parameter $\lambda$ is provably bounded, in the case of $R$-trees it can be as large as $O(n_V)$ in the worst case. Nevertheless, in practice $\lambda$ is typically much smaller and is often observed to be $O(\log n_V)$.
We refer to $\lambda$ as the \emph{range cover size}.

For any given $\base \geq 1$, we construct an $R$-tree or k-d tree over $\V$ and execute the algorithm from Subsection~\ref{subsec:mainalg}, iterating over the canonical hyper-rectangles of the $R$-tree (resp. k-d tree) instead of those of the range tree.
Using the same analysis, but assuming a range cover size of $\lambda$ instead of $O(\log^{d} n_V)$, we obtain the following result.

\vspace{-0.2em}
\begin{theorem}
\label{thm:approxpractical}
    Given a \pbg\ $\G(\V,\U,\E)$ in $\ell_\base^d$, where $\base\geq 1$ and $d=O(1)$, such that $n_V=|\V|$, $n_U=|\U|$, and $|\E|=m$,  there exists an $O(\lambda \cdot\log(n_U))$-approximation algorithm for the $\ptwo$ (and $\pone$) problem for any $c\geq 0$, where $\lambda$ is the range cover size, that runs in $O(n_V\cdot (n_U+m)\cdot m)$ time and uses $O(n_U+n_V+m)$ space. 
\end{theorem}
\vspace{-0.5em}
\paragraph{Remark 1}
Due to its practicality in all our experiments we use the algorithm from Theorem~\ref{thm:approxpractical}.
\vspace{-0.5em}

\paragraph{Remark 2} 
The running times stated in our theorems are loose worst-case upper bounds. In practice, performance is significantly better. Specifically, each greedy iteration typically covers many edges (rather than a single edge), the edges in $G_x$ are much smaller than $m$, and only a small subset of canonical hyper-rectangles needs to be examined per iteration. If the greedy algorithm runs for $\kappa_1$ iterations, each $G_x$ contains at most $\kappa_2$ edges, and at most $\kappa_3$ hyper-rectangles are visited per iteration, a more realistic running time is $O(\kappa_3 (n_U+\kappa_2)\kappa_1)$.


%% file: exp2.tex
\vspace{-0.8em}
\section{Experiments}
\label{sec:Newexp}
\vspace{-0.3em}
In this section, we implement the practical algorithm from Theorem~\ref{thm:approxpractical} and compare its effectiveness and efficiency against standard baseline methods for $\pone$ and $\ptwo$ on both real and semi-synthetic datasets. Specifically, for each algorithm we evaluate the following metrics: (i) the total size of the returned biclique edge cover, (ii) the maximum memory consumption, and (iii) the running time required to compute a biclique edge cover.
All datasets and our code can be found in~\cite{ourcode}.

\vspace{-1.1em}
\subsection{Datasets}
\vspace{-0.3em}
We evaluate our algorithms on three real geometric datasets \textbf{Adults}~\cite{adult_2}, \textbf{Default of Credit Card Clients}~\cite{default_of_credit_card_clients_350}, \textbf{MAGIC Gamma Telescope}~\cite{magic_gamma_telescope_159}, two real bipartite graphs \textbf{MovieLens100K}~\cite{grouplens_movielens_100k}, \textbf{MovieLens1M}~\cite{grouplens_movielens_1m}, two matrix datasets \textbf{WorldCities}, \textbf{162bit} from~\cite{davis2011university} corresponding to real bipartite graphs, and one semi-synthetic dataset \textbf{POPSIM}~\cite{nguyen2023popsimindividuallevelpopulationsimulator}.

In the datasets \textbf{Adults}, \textbf{Credit}, \textbf{Gamma}, and \textbf{POPSIM}, each data point lies in $\mathbb{R}^d$ and is associated with a categorical attribute.
In all cases, we normalize the coordinates so that they lie in the range $[0,1]$.
We select a categorical attribute to partition the input points into two disjoint sets, $\V$ and $\U$, thereby defining the bipartition. 
We consider five values of $\delta$ and construct the corresponding $\delta$-disk bipartite graphs. For a fixed value of $\delta$, a vertex $v\in \V$ and a vertex $u\in \U$ are connected by an edge if $\dist_\base(v,u)\leq \delta$. In our experiments, we use the $\ell_2^d$ and $\ell_\infty^d$ metric space ($\base=\{2,\infty\}$).
For \textbf{POPSIM}, we construct two additional geometric bipartite graphs: a large and very sparse graph, called \textbf{POPSIM-S}, in which the average degree is $3.5$, and a large and very dense graph, called \textbf{POPSIM-D}, in which the average degree is greater than $325$. We show all statistics of our datasets in Table~\ref{tab:datasetsell2}.

While \textbf{Adults}, \textbf{Credit}, \textbf{Gamma}, and \textbf{POPSIM} datasets are inherently geometric, this is not the case for the real bipartite graphs \textbf{MovieLens100K} and \textbf{MovieLens1M}. In these datasets, one side of the bipartite graph corresponds to users, the other to movies, and edges represent user ratings. Similarly, the bipartite graphs derived by the matrix datasets \textbf{WorldCities} and \textbf{162bit} are not geometric.
However, our algorithm \ours\ operates on geometric bipartite graphs. Therefore, we embed the original bipartite graphs into a geometric space while attempting to preserve their edge structure with minimal distortion.
More specifically, given a bipartite graph $G(U,V,E)$, we construct an embedding $G'$ in $\mathbb{R}^d$, for a sufficiently large dimension $d$, such that every $u \in U$ and $v \in V$ is mapped to points $p_u,p_v \in \mathbb{R}^d$, and edges are induced by a distance threshold, i.e., $(u,v)\in E \iff \dist_\alpha(p_u,p_v)\le r$.
Our objective is to ensure that $G'$ approximates $G$ as closely as possible, so that the experimental evaluation reflects the performance of \ours\ rather than artifacts introduced by the embedding process.
In practice, we fix $\alpha=2$ (Euclidean distance) and $r=1$, and optimize over the embedding coordinates in dimension $d$ using a gradient-based optimization procedure. We enforce that all original edges of $G$ are preserved in $G'$, while minimizing the number of additional induced edges, ideally introducing none. Our approach is inspired by the geometric graph embedding framework of~\cite{maehara1984space}.

For both \textbf{MovieLens100K} and \textbf{MovieLens1M}, we construct embeddings in $\Re^{256}$ that are highly accurate, introducing only at most $18$ additional edges overall.
For \textbf{WorldCities} and \textbf{162bit}, we construct embeddings in $\mathbb{R}^{256}$ introducing no additional edges.

We report summary statistics (including the number of vertices, number of edges, dimensionality, and grouping attribute) for all datasets (for $\base=2$) and all values of $\delta$ in Table~\ref{tab:datasetsell2}. We also show the statistics for $\base=\infty$ in 
Table~\ref{tab:datasetsellinfty}.
\begin{table}[t]
\centering
\small
\setlength{\tabcolsep}{3pt}
\renewcommand{\arraystretch}{0.9}
\begin{tabular}{|c|c|c|c|c|c|}

\hline
Dataset & $\delta$ & $m$ & $n$ & $d$ & Group Attribute \\
\hline
\multirow{5}{*}{Adults} & $0.05$ & \num{47963} & \multirow{5}{*}{\num{14846}} & \multirow{5}{*}{5} & \multirow{5}{*}{Sex} \\ \cline{2-3}
                     & $0.1$ & \num{445724} &                     &                    & \\ \cline{2-3}
                     & $0.15$ & \num{1415555} &                     &                  &   \\ \cline{2-3}
                     & $0.2$ & \num{3016343} &                     &                   &  \\ \cline{2-3}
                     & $0.25$ & \num{5498540} &                     &                  &   \\ \hline
\multirow{5}{*}{Credit} & $0.002$ & \num{889671} & \multirow{5}{*}{\num{28726}} & \multirow{5}{*}{6} & \multirow{5}{*}{Sex} \\ \cline{2-3}
                     & $0.003$ & \num{1739709} &                     &                 &    \\ \cline{2-3}
                     & $0.004$ & \num{2710136} &                     &                 &    \\ \cline{2-3}
                     & $0.005$ & \num{3750745} &                     &                 &    \\ \cline{2-3}
                     & $0.006$ & \num{4824359} &                     &                 &    \\ \hline
\multirow{5}{*}{Gamma} & $0.13$ & \num{194737} & \multirow{5}{*}{\num{18905}} & \multirow{5}{*}{10} & \multirow{5}{*}{Class} \\ \cline{2-3}
                     & $0.14$ & \num{284626} &                     &                 &    \\ \cline{2-3}
                     & $0.15$ & \num{399339} &                     &                 &    \\ \cline{2-3}
                     & $0.16$ & \num{541050} &                     &                 &    \\ \cline{2-3}
                     & $0.17$ & \num{711158} &                     &                 &    \\ \hline
\multirow{5}{*}{POPSIM} & $0.04$ & \num{1468984} & \multirow{5}{*}{8,591} & \multirow{5}{*}{2} & \multirow{5}{*}{Race} \\ \cline{2-3}
                     & $0.06$ & \num{2660732} &                     &                 &    \\ \cline{2-3}
                     & $0.08$ & \num{3801451} &                     &                 &    \\ \cline{2-3}
                     & $0.1$ & \num{4857704} &                     &                  &   \\ \cline{2-3}
                     & $0.12$ & \num{5844122} &                     &                 &    \\ \hline
POPSIM-S & - & \num{5625395} & \num{3211965} & \num{2} & Race \\ \hline
POPSIM-D & - & \num{122880782} & \num{755182} & \num{2} & Race \\ \hline
MovieLens100K & - & \num{100000} & \num{2700} & \num{256} & - \\ \hline
MovieLens1M & - & \num{1000000} & \num{10000} & \num{256} & - \\ \hline
    WorldCities & - & \num{7518} & \num{415} & \num{256} & -\\ \hline
162bit & - & \num{37118} & \num{7203} & \num{256} & -\\ \hline
\end{tabular}
\caption{\label{tab:datasetsell2}Statistics of $\delta$-disk graphs in our experiments for $\ell_2^d$.}
\vspace{-1em}
\end{table}

\vspace{-1.3em}
\subsection{Algorithms}
\vspace{-0.2em}
We describe the algorithms compared in our experiments.

\vspace{-0.5em}
\paragraph{Our algorithm: \ours}
We implemented our practical algorithm (with provable guarantees) from Theorem~\ref{thm:approxpractical}. To generate the canonical hyper-rectangles we use a modified implementation of SciPy's \texttt{KDTree}.
For the $\pone$ problem, we use Charikar's $2$-approximation algorithm~\cite{charikar2000greedy} to compute the densest subgraph, as implemented in~\cite{ambavi2020densest}.
For the $\ptwo$ problem, we modified the implementation of the exact algorithm~\cite{goldberg1984finding} for the densest subgraph problem, as shown in Section~\ref{sec:algs} to solve the $c$-densest subgraph.

\paragraph{Baselines}
For each $\pone$ and $\ptwo$ problem, we compare our approach against three baselines.
The first two baselines share the same high level approach: They are greedy algorithms and run in two phases. In the first phase they generate a candidate set of bicliques. In the second phase they run a greedy algorithm to compute a biclique edge cover using the selected candidate bicliques.
The third baseline is the CPGR algorithm from~\cite{chavan2026speeding}.


$\bullet$ \textbf{\baseone}: This is a heuristic for the edge-RMP problem in role-mining~\cite{role-mining-survey,edge-RMP-hardness}, which is equivalent to $\pone$ for bipartite graphs. It uses the RoleMiner algorithm~\cite{10.1145/1180405.1180424} to generate a candidate set of bicliques. 

$\bullet$ \textbf{\basetwo}: 
   It constructs a candidate set of bicliques by enumerating all maximal bicliques using AMBEA~\cite{10633882}.

$\bullet$ \textbf{\basethree}: This is the state-of-the-art algorithm for constructing a biclique partition, as implemented in~\cite{chavan2026speeding}.

\paragraph{Details of \baseone\ and \basetwo}
Following candidate generation, the first two baselines employ the same greedy selection criteria.
For the $\pone$ problem, the greedy procedure chooses in each iteration the biclique that minimizes the ratio of the number of nodes over the number of uncovered edges. For the $\ptwo$ problem, the greedy procedure chooses the biclique that minimizes the ratio of the number of nodes plus $c$ over the number of uncovered edges.

We initially observed that, in all cases, \baseone\ and \basetwo\ run significantly slower than our algorithm and frequently exceed available memory limits. This behavior is primarily due to the extremely large number of candidate bicliques generated by the baselines. 
This behavior is consistent with prior observations. For example, in~\cite{edge-RMP-hardness}, the authors generated synthetic datasets with fewer than $400$ vertices and reported more than $6{,}000$ candidate bicliques, with running times exceeding $60$ seconds. In our setting, the effect is even more pronounced: On the \textbf{Adults} dataset in $\ell_\infty^d$ with $\delta=0.1$, \baseone\ requires $5{,}473$ seconds, generates \num{1130246} candidate bicliques, and reaches a peak memory usage of \num{42068}~MB. 

To ensure a fair and feasible comparison, we introduce an early-stopping criterion for \baseone\ and \basetwo. At a high level, we constrain all algorithms to operate within comparable time budgets. Specifically, we allow baselines to spend between $10\%$ and $35\%$ of our algorithm’s total execution time generating candidate bicliques.
Under this setting, both baselines still require slightly more total running time than our algorithm. Consequently, our evaluation grants the baselines an advantage by allowing them to run longer and generate more candidate bicliques, while remaining within feasible time and memory limits.
Additionally, the consecutively generated candidates by both of our baselines tend to share many vertices. Thus, to ensure the quality of candidates, we also randomly shuffle the vertex order before generating them.

\paragraph{Details of \basethree}
Across all datasets we evaluated, we observed that \basethree\ produced representations that were substantially larger than those produced by \ours. This behavior is expected, since biclique partitioning is substantially more restrictive than biclique covering (every biclique partition is a biclique cover, but not vice versa), and the objective in~\cite{chavan2026speeding} is not to compute a minimum-size biclique partition. Consequently, we omitted this baseline from most of our figures, although we still report some examples of the sizes of its computed representations.

\paragraph{Setup}
All experiments were conducted on a Windows~11 workstation equipped with an Intel Core i7-9700 CPU and 32~GB of DDR4 RAM. All algorithms were implemented in Python~3.12.9. The AMBEA algorithm~\cite{10633882} used in \basetwo, however, is a parallel implementation written in C++. It was compiled as a standalone executable and invoked via Python subprocesses. As a result, \basetwo\ enjoys an additional performance advantage compared to both \baseone\ and \ours.
Memory usage was monitored using the Python library psrecord, with a 3-second polling interval to minimize measurement overhead.

\begin{figure*}
\centering
    \begin{minipage}[t]{\textwidth}
        \begin{minipage}[t]{\textwidth}
            \centering
            \begin{minipage}[t]{0.24\linewidth}\centering \textbf{Adults}\end{minipage}
            \begin{minipage}[t]{0.24\linewidth}\centering \textbf{Credits}\end{minipage}
            \begin{minipage}[t]{0.24\linewidth}\centering \textbf{Gamma}\end{minipage}
            \begin{minipage}[t]{0.24\linewidth}\centering \textbf{POPSIM}\end{minipage}
        \end{minipage}
        \begin{minipage}[t]{\textwidth}
            \centering
            \begin{minipage}[t]{0.24\linewidth}
                \centering
                \includegraphics[width=\textwidth]{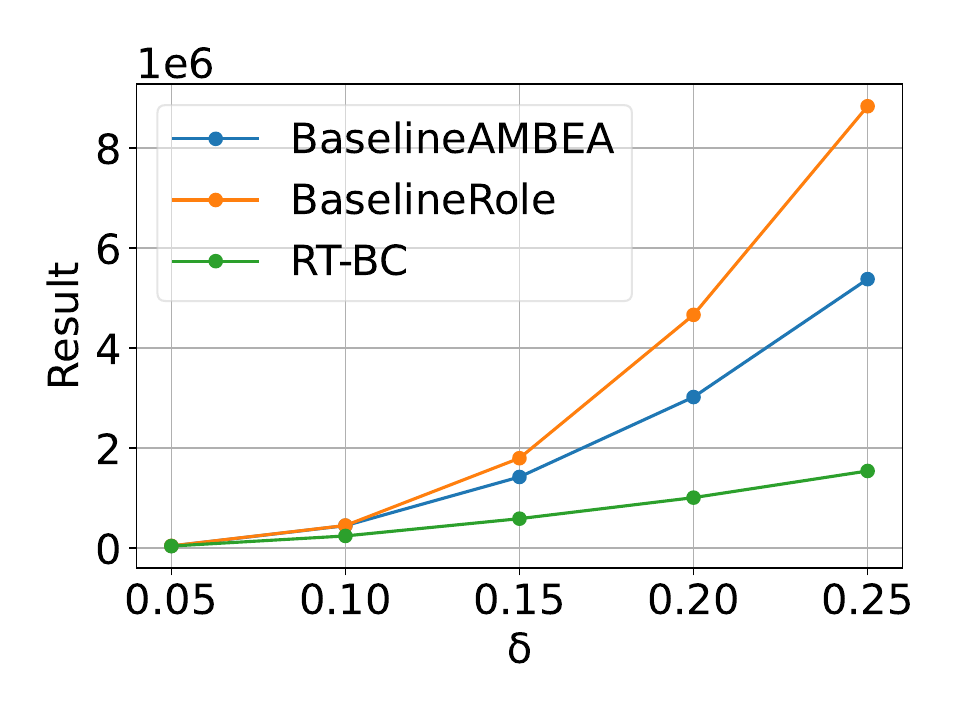}
                \vspace{-2.5em}
            \end{minipage}
            \begin{minipage}[t]{0.24\linewidth}
                \centering
                \includegraphics[width=\textwidth]{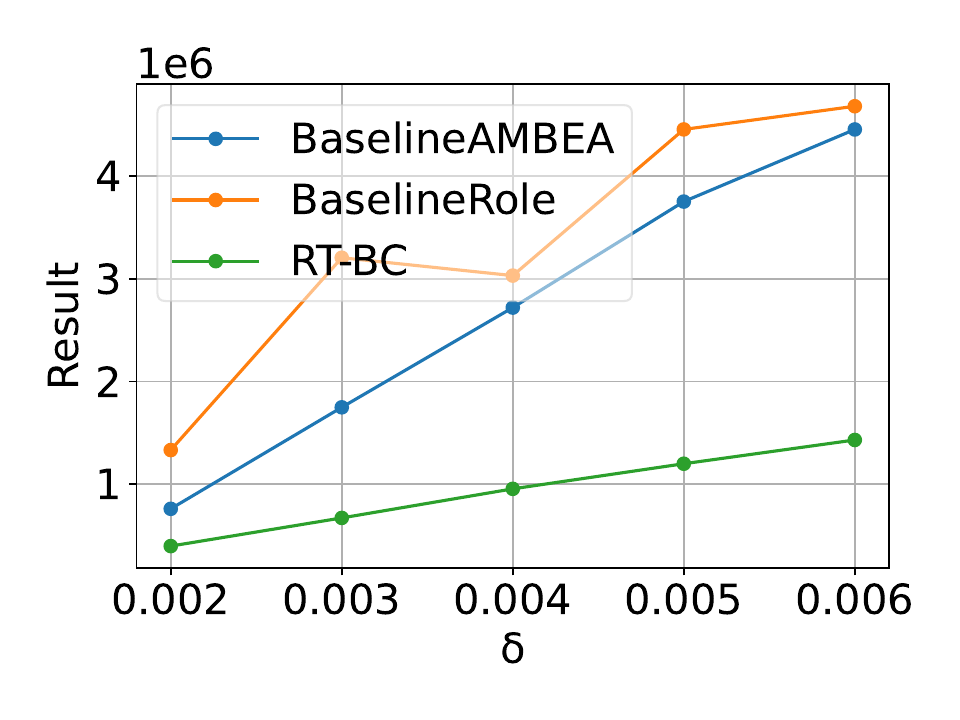}
                \vspace{-2.5em}
            \end{minipage}
            \begin{minipage}[t]{0.24\linewidth}
                \centering
                \includegraphics[width=\textwidth]{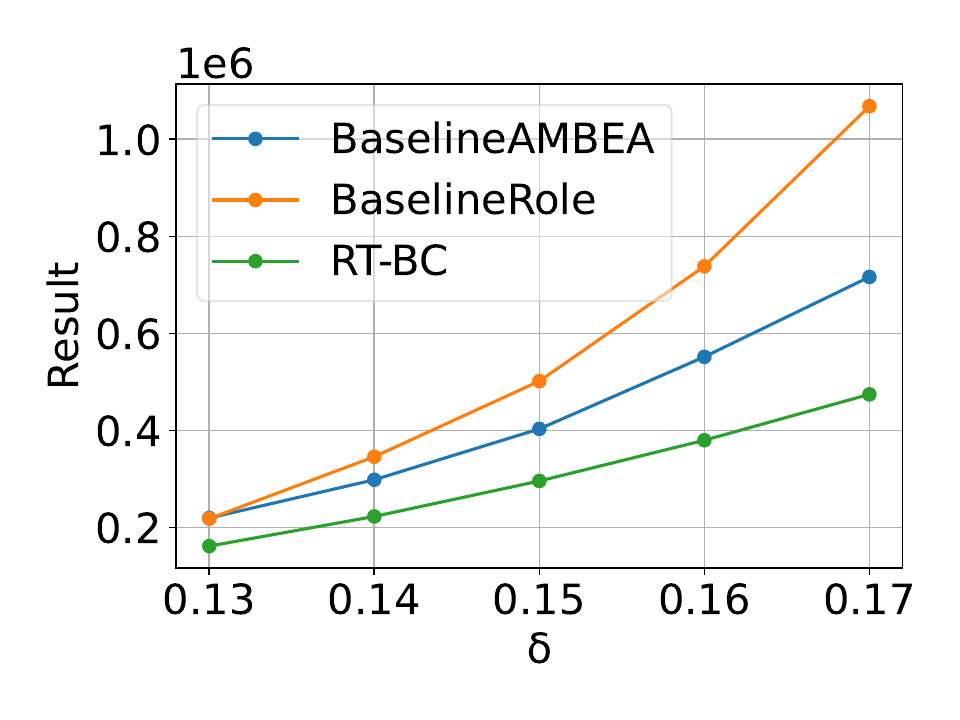}
                \vspace{-2.5em}
            \end{minipage}
            \begin{minipage}[t]{0.24\linewidth}
                \centering
                \includegraphics[width=\textwidth]{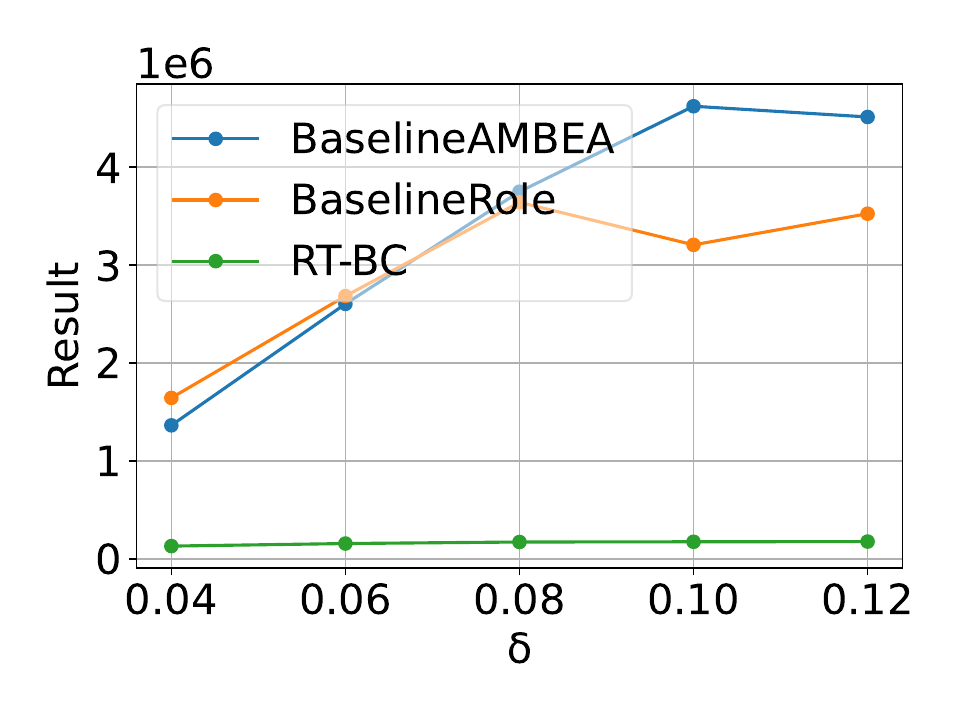}
                \vspace{-2.5em}
            \end{minipage}
            \vspace{-1em}
            \caption*{(a) Size of representation}
        \end{minipage}
        \begin{minipage}[t]{\textwidth}
            \centering
            \begin{minipage}[t]{0.24\linewidth}
                \centering
                \includegraphics[width=\textwidth]{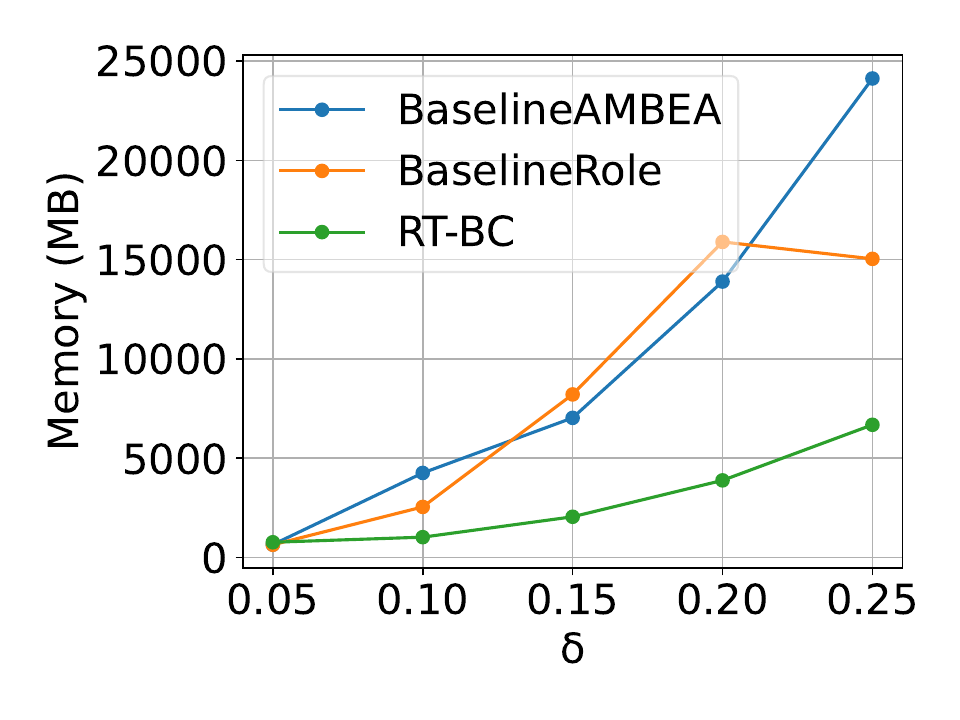}
                \vspace{-2.5em}
            \end{minipage}
            \begin{minipage}[t]{0.24\linewidth}
                \centering
                \includegraphics[width=\textwidth]{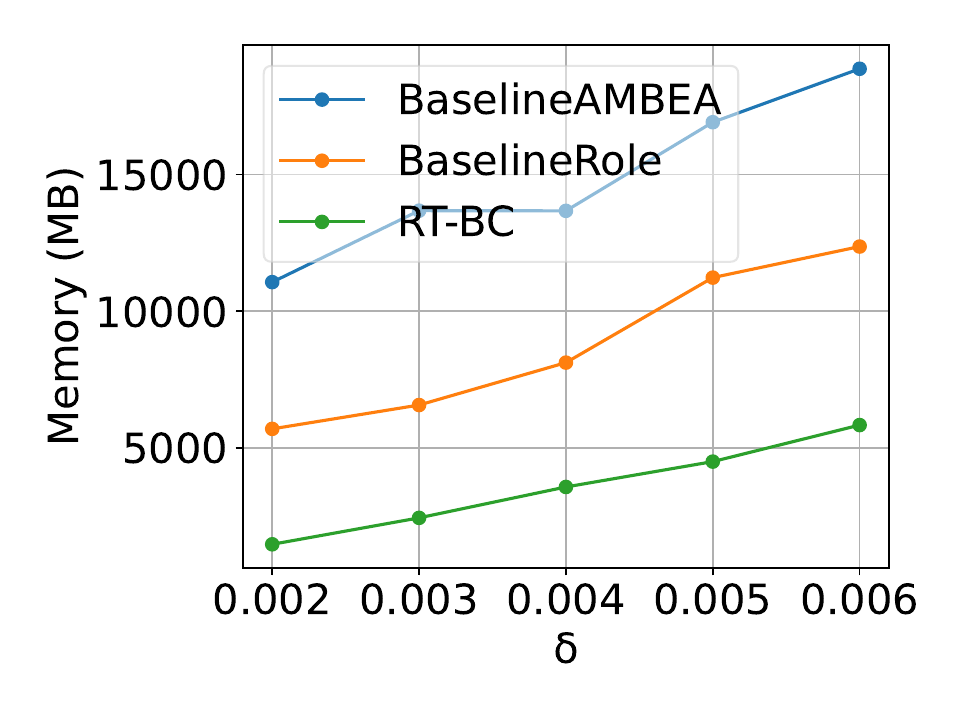}
                \vspace{-2.5em}
            \end{minipage}
            \begin{minipage}[t]{0.24\linewidth}
                \centering
                \includegraphics[width=\textwidth]{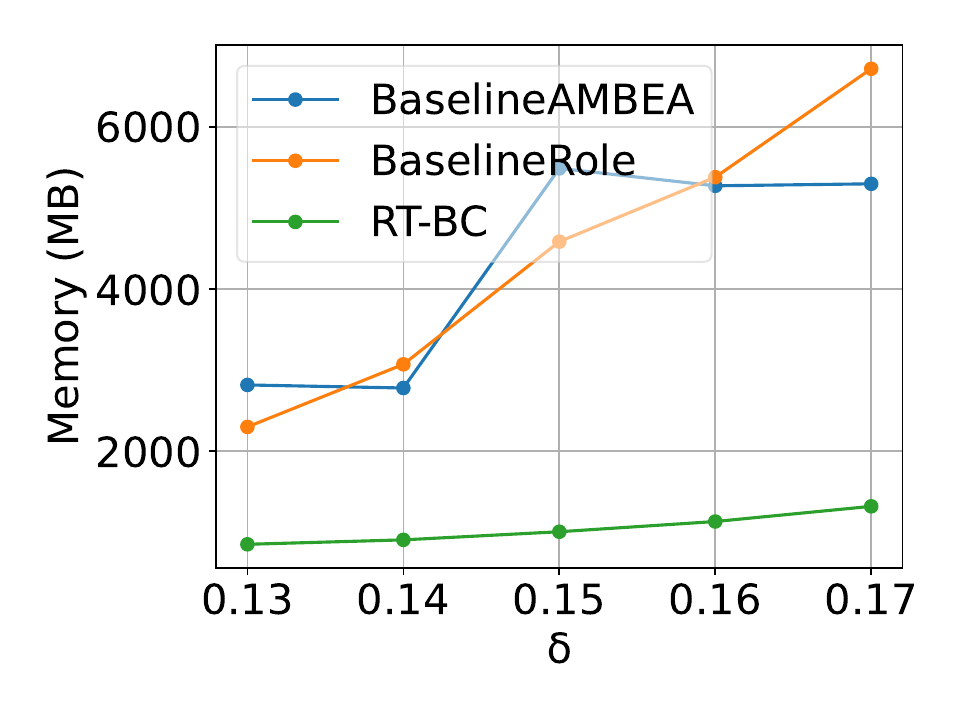}
                \vspace{-2.5em}
            \end{minipage}
            \begin{minipage}[t]{0.24\linewidth}
                \centering
                \includegraphics[width=\textwidth]{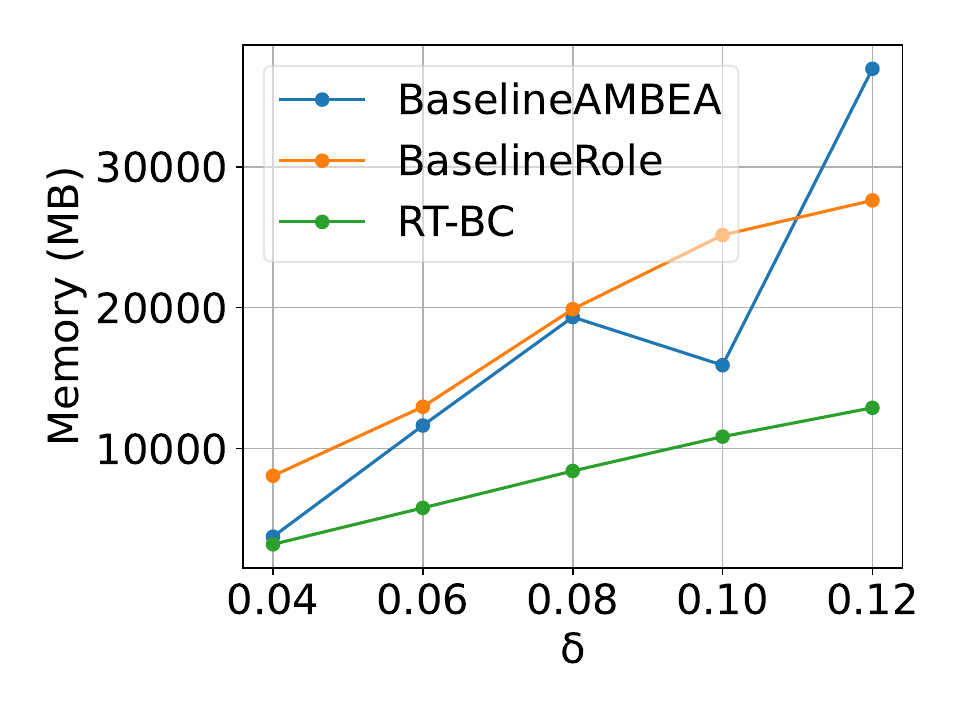}
                \vspace{-2.5em}
            \end{minipage}
            \vspace{-1em}
            \caption*{(b) Peak memory in MB}
        \end{minipage}
        \begin{minipage}[t]{\textwidth}
            \centering
            \begin{minipage}[t]{0.24\linewidth}
                \centering
                \includegraphics[width=\textwidth]{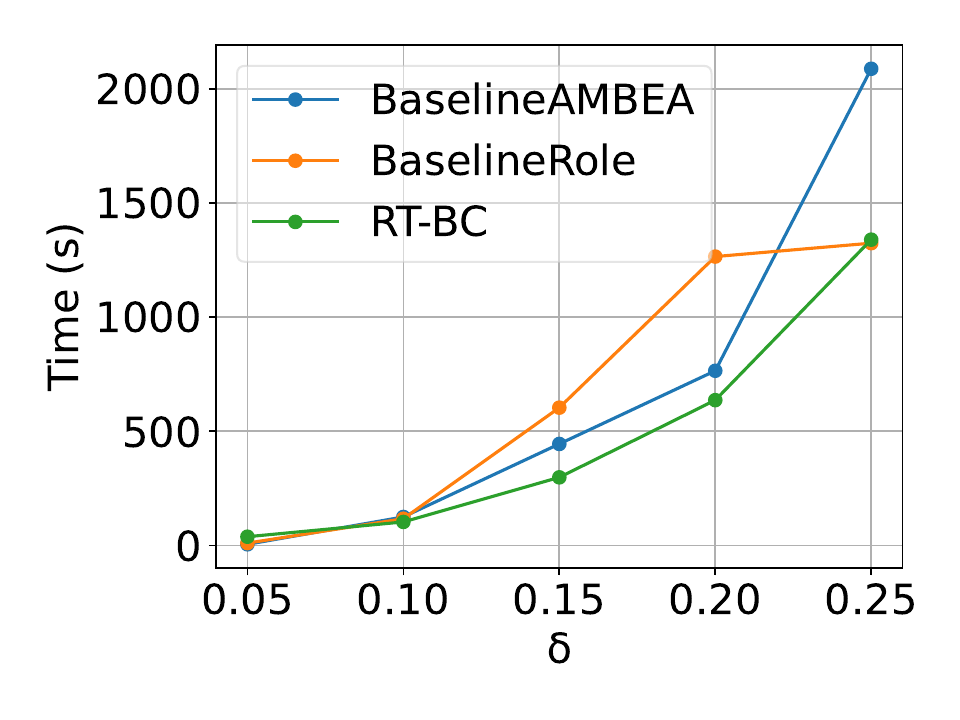}
                \vspace{-2.5em}
            \end{minipage}
            \begin{minipage}[t]{0.24\linewidth}
                \centering
                \includegraphics[width=\textwidth]{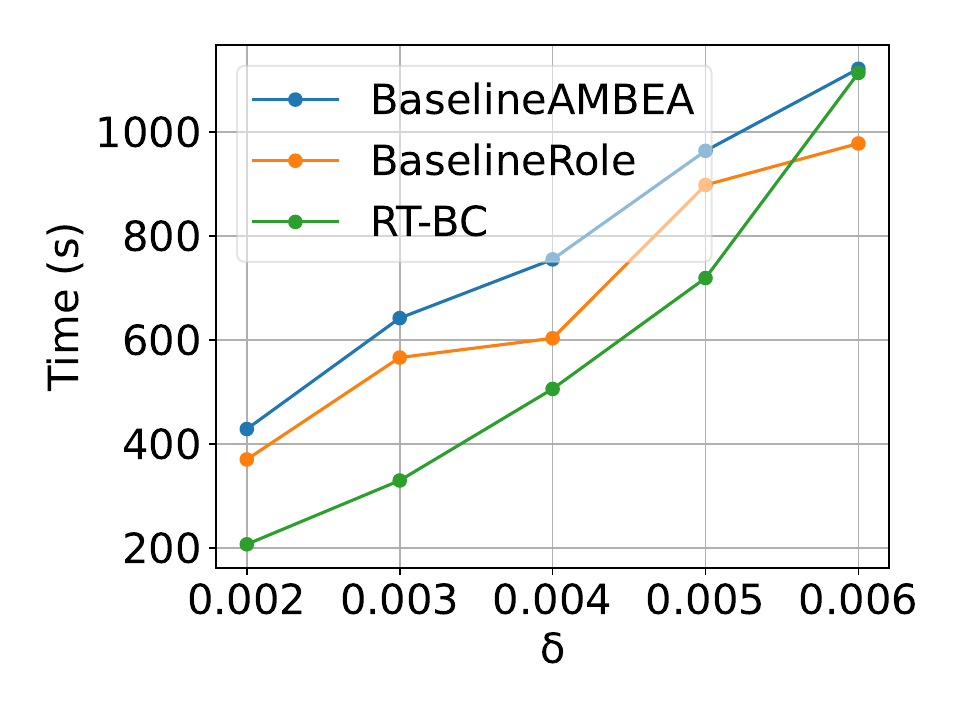}
                \vspace{-2.5em}
            \end{minipage}
            \begin{minipage}[t]{0.24\linewidth}
                \centering
                \includegraphics[width=\textwidth]{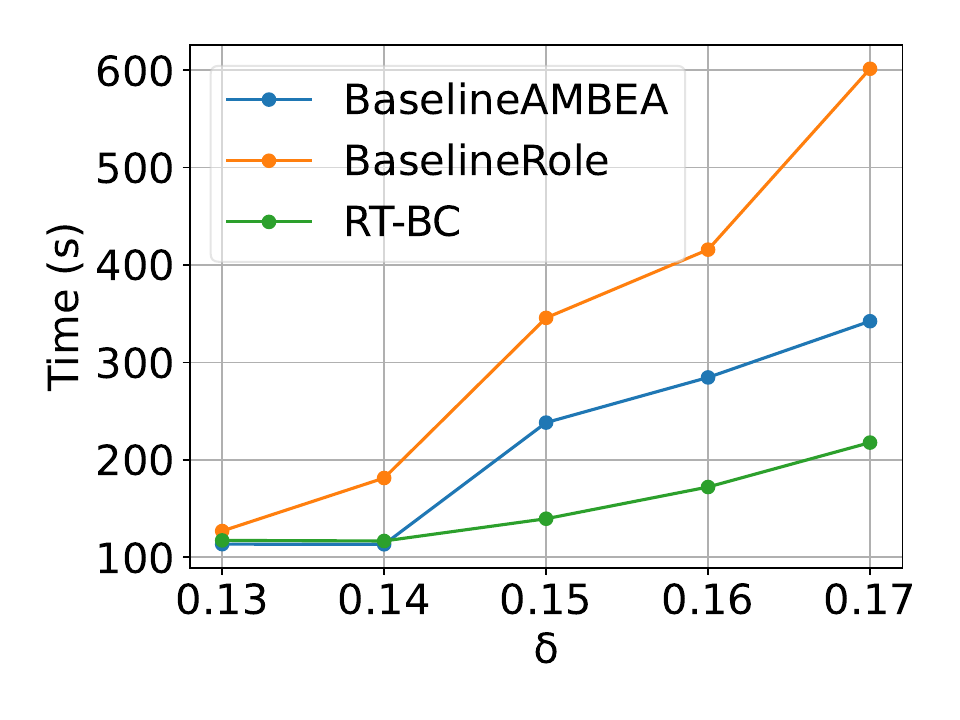}
                \vspace{-2.5em}
            \end{minipage}
            \begin{minipage}[t]{0.24\linewidth}
                \centering
                \includegraphics[width=\textwidth]{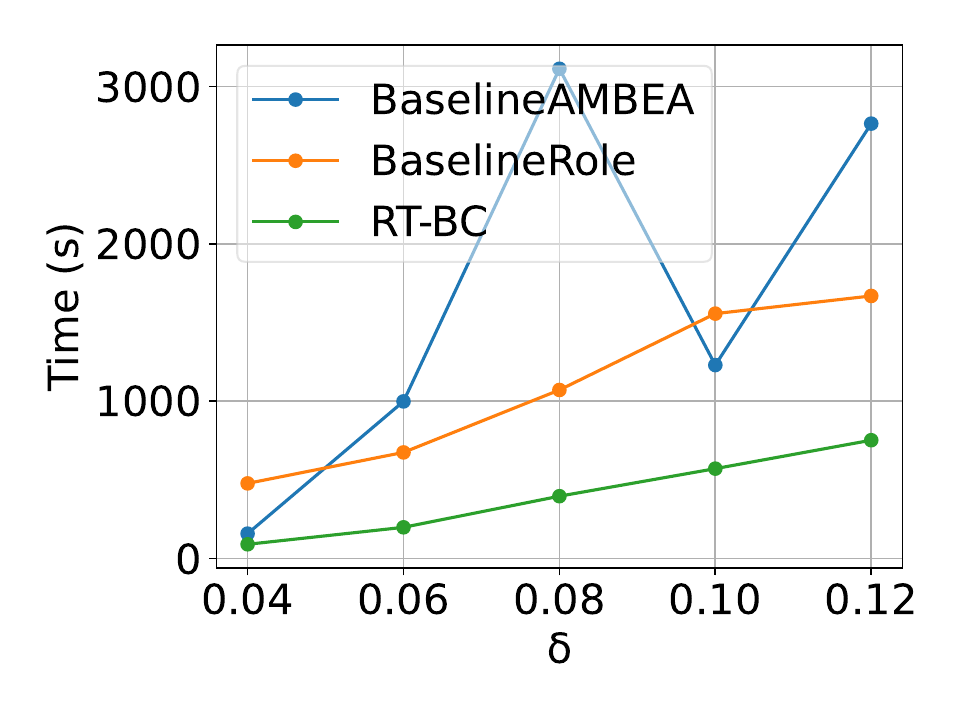}
                \vspace{-2.5em}
            \end{minipage}
            \vspace{-0.5em}
            \caption*{(c) Running time in seconds}
        \end{minipage}
    \caption{Experimental results for $\pone$ in $\ell_2^d$ across all datasets}
    \label{fig:wbcl2}
    \end{minipage}
    \vspace{0.5em}
    \begin{minipage}[t]{\textwidth}
        \begin{minipage}[t]{\textwidth}
            \centering
            \begin{minipage}[t]{0.24\linewidth}\centering \textbf{Adults}\end{minipage}
            \begin{minipage}[t]{0.24\linewidth}\centering \textbf{Credits}\end{minipage}
            \begin{minipage}[t]{0.24\linewidth}\centering \textbf{Gamma}\end{minipage}
            \begin{minipage}[t]{0.24\linewidth}\centering \textbf{POPSIM}\end{minipage}
        \end{minipage}
        \begin{minipage}[t]{\textwidth}
        \centering
        \begin{minipage}[t]{0.24\linewidth}
            \centering
            \includegraphics[width=0.975\textwidth, height=2.88cm]{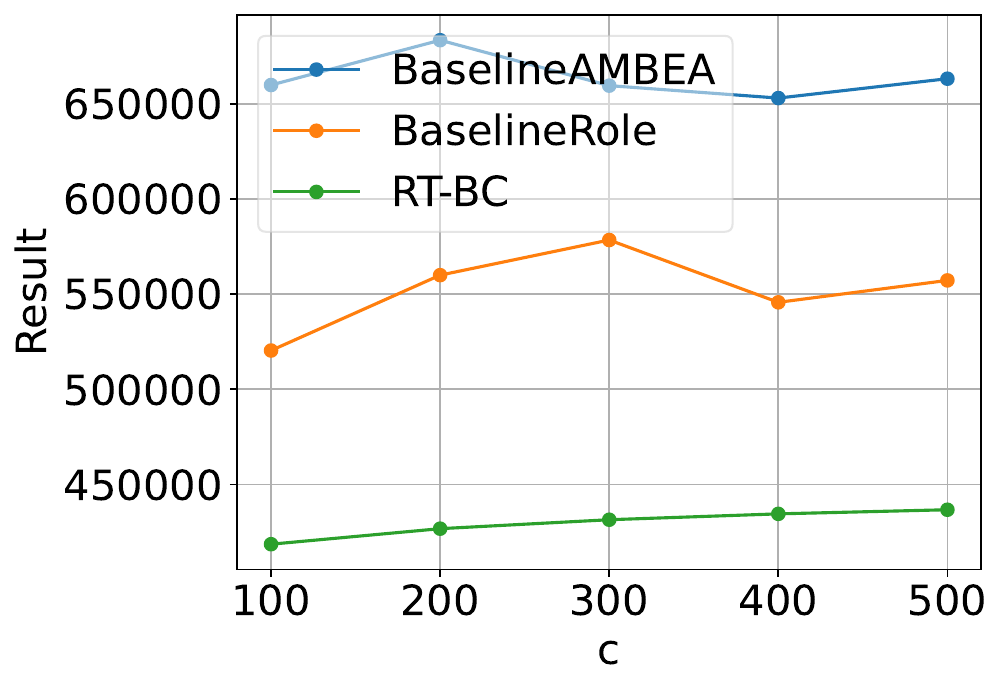}
            \vspace{-2.5em}
        \end{minipage}
        \begin{minipage}[t]{0.24\linewidth}
            \centering
            \includegraphics[width=\textwidth]{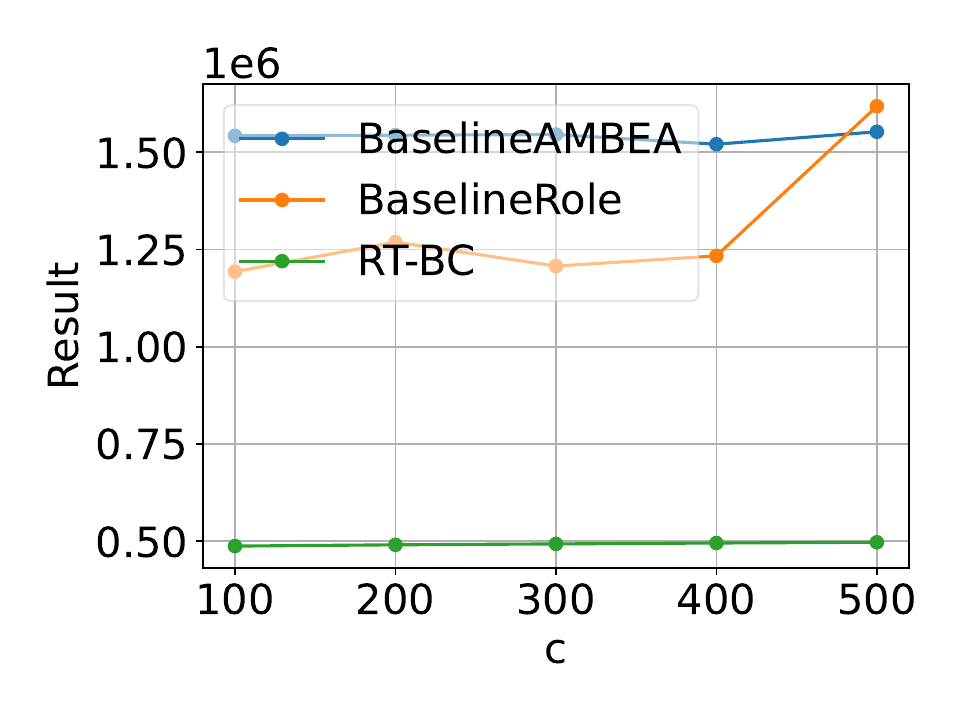}
            \vspace{-2.5em}
        \end{minipage}
        \begin{minipage}[t]{0.24\linewidth}
            \centering
            \includegraphics[width=\textwidth]{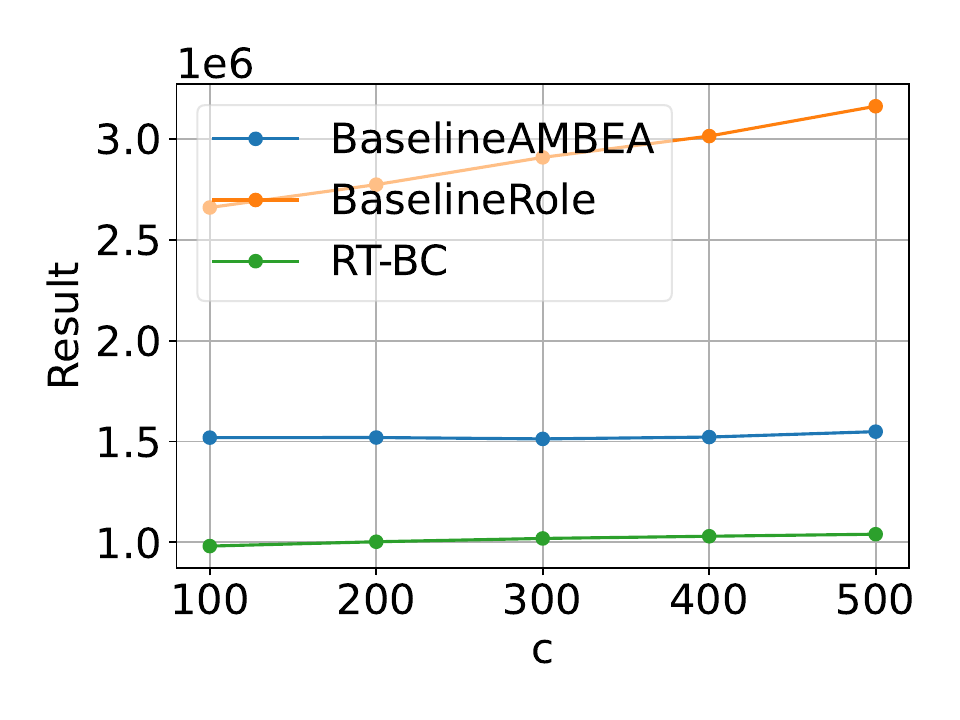}
            \vspace{-2.5em}
        \end{minipage}
        \begin{minipage}[t]{0.24\linewidth}
            \centering
            \includegraphics[width=\textwidth]{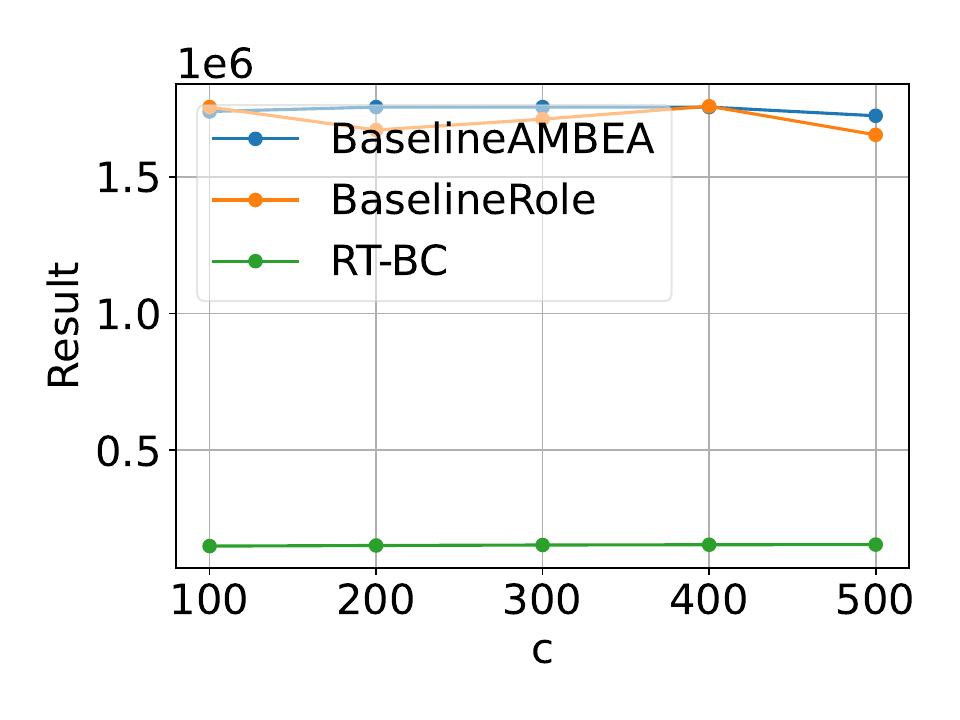}
            \vspace{-2.5em}
        \end{minipage}
        \vspace{-1em}
        \caption*{(a) Size of representation}
    \end{minipage}
        \begin{minipage}[t]{\textwidth}
            \centering
            \begin{minipage}[t]{0.24\linewidth}
                \centering
                \includegraphics[width=\textwidth]{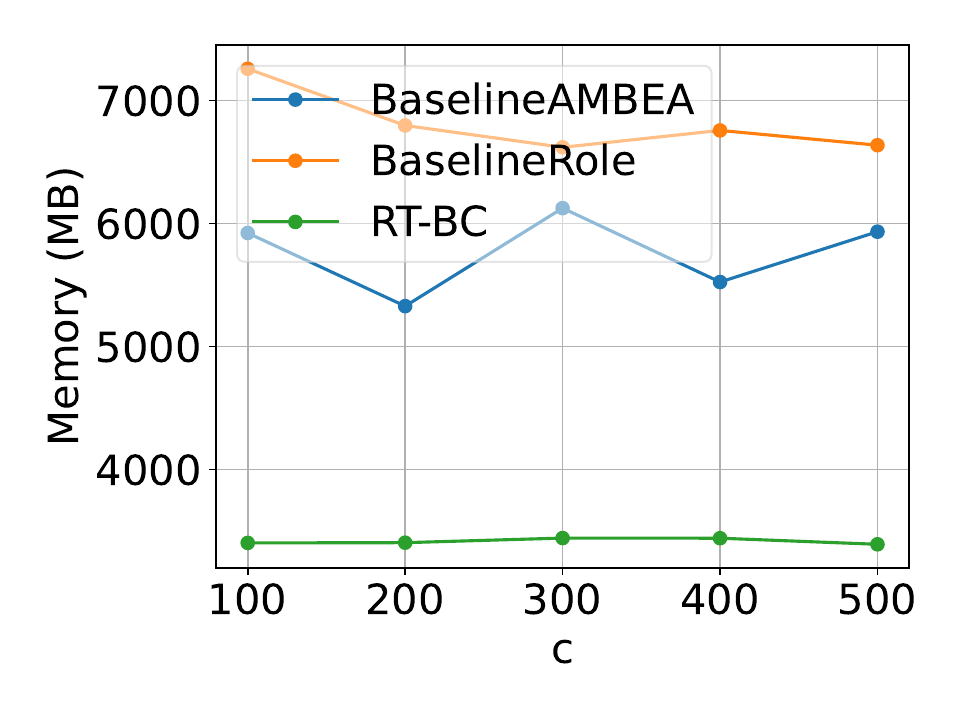}
                \vspace{-2.5em}
            \end{minipage}
            \begin{minipage}[t]{0.24\linewidth}
                \centering
                \includegraphics[width=\textwidth]{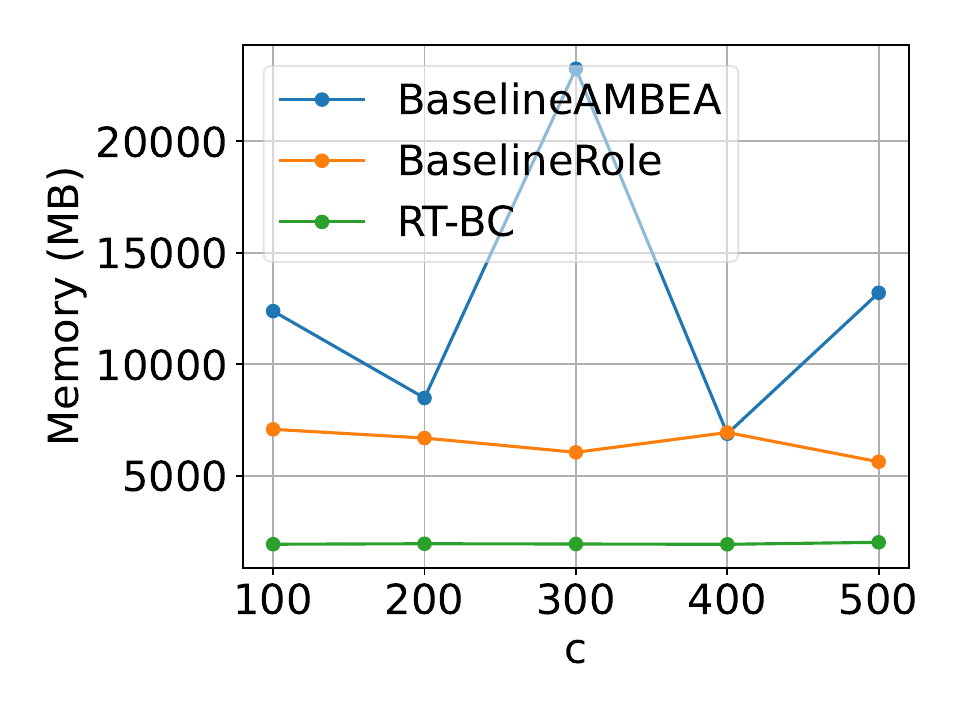}
                \vspace{-2.5em}
            \end{minipage}
            \begin{minipage}[t]{0.24\linewidth}
                \centering
                \includegraphics[width=\textwidth]{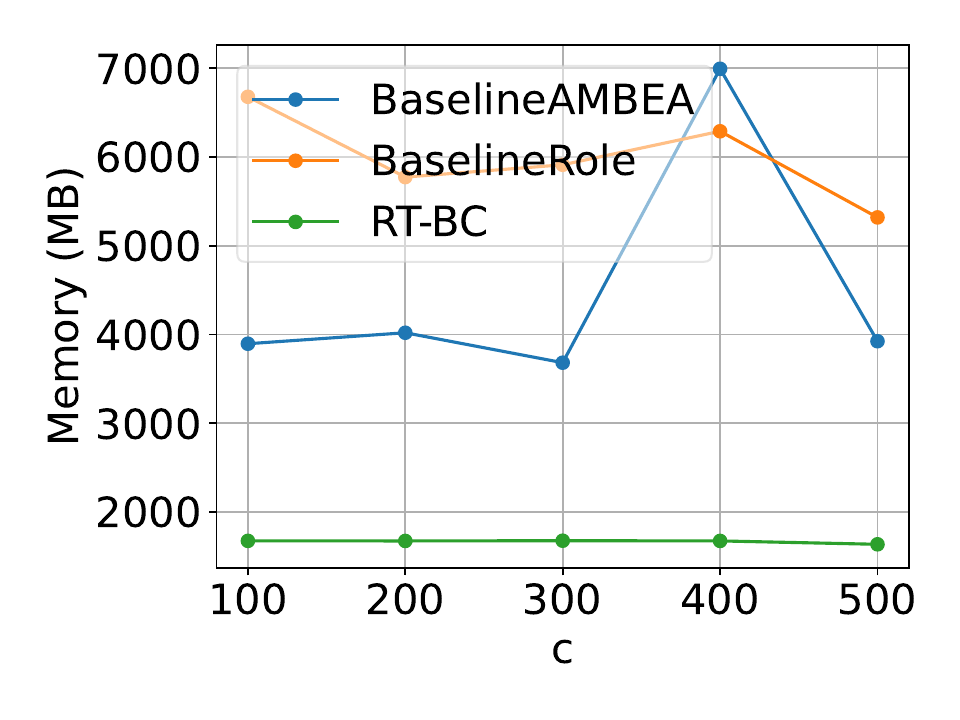}
                \vspace{-2.5em}
            \end{minipage}
            \begin{minipage}[t]{0.24\linewidth}
                \centering
                \includegraphics[width=\textwidth]{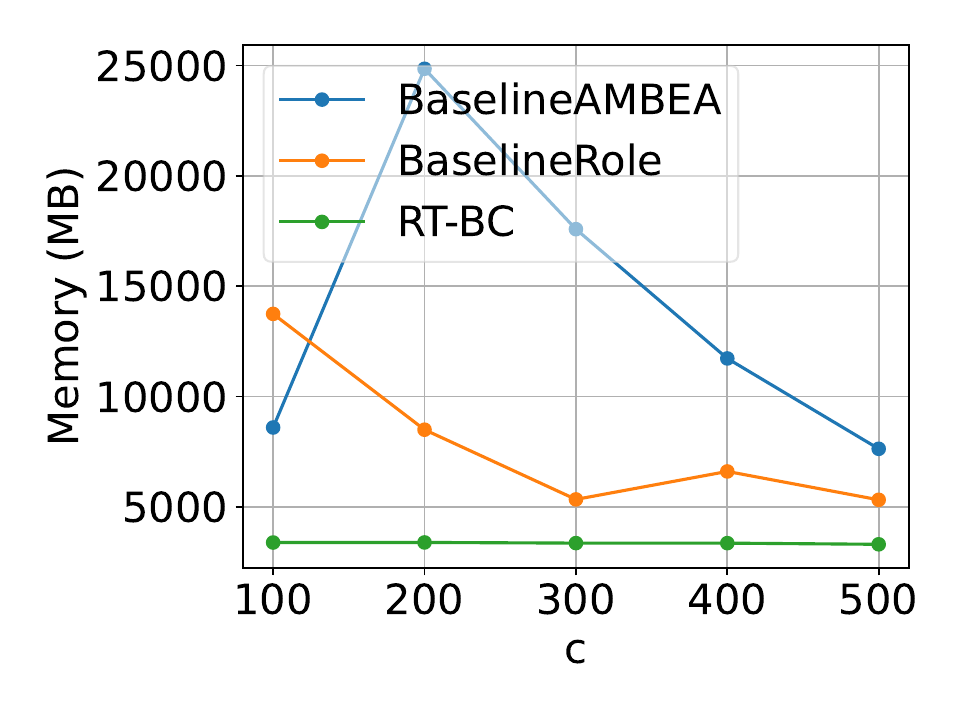}
                \vspace{-2.5em}
            \end{minipage}
            \vspace{-1em}
            \caption*{(b) Peak memory in MB}
        \end{minipage}
        \caption{Experimental results for $\ptwo$ in $\ell_\infty^d$ across all datasets}
        \label{fig:gwbclinf}
    \end{minipage}
\end{figure*}

\begin{table*}[t]
\centering
\resizebox{\textwidth}{!}{%
\begin{tabular}{|c|c|c|c|c|c|c|c|c|}
\hline
Baseline/Dataset
& Adults $\delta=0.2$
& Gamma $\delta=0.16$
& POPSIM-S
& POPSIM-D
& MovieLens100K
& MovieLens1M
& WorldCities
&  162bit\\
\hline

\baseone
& \num{4668513} & \num{738232} & \num{7038915} & \num{108138671} & \num{188227} & \num{2336652} & \num{13429} & \num{52577}\\
\hline

\basetwo
& \num{3024324} & \num{551686} & \num{5387872} & \num{122279423} & \num{108014} & \num{1079149} & \num{7427} & \num{51368}\\
\hline

\basethree
& \num{5242099} & \num{1082100} & \num{11250790} & \num{245761564} & \num{200000} & \num{1908017} & \num{15036} & \num{74236}\\
\hline

Original
& \num{6032686} & \num{1082100} & \num{11250790} & \num{245761564} & \num{200000} & \num{2000000} & \num{15036} & \num{74236}\\
\hline

\ours
& \textbf{\num{1012402}} & \textbf{\num{379675}} & \textbf{\num{3396491}} & \textbf{\num{904425}} & \textbf{\num{101152}} & \textbf{\num{952869}} & \textbf{\num{6034}} & \textbf{\num{40590}}\\
\hline

\end{tabular}
}
\caption{Representation size produced by each baseline for $\pone$ on each dataset.}
\vspace{-1em}
\label{tab:newresults}
\end{table*}

\vspace{-1em}
\subsection{Experimental results for $\pone$} \label{exp:results}
\vspace{-0.3em}
In this section, we first provide a comprehensive comparison of our algorithm (\ours) against \baseone\ and \basetwo\ on all versions of \textbf{Adults}, \textbf{Credit}, \textbf{Gamma}, and \textbf{POPSIM}, for different values of $\delta$, as reported in Table~\ref{tab:datasetsell2}. For each setting, we report the size of the returned representation, the peak memory usage during execution, and the running time of each algorithm. We omit \basethree\ from these experiments because, in all cases, it returns representations that are substantially larger than those of the other baselines. We discuss the performance of \basethree\ separately later.

Experiments are conducted under both $\ell_2$ and $\ell_\infty$ metrics. Since the observed trends are consistent across the two settings, we present the results for the $\ell_2$ metric here and defer the $\ell_\infty$ results to 
Appendix~\ref{appndx:exp}.
As discussed above, to ensure a fair comparison, we apply an early-stopping criterion that allows the baselines to run slightly longer than our algorithm, as illustrated in Figure~\ref{fig:wbcl2}(c). 
This enables a direct comparison of the quality of the resulting representations, as well as the peak memory consumption, under comparable time budgets.

As shown in Figure~\ref{fig:wbcl2}(a), our algorithm \ours\ consistently outperforms both \baseone\ and \basetwo\ in terms of representation size across all datasets and for all values of $\delta$. The only exception occurs for the \textbf{Adults} dataset with $\delta=0.05$, where all algorithms produce representations of comparable size. In this instance, the corresponding bipartite graph is sparse, with a low average degree leaving limited room for improvement.
As the graphs become larger and denser, however, the advantage of \ours\ becomes increasingly pronounced. For example, on real-world datasets, \ours\ produces representations that are up to four times smaller than those returned by both baselines (e.g., on the \textbf{Credit} dataset). The improvement is even more substantial on the semi-synthetic \textbf{POPSIM} dataset, where the representation produced by \ours\ is approximately $\mathbf{26}$ times smaller than those produced by the baselines. 
This behavior can be attributed to the low intrinsic dimensionality of the \textbf{POPSIM} dataset and the highly geometric manner in which it is generated. The tree data structure employed by \ours\
performs particularly well in low-dimensional settings. By exploring the canonical hyper-rectangles, the algorithm is able to identify regions $x$ whose corresponding subgraphs $G_x$ capture near-optimal bicliques, leading to significantly more compact representations.

Next, in Figure~\ref{fig:wbcl2}(b), we compare the peak memory usage across all datasets. \ours\ consistently uses significantly less memory than both baselines. In some cases, our method requires up to three times less memory than the baselines (e.g., \textbf{Gamma} dataset).
These results are consistent with our theoretical analysis. As shown in Theorem~\ref{thm:approxpractical}, our algorithm uses asymptotically only linear space with respect to the size of the input graph. In contrast, the memory consumption of the baselines is dominated by the total size of the candidate bicliques they generate, which is typically very large.

Finally, \baseone\ and \basetwo\ exhibit similar behavior, producing biclique edge covers of comparable size. In some cases, \basetwo\ outperforms \baseone\ due to its parallel C++ implementation of the candidate biclique generation step, which allows it to explore more candidates within the same time budget.

\vspace{-0.3em}
\paragraph{Comprehensive comparison of all baselines on small, large, sparse, and dense bipartite graphs}
Table~\ref{tab:newresults} reports the representation size returned by \ours\ and all baselines (\baseone, \basetwo, and \basethree) across a diverse collection of bipartite graphs; for reference, we also report the size of the original edge-based representation, obtained without applying any compression algorithm. To ensure a fair comparison, we configure all baselines so that their running times are roughly comparable to that of \ours. The datasets in this comparison consist of two representative geometric bipartite graphs, \textbf{Adults} and \textbf{Gamma}; two variants of \textbf{POPSIM}, capturing both sparse and very large dense instances; and four real, non-geometric bipartite graphs, namely \textbf{MovieLens100K}, \textbf{MovieLens1M}, \textbf{WorldCities}, and \textbf{162bit}, after first embedding them into a geometric space. Across all real bipartite graphs, \ours\ consistently produces smaller representations than all compared baselines. Finally, \basethree\ often returns a representation whose size is close to, or even the same as, that of the original graph, suggesting that biclique partitioning is not an effective approach for $\pone$ on these instances. There are two main reasons for this. First, as discussed above, constructing a biclique partition is substantially more restrictive than constructing a biclique cover. Second, the implementation of \basethree\ from~\cite{chavan2026speeding} is restricted to graphs satisfying $m \geq 2\cdot (\max\{|\V|,|\U|\})^{3/2}$, a condition that does not hold for many of the graphs in our evaluation.

\vspace{-1em}
\subsection{Experimental results for $\ptwo$} \label{exp:Gresults}

In Figure~\ref{fig:gwbclinf}, we report the representation size and peak memory usage of our algorithm and the baselines for the $\ptwo$ problem under the $\ell_\infty$ metric, for varying values of $c$ and a fixed value of $\delta$. Specifically, we use $\delta=0.1$ for \textbf{Adult}, $\delta=0.002$ for \textbf{Credit}, $\delta=0.13$ for \textbf{Gamma}, and $\delta=0.04$ for \textbf{POPSIM}. Fixing $\delta$ isolates the effect of $c$ on the quality and memory footprint of the resulting representations.
As in the $\pone$ experiments, all algorithms are evaluated under comparable time budgets to ensure a fair comparison. Since the running-time trends are consistent with those observed for $\pone$, we defer the detailed timing results to 
Appendix~\ref{appndx:exp}.

Consistent with our earlier observations, our algorithm \ours\ outperforms both baselines across all evaluation metrics, consistently producing smaller representations while using less memory. Notably, the performance gap between our method and the baselines becomes even more pronounced in the $\ptwo$ setting. For example, on the \textbf{Gamma} dataset with $\delta=0.13$, under the $\pone$ objective \ours\ produces a representation that is 2-times smaller than those returned by the baselines. In contrast, for the $\ptwo$ problem under the same value of $\delta$ and with $c=500$, \ours\ returns a representation that is 3-times smaller than the baselines. 
Similar trends are observed for the remaining bipartite graphs, including MovieLens and the two matrix datasets; therefore, we omit them from the figures.

We observe occasional memory spikes for $\basetwo$. This behavior is due to the combination of early stopping and the randomized vertex ordering used at each run to avoid repeatedly exploring the same regions of the graph.

%% file: relwork.tex
\vspace{-1em}
\section{Related work}
\label{sec:rel}
\vspace{-0.4em}

Lossless compact representations for graphs have been studied extensively~\cite{bannach2024faster,buehrer2008scalable,dhulipala2016compressing,rossi2018graphzip,fan2021making}. 
While these techniques often achieve high compression ratios, the resulting representations are typically not directly usable by general-purpose graph algorithms and may require substantial reconstruction overhead.

Feder and Motwani~\cite{FEDERcompress} introduced a graph compression method that preserves path information by partitioning edges into bicliques; Chavan et al.~\cite{chavan2026speeding} later improved its running time. In contrast, we study biclique coverings rather than partitions and explicitly formulate the optimization problem of minimizing the representation size. Unlike prior heuristic approaches, we provide approximation algorithms with provable guarantees for meaningful graph classes.

The $\pone$ problem has been studied in theory~\cite{tarjan1975complexity,chung1983decomposition,davoodi2016edge,tuza1984covering}. This body of work primarily focuses on establishing upper and lower bounds on the cost of an optimal solution.
The decision version was later shown to be $\mathsf{NP}$-Complete in~\cite{Drawing-Hardness,edge-RMP-hardness}. For an overview of $\pone$ and its data mining applications, see~\cite{role-mining-survey}.
To the best of our knowledge, no approximation algorithms or inapproximability results are known. Existing approaches are limited to heuristics without theoretical performance guarantees~\cite{role-mining-survey,newbery1989edge,edge-RMP-hardness}.

A related and well-studied variant is the bipartite dimension problem where the goal is to find the minimum number of bicliques whose union of edges covers all edges of the graph. It is known to be $\mathsf{NP}$-hard~\cite{orlin1977contentment, gottlieb2005efficient} and admits strong inapproximability results~\cite{simon1990approximate, gruber2007inapproximability, chalermsook2014nearly, feige1998zero, khot2006better, zuckerman2006linear}.
On the positive side, the problem is fixed-parameter tractable~\cite{fleischner2009covering, nor2012mod, chandran2017parameterized}
and exact polynomial-time algorithms are known for some graph classes
~\cite{amilhastre1998complexity, lubiw1991weighted, muller1996edge}. 

A variety of problems have been studied on $\delta$-disk graphs, including the densest subgraph problem~\cite{har2025approximating}, the maximum clique problem~\cite{clark1990unit}, and others~\cite{agarwal2014near, chekuri2020fast}. 
Recently, Lahn et al.~\cite{lahn2021faster} designed an efficient algorithm for the maximum cardinality bipartite matching problem on $\delta$-disk bipartite graphs. 

%% file: conclusion.tex
\vspace{-1em}
\section{Future Work}
\vspace{-0.5em}

Several directions remain for future work. For large $c$ the $\ptwo$ problem reduces to the bipartite dimension problem, which is $\mathsf{NP}$-hard to approximate~\cite{simon1990approximate, gruber2007inapproximability}. It is an open question whether similar inapproximability results hold for smaller values of $c$. Another direction is to design approximation or exact algorithms for $\pone$ and $\ptwo$ on graph classes beyond intersection bipartite graphs. Finally, it would be interesting to integrate our representations into graph compression pipelines and evaluate the resulting compression ratios on real-world graphs.


%% file: appendix.tex
\appendix

\section{Missing proofs from Section~\ref{sec:hardnessnew}}
\label{appndx:NPhard}

\begin{proof}[Proof of Lemma~\ref{lem:gadget-biclique-count}]
An unselected vertex $w_i$ requires one biclique containing $z_i$,
whereas a selected vertex requires two such bicliques. Moreover, a
single biclique cannot contain both $z_i$ and $z_j$ for $i\neq j$,
since
\[
N(z_i)\cap N(z_j)=\emptyset.
\]
Hence, if $k'$ vertices are selected, the vertex gadgets require at
least
\[
(n-k')+2k'=n+k'
\]
bicliques.

Now consider an edge $e=\{w_i,w_j\}\in F$. Consider the four edges
\[
(p_e^1,u_e^1),\qquad
(p_e^2,u_e^2),\qquad
(p_e^3,u_e^3),\qquad
(p_e^4,u_e^4).
\]
No two of these edges can belong to the same biclique. Indeed, for
any two distinct edges in this collection, at least one of the two
corresponding cross edges is absent from $\G$. Therefore, at least
four distinct bicliques are necessary to cover them.

Each of these bicliques contains one of the private vertices
$p_e^1,\ldots,p_e^4$. Consequently, none of them is a biclique
already counted for a vertex gadget, since no private vertex $p_e^r$
is adjacent to any $z_i$. Furthermore, a biclique containing a
private vertex of the gadget corresponding to $e$ cannot contain a
private vertex of a different edge gadget. Thus, the four bicliques
can be charged uniquely to the gadget corresponding to $e$.

It remains to consider the case in which neither $w_i$ nor $w_j$
is selected. Consider the following five edges:
\[
(x_i,u_e^1),\qquad
(y_j,u_e^2),\qquad
(y_i,u_e^3),\qquad
(x_j,u_e^4),\qquad
(p_e^1,u_e^5).
\]
No two of these five edges can be contained in the same biclique.
For two of the first four edges this follows because each of
$x_i,y_j,y_i,x_j$ is adjacent to exactly its corresponding vertex
among $u_e^1,u_e^2,u_e^3,u_e^4$ within this gadget. Moreover,
$u_e^5$ is adjacent to none of $x_i,y_j,y_i,x_j$, so the last edge
cannot share a biclique with any of the first four.

Hence, at least five distinct bicliques are needed to cover these
five edges. We claim that each of these five bicliques contains a
private vertex of the gadget $e$. This is immediate for the
biclique covering $(p_e^1,u_e^5)$. Consider, for example, a
biclique covering $(x_i,u_e^1)$. Since $w_i$ is not selected, by
the normal form described above there is no biclique whose
$\V$-side is exactly $\{x_i\}$. Therefore, such a biclique must
contain another vertex on its $\V$-side. The only other neighbors
of $u_e^1$ are $p_e^1$ and $p_e^2$, so the biclique must contain
one of these private vertices. The same argument applies to
$(y_j,u_e^2)$, $(y_i,u_e^3)$, and $(x_j,u_e^4)$.

Thus, all five required bicliques contain private vertices of the
gadget $e$. They are therefore distinct from the vertex-gadget
bicliques and cannot simultaneously be charged to any other edge
gadget. Hence, the lower bounds are additive over the edge gadgets.
\end{proof}

\begin{proof}[Proof of Lemma~\ref{lem:gwbc-hardness}]
We first prove the forward direction. Suppose that $C\subseteq W$
is a vertex cover of $J$ of size $k\leq K$.

For every vertex $w_i\notin C$, we create the biclique
\[
\{x_i,y_i\}\times\{z_i\}.
\]
For every vertex $w_i\in C$, we instead create the two bicliques
\[
\{x_i\}\times\{z_i\}
\qquad\text{and}\qquad
\{y_i\}\times\{z_i\}.
\]
These bicliques cover all edges incident to the vertices $z_i$.
They use
\[
(n-k)+2k=n+k
\]
bicliques and initially contribute
\[
3(n-k)+4k=3n+k
\]
to the total biclique size.

We next cover the edges of every edge gadget. Consider an edge
$e=\{w_i,w_j\}\in F$, where $i<j$. Since $C$ is a vertex cover,
at least one of $w_i,w_j$ belongs to $C$.

Suppose first that $w_i\in C$. We add $u_e^1$ to the $\U$-side
of the existing biclique whose $\V$-side is $\{x_i\}$, and add
$u_e^3$ to the $\U$-side of the existing biclique whose $\V$-side
is $\{y_i\}$. We then introduce the following four bicliques:
\[
\{p_e^1,p_e^2\}\times\{u_e^1,u_e^5\},
\]
\[
\{p_e^3,p_e^4\}\times\{u_e^3,u_e^5\},
\]
\[
\{y_j,p_e^2,p_e^3\}\times\{u_e^2\},
\]
and
\[
\{x_j,p_e^4,p_e^1\}\times\{u_e^4\}.
\]
These bicliques, together with the two existing bicliques containing
$x_i$ and $y_i$, cover all edges of the gadget corresponding to
$e$.

If instead $w_j\in C$, we use the symmetric construction. We add
$u_e^2$ to the $\U$-side of the existing biclique whose $\V$-side
is $\{y_j\}$ and add $u_e^4$ to the $\U$-side of the existing
biclique whose $\V$-side is $\{x_j\}$. We then introduce
\[
\{x_i,p_e^1,p_e^2\}\times\{u_e^1\},
\]
\[
\{y_i,p_e^3,p_e^4\}\times\{u_e^3\},
\]
\[
\{p_e^2,p_e^3\}\times\{u_e^2,u_e^5\},
\]
and
\[
\{p_e^4,p_e^1\}\times\{u_e^4,u_e^5\}.
\]

Therefore, every edge of $J$ contributes exactly four new
bicliques. The four new bicliques have total size $16$, while the
two additions of gadget vertices to already existing bicliques
increase the total biclique size by $2$. Hence, every edge gadget
contributes exactly $18$ to the total biclique size.

Consequently, the resulting biclique edge cover $\edgecover$
satisfies
\[
|\edgecover|=n+k+4m
\]
and
\[
\mu(\edgecover)=3n+k+18m.
\]
It follows that
\begin{align*}
\sigma(\edgecover)
&=\mu(\edgecover)+c|\edgecover|\\
&=3n+k+18m+c(n+k+4m)\\
&=3n+18m+c(n+4m)+(1+c)k\\
&\leq 3n+18m+c(n+4m)+(1+c)K\\
&=\tau_c.
\end{align*}
This proves the forward direction.

For the reverse direction, suppose that $\G$ admits a biclique edge
cover $\edgecover$ satisfying
\[
\sigma(\edgecover)\leq\tau_c.
\]
By applying the normalization described above, we may assume that
$\edgecover$ is in normal form without increasing its cost.

Let $C'\subseteq W$ be the set of selected vertices, and let
\[
k'=|C'|.
\]
Let $F'\subseteq F$ denote the set of edges having at least one
endpoint in $C'$, and define
\[
q=|F\setminus F'|.
\]
Thus, $q$ is the number of edges of $J$ having neither endpoint
selected.

By Lemma~\ref{lem:gadget-biclique-count}, the vertex gadgets
require at least $n+k'$ bicliques. Each of the $m-q$ edge gadgets
corresponding to an edge in $F'$ requires at least four additional
bicliques, whereas each of the $q$ remaining edge gadgets requires
at least five additional bicliques. Hence,
\begin{align}
|\edgecover|
&\geq n+k'+4(m-q)+5q \notag\\
&=n+k'+4m+q.
\label{eq:gwbc-num-bicliques}
\end{align}

Similarly, by the size lower bounds from
\cite{edge-RMP-hardness}, the vertex gadgets contribute at least
$3n+k'$ to the total biclique size. Each edge gadget corresponding
to an edge in $F'$ contributes at least $18$, whereas every edge
gadget corresponding to an edge in $F\setminus F'$ contributes at
least $20$. Therefore,
\begin{align}
\mu(\edgecover)
&\geq 3n+k'+18(m-q)+20q \notag\\
&=3n+k'+18m+2q.
\label{eq:gwbc-total-size}
\end{align}

Combining Equations~\eqref{eq:gwbc-num-bicliques}
and~\eqref{eq:gwbc-total-size}, and using $c\geq0$, we obtain
\begin{align*}
\sigma(\edgecover)
&=\mu(\edgecover)+c|\edgecover|\\
&\geq
3n+k'+18m+2q+c(n+k'+4m+q)\\
&=
3n+18m+c(n+4m)+(1+c)k'+(2+c)q.
\end{align*}
Since $\sigma(\edgecover)\leq\tau_c$, it follows from
the definition of $\tau_c$ that
\[
(1+c)k'+(2+c)q\leq(1+c)K.
\]
Since $c\geq0$, we have $2+c\geq1+c$, and therefore
\[
(1+c)(k'+q)
\leq
(1+c)k'+(2+c)q
\leq
(1+c)K.
\]
Since $1+c>0$, we conclude that
\[
k'+q\leq K.
\]

Starting from $C'$, for every edge in $F\setminus F'$ we add an
arbitrary endpoint of that edge. Since there are $q$ such edges,
the resulting set $C$ is a vertex cover of $J$ and satisfies
\[
|C|\leq k'+q\leq K.
\]
Thus, $J$ admits a vertex cover of size at most $K$.
\end{proof}

\begin{proof}[Proof of Theorem~\ref{thm:gwbc-hardness}]
First, dGWBC belongs to NP. Given a collection of bicliques, we can
verify in polynomial time that every selected subgraph is a
biclique, that all edges of the input graph are covered, and that
the total cost is at most the prescribed threshold. Moreover, we
may assume that a biclique edge cover contains at most $|\E|$
bicliques. Indeed, if a biclique does not contain an edge that is
uncovered by all other bicliques, then that biclique can be removed
without increasing the cost. Hence, a polynomial-size certificate
always exists.

For NP-hardness, let $c\geq0$ be arbitrary. Given an instance
$(J,K)$ of Vertex Cover, we construct the graph $\G$ and the
threshold
\[
\tau_c=3n+K+18m+c(n+K+4m)
\]
as described above. The graph $\G$ contains $O(n+m)$ vertices and
edges, and its size is independent of the value of $c$. Moreover,
$\tau_c$ can be computed in polynomial time.

By Lemma~\ref{lem:gwbc-hardness}, $J$ has a vertex cover of size
at most $K$ if and only if $\G$ admits a biclique edge cover
$\edgecover$ satisfying
\[
\sigma(\edgecover)\leq\tau_c.
\]
Therefore, dGWBC is NP-hard for any rational value $c\geq0$.
Together with membership in NP, this proves the theorem.
\end{proof}

\section{Simpler proof that $\ptwo$ problem is $\mathsf{NP}$-Complete for constant $c$}
\label{sec:gwbc-hardold}

It is known that covering the edges of a bipartite graph with the minimum number of bicliques (the bipartite dimension problem) is NP-complete~\cite{orlin1977contentment}. Similarly, minimizing the total size of the bicliques in a cover (the $\dpone$ problem) is also NP-complete~\cite{edge-RMP-hardness}. 
The $\ptwo$ problem generalizes both objectives by jointly accounting for the number and total size of the selected bicliques. Indeed, when $c=0$, $\ptwo$ reduces to $\pone$, while for sufficiently large $c$ (e.g., $c\ge n^6$), it reduces to the bipartite dimension problem (the term $c\kappa$ dominates the weight term $\sum_{\cover_i\in\edgecover}|\nodes(\cover_i)|$). Since both extremes are NP-complete, we show that for any fixed integer value of $c$, the decision version $\dptwo$ is also NP-complete.


To prove the hardness of $\dptwo$, we present a polynomial-time reduction from the $\dpone$ problem. Let $\G(\V, \U, \E)$ be a bipartite graph and let $\tau$ be a positive integer representing the target maximum weight for the $\dpone$ problem over $\G$. We construct a corresponding graph and target cost for the $\dptwo$ problem as follows.
Let $m = |\E|$ be the number of edges in $\G$. We choose an integer $k$ such that $k > cm$. We define the graph $\G'$ to be a $k$-blow-up of $\G$, denoted as $k\G$. For every vertex $x \in \nodes(\G)$, we create $k$ copies in $k\G$, denoted by $\{x_1, x_2, \dots, x_k\}$. If $(u, v) \in \E$, we connect every copy of $u$ to every copy of $v$ in $k\G$. Thus, if $u \in \V$, $v \in \U$, and $(u, v) \in \E$, then for all $i, j \in [k]$, there is an edge between $u_i$ and $v_j$ in $k\G$. More formally, we set $\nodes(k\G) = \{x_i \mid x \in \nodes(\G), i \in [k]\}$ and $\edges(k\G) = \{(u_i, v_j) \mid (u,v) \in \edges(\G), i,j \in [k]\}$. An example of this construction is shown in Figure~\ref{fig:gwbc_reduction_example}.

We set the target cost for the $\dptwo$ problem on $k\G$ to be $\tau' = k\tau + cm$. We now show that a solution to $\pone$ on $\G$ with cost at most $\tau$ corresponds to a solution to $\ptwo$ on $k\G$ with cost at most $\tau'$, and vice versa.

\input{bigraph}

\begin{lemma}
\label{lemma:forward}
There exists a biclique edge cover $\edgecover$ for $\G$ such that $\mu(\edgecover) \leq \tau$, if and only if there exists a biclique edge cover $\edgecover'$ for $k\G$ such that $\sigma(\edgecover') \leq k\tau + cm$.
\end{lemma}
\begin{proof}
We begin by showing the first direction.
Let $\edgecover = \{\cover_1, \dots, \cover_r\}$ be a biclique edge cover of $\G$ satisfying $\mu(\edgecover) = \sum_{j=1}^r |\nodes(\cover_j)| \leq \tau$. We construct a cover $\edgecover'$ for $k\G$ by expanding each biclique in $\edgecover$. For each $\cover_j \in \edgecover$, we define a corresponding subgraph $\cover^k_j$ in $k\G$ such that $\nodes(\cover^k_j) = \{v_i \mid v \in \nodes(\cover_j), i \in [k]\}$. Since $\cover_j$ is a biclique in $\G$, it follows from the construction of $k\G$ that for every $v \in \Vnodes(\cover_j)$ and $u \in \Unodes(\cover_j)$, all edges between copies of $u$ and copies of $v$ exist in $k\G$. Therefore, $\cover^k_j$ is a biclique in $k\G$. 

The collection $\edgecover' = \{\cover^k_1, \dots, \cover^k_r\}$ forms a valid biclique edge cover for $k\G$ since every edge $(u_i, v_j)$ in $\edges(k\G)$ corresponds to an edge $(u, v)$ in $\edges(\G)$, which is covered by some $\cover_l$. Therefore, the edge $(u_i, v_j)$ is covered by $\cover^k_l \in \edgecover'$.
The size of each blown-up biclique is $|\nodes(\cover^k_j)| = k \cdot |\nodes(\cover_j)|$, by definition. Therefore, the total weight portion of the cost of the cover $\edgecover'$ is $\sum_{j=1}^r |\nodes(\cover^k_j)| = k \cdot \mu(\edgecover) \leq k\tau$.
Regarding the constant cost component of $\sigma$, note that $|\edgecover'| = |\edgecover| = r$. Since a trivial edge cover consisting of individual edges has size $m$ (since $|\E|=m$), we can assume without loss of generality that $r \leq m$. Thus, the total cost for $\edgecover'$ is:
\begin{equation*}
\sigma(\edgecover') = c|\edgecover'| + \sum_{j=1}^r |\nodes(\cover^k_j)| \leq cm + k\tau.
\end{equation*}
This completes the first part of the lemma.

Next, we show the other direction. Let $\edgecover'$ be a biclique edge cover of $k\G$ satisfying the cost constraint. For all vertices $x \in \nodes(k\G)$ let $n_x(\edgecover')$ denote the number of bicliques in $\edgecover'$ that contain the vertex $x$. The total weight component of the cost function $\sigma(\edgecover')$ can be written as $\sum_{S' \in \edgecover'} |\nodes(S')| = \sum_{x \in \nodes(k\G)} n_x(\edgecover')$. 
For each vertex $x \in \nodes(\G)$, we identify the copy in $k\G$ that contributes the least to the weight of the cover. Formally, we define $i(x) = \arg\min_{j \in [k]} \{n_{x_j}(\edgecover')\}$ to be the index of the copy of $x$ appearing in the fewest bicliques in $\edgecover'$. We construct a cover $\edgecover$ for $\G$ by projecting the bicliques of $\edgecover'$ onto these specific copies. Specifically, for each $S' \in \edgecover'$, we construct a set $\pi(S') = \{x \in \nodes(\G) \mid x_{i(x)} \in \nodes(S')\}$. We discard any empty sets resulting from this projection. Finally, we set $\edgecover = \{\pi(S') \mid S' \in \edgecover'\}$.

We first verify that $\edgecover$ is a valid cover for $\G$. Consider any edge $(u, v) \in \E$. By the definition of $k\G$, the copies $u_{i(u)}$ and $v_{i(v)}$ are connected by an edge in $k\G$. Therefore, there must be some biclique $S' \in \edgecover'$ covering the edge $\{u_{i(u)}, v_{i(v)}\}$. By our construction, $u$ and $v$ will both be included in $\pi(S')$, implying $\pi(S')$ covers the edge $(u, v)$. Furthermore, it is easy to check that for all $S' \in \edgecover'$, the subgraph $\pi(S')$ is a biclique in $\G$, since $S'$ is a biclique in $k\G$, by definition. 
Thus, every $\cover \in \edgecover$ is a biclique, and $\edgecover$ is a valid biclique cover for $\G$.

Now we analyze the cost $\mu(\edgecover)$. By definition, for all $x \in \nodes(\G)$, the number of bicliques in $\edgecover$ containing $x$ (i.e., $n_{x}(\edgecover)$) is at most $n_{x_{i(x)}}(\edgecover')$. By our choice of $i(x)$, we have $n_{x_{i(x)}}(\edgecover') \leq \frac{1}{k} \sum_{j=1}^k n_{x_j}(\edgecover')$. Summing over all vertices in $\G$, we obtain:
\begin{align*}
    \mu(\edgecover) &= \sum_{x \in \nodes(\G)} n_{x_{i(x)}}(\edgecover') \leq \sum_{x \in \nodes(\G)} \frac{1}{k} \sum_{j=1}^k n_{x_j}(\edgecover')
    \\ &= \frac{1}{k} \sum_{x \in \nodes(k\G)} n_x(\edgecover') 
    \leq  \frac{1}{k} \left( \sum_{x \in \nodes(k\G)} n_x(\edgecover') + c|\edgecover'| \right)
    \\ &= \frac{1}{k}\sigma(\edgecover') \leq \frac{1}{k}(k\tau + cm) = \tau + \frac{cm}{k}.
\end{align*}
Since $\mu(\edgecover)$ must be an integer (as it is a sum of vertex counts) and we chose $k > cm$, we have $\frac{cm}{k} < 1$. This implies $\mu(\edgecover) \leq \tau$, and the result follows.
\end{proof}


\begin{theorem}
For any fixed integer $c \geq 0$, the $\dptwo$ problem is NP-complete.
\end{theorem}
\vspace{-1em}
\begin{proof}
First, we observe that $\dptwo$ is in NP. Given a collection of subgraphs, we can verify in polynomial time whether each subgraph is a biclique, whether they cover all edges in $\E$, and whether the total cost $\sigma$ is at most the given threshold.

Next, we show NP-hardness. 
We choose $k = cm + 1$ (so that we have $k > cm$). Since $m \leq n^2$ 
we have that 
$k$ is bounded by a polynomial in $n$. Consequently, the constructed graph $k\G$, which has $k \cdot n$ vertices, has size polynomial in the size of the input graph $\G$.
By Lemma~\ref{lemma:forward}, 
the instance of $\dpone$ has a solution with cost at most $\tau$ if and only if the instance of $\dptwo$ has a solution with cost at most $k\tau + cm$.
\end{proof}

\section{Missing details from Section~\ref{sec:algs}}
\label{appndx:algs}

\subsection{Algorithm for the $c$-densest subgraph problem}
\label{appndx:cdensest}
We first review the exact algorithm for the densest subgraph problem.
Given $\G$, Goldberg~\cite{goldberg1984finding} created a max-flow instance to check whether there exists a subgraph $X$ such that $\frac{|\edges(X)|}{|\nodes(X)|}\geq \gamma$, for any real number $\gamma$. They construct the directed graph $G'(V', E')$ such that $V'=\V\cup\U\cup\{s,t\}$, where $s$ is the source and $t$ is the sink node.
Furthermore $E'=\E\cup (\bigcup_{v\in \V}(s,v))\cup(\bigcup_{u\in \U}(s,u))\cup (\bigcup_{v\in \V}(v,t))\cup(\bigcup_{u\in \U}(u,t))$.
Each edge $(v,u)\in\E$, with $v\in\V$ and $u\in\U$, is replaced by two directed edges $(v,u)$ and $(u,v)$, each of capacity $1$.
The capacity of each edge $(s,v)$ (or $(s,u)$) is $\deg(v)$ (or $\deg(u)$), and the capacity of each edge $(v,t)$ (or $(u,t)$) is $2\gamma$, where $\deg(v)$ is the degree of node $v$ in graph $\G$. 
Let $\bar{X}$ be the subgraph of $\G$ such that $\nodes(\bar{X})=(\V\cup\U)\setminus \nodes(X)$ and $\edges(\bar{X})=\{(u,v)\in \E\mid v\in \Vnodes(\bar{X}), u\in\Unodes(\bar{X})\}$. 
Let also $\edges(X,\bar{X})$ be the edges in $\E$ between nodes in $\nodes(X)$ and $\nodes(\bar{X})$.
They show that i) $\frac{|\edges(X)|}{|\nodes(X)|}\geq \gamma \Leftrightarrow \sum_{v\in \nodes(\bar{X})}\deg(v)+|\edges(X,\bar{X})|+2\gamma |\nodes(X)|\leq 2|\E|$, and ii) the cost of the $s$-$t$ cut of $G'$ defined by the nodes $\nodes(X)$ in the $s$ side and nodes $\nodes(\bar{X})$ in the $t$ side is exactly $\sum_{v\in \nodes(\bar{X})}\deg(v)+|\edges(X,\bar{X})|+2\gamma |\nodes(X)|$. Hence by the max-flow min-cut theorem, they compute a max-flow on $G'$
and if its value is at most $2|E|$ and $X\neq \emptyset$ then they run their procedure for a larger value of $\gamma$. Otherwise, they run their procedure for smaller values of $\gamma$. By running a binary search on the values of $\gamma$ they get the densest subgraph from $\G$, which is the subgraph $X$ returned in the last step of the binary search where the cost of the $s$-$t$ min-cut was at most $2|\E|$ and $X\neq \emptyset$.

For the $c$-densest subgraph problem, we have a similar construction. In fact, the graph $G'$ we compute is exactly the same. We decide whether there exists a subgraph $X$ such that $\frac{|\edges(X)|}{|\nodes(X)|+c}\geq \gamma$. Following the analysis in~\cite{gionis2016dense, goldberg1984finding} we get $\frac{|\edges(X)|}{|\nodes(X)|+c}\geq \gamma \Leftrightarrow \sum_{v\in \nodes(\bar{X})}\deg(v)+|\edges(X,\bar{X})|+2\gamma |\nodes(X)|\leq 2|\E|-2\gamma\cdot c$. We compute a max-flow on $G'$
and if its value is at most $2|E|-2\gamma\cdot c$ then we run our procedure for a larger value of $\gamma$. Otherwise, we run our procedure for smaller values of $\gamma$.
We note that we can also run a binary search on $\gamma$ since the function $$f(\gamma)= \min_{X\subseteq G}\left(\sum_{v\in \nodes(\bar{X})}\deg(v)+|\edges(X,\bar{X})|+2\gamma |\nodes(X)|\right)- 2|\E|+2\gamma\cdot c$$ is a non-decreasing function.

Overall, notice that $\gamma$ can take $O(|\E|\cdot (|\V|+|\U|))$ different values so the binary search executes $O(\log(\max\{|\V|,|\U|,c\}))$ steps. In each step of the binary search a maximum flow is computed on $G'$ that contains $O(|\V|+|\U|)$ nodes and $O(|\E|)$ edges. Using~\cite{orlin2013max}, the max-flow (min $s$-$t$ cut) on $G'$ is computed in $O((|\V|+|\U|)|\E|)$ time. Using the recent breakthrough of the max-flow problem~\cite{chen2025maximum}, we can compute the max-flow in each iteration of the binary search in $|\E|^{1+o(1)}$ time.

Next, we consider the approximation algorithm. 
Chekuri et al.~\cite{chekuri2022densest} consider the problem of
maximizing
$\frac{|\edges(X)|}{g(|\nodes(X)|)}$
and show that the greedy peeling algorithm gives a $1/3$-approximation in linear time
when $g$ is concave and satisfies $g(0)=0$. Their analysis uses the
assumption $g(0)=0$ to establish that $t/g(t)$ is nondecreasing.
In our case, $g(t)=t+c$ does not satisfy $g(0)=0$ when $c>0$;
however,
$\frac{t}{g(t)}=\frac{t}{t+c}$
is nondecreasing for every $c\geq 0$. Hence, the same analysis applies
and yields an $O(1)$-approximation for the $c$-densest subgraph problem.

Putting everything together, we conclude with Lemma~\ref{lem:denssub}.


\section{Missing details from Section~\ref{sec:ext}}
\label{appndx:exetnsions}

\subsection{Extension to other graph structures}
\label{appndx:extintersection}
For intersection bipartite graphs we extend the proof of
Lemma~\ref{lem:corr}. We first recall why the transformation to
$\Re^{2d}$ preserves the edge relation. For every $u\in\U$, define
\[
\widehat{\rho}_u=
(-\infty,\rect_u^{(1)+}]\times\cdots\times
(-\infty,\rect_u^{(d)+}]
\times
[\rect_u^{(1)-},\infty)\times\cdots\times
[\rect_u^{(d)-},\infty).
\]
We show that for every $v\in\V$ and $u\in\U$,
\begin{equation}
\label{eq:rects}
\rect_v\cap\rect_u\neq\emptyset
\quad\Longleftrightarrow\quad
p_v\in\widehat{\rho}_u.
\end{equation}

Indeed, the two hyper-rectangles $\rect_v$ and $\rect_u$ intersect
if and only if, for every coordinate $j\in[d]$, the intervals
$[\rect_v^{(j)-},\rect_v^{(j)+}]$ and
$[\rect_u^{(j)-},\rect_u^{(j)+}]$ intersect. This is equivalent to
\[
\rect_v^{(j)-}\leq \rect_u^{(j)+}
\qquad\text{and}\qquad
\rect_v^{(j)+}\geq \rect_u^{(j)-}
\]
for every $j\in[d]$. Equivalently,
\begin{align*}
\rect_v^{(1)-}&\in(-\infty,\rect_u^{(1)+}],\ldots,
\rect_v^{(d)-}\in(-\infty,\rect_u^{(d)+}],\\
\rect_v^{(1)+}&\in[\rect_u^{(1)-},\infty),\ldots,
\rect_v^{(d)+}\in[\rect_u^{(d)-},\infty),
\end{align*}
which, by the definition of $p_v$, is precisely
$p_v\in\widehat{\rho}_u$. Hence Equation~\eqref{eq:rects} holds.

Now consider the $i$-th iteration of the greedy procedure and let
\[
y=\bigcap_{u\in\Unodes(S_i^*)}\widehat{\rho}_u.
\]
Since $S_i^*$ is a biclique, for every
$v\in\Vnodes(S_i^*)$ and every $u\in\Unodes(S_i^*)$ we have
$\rect_v\cap\rect_u\neq\emptyset$. By Equation~\eqref{eq:rects},
$p_v\in\widehat{\rho}_u$ for every
$u\in\Unodes(S_i^*)$. Therefore,
\[
\{p_v\mid v\in\Vnodes(S_i^*)\}\subseteq y.
\]

Since $y$ is a hyper-rectangle in $\Re^{2d}$, all steps of the proof
of Lemma~\ref{lem:corr} follow almost verbatim, with
$O(\log^{2d}n_V)$ canonical hyper-rectangles instead of
$O(\log^d n_V)$.

\paragraph{Interval bipartite graphs}
A graph $\G(\V,\U,\E, \mathcal{I}_V, \mathcal{I}_U)$ is an interval bipartite graph if for every $v\in \V$ (resp. $u\in \U$), $v$ is associated with an interval $I_v=[I_v^-, I_v^+]\in \mathcal{I}_V$ (resp. $u$ is associated with$I_u=[I_u^-, I_u^+]\in\mathcal{I}_U)$ and $(v,u)\in \E$ if and only if $I_v\cap I_u\neq \emptyset$. Notice that this is an intersection bipartite graph in $\Re^1$.
Hence, the next theorem follows directly from Theorem~\ref{thm:approxG}.

\begin{theorem}
\label{thm:intGraph}
    Given an interval bipartite graph $\G(\V,\U,\E,\mathcal{I}_V,\mathcal{I}_U)$ such that $n_V=|\V|$, $n_U=|\U|$ (and without loss of generality $n_V\leq n_U$), and $|\E|=m$,  there exists an $O(\log(n_U)\cdot\log^{2} (n_V))$-approximation algorithm for the $\ptwo$ (and $\pone$) problem for any $c\geq 0$, that runs in $O(n_V\cdot (n_U+m)\cdot m\cdot \log(n_V))$ time.
\end{theorem}

\paragraph{$k$-NN bipartite graphs}
A graph $\G(\V,\U,\E,k)$ is a $k$-NN bipartite graph in $\Re^d$ if
$(v,u)\in\E$ if and only if $u$ is one of the $k$-nearest neighbors
of $v$ in $\U$ and $v$ is one of the $k$-nearest neighbors of $u$
in $\V$. We assume that there are no ties in the $k$-nearest-neighbor
distances; ties can be handled by a consistent arbitrary perturbation.

For every $v\in\V$ and $u\in\U$, let
\[
R_v=\dist_\infty(v,\textsf{k-NN}(v,\U))
\qquad\text{and}\qquad
R_u=\dist_\infty(u,\textsf{k-NN}(u,\V)),
\]
respectively. Hence,
\[
(v,u)\in\E
\quad\Longleftrightarrow\quad
\dist_\infty(v,u)\leq R_v
\text{ and }
\dist_\infty(v,u)\leq R_u.
\]

We map every $v\in\V$ to a point $p_v\in\Re^{3d}$ defined as
\[
p_v=
\big(
v^{(1)},\ldots,v^{(d)},
v^{(1)}-R_v,\ldots,v^{(d)}-R_v,
v^{(1)}+R_v,\ldots,v^{(d)}+R_v
\big).
\]
For every $u\in\U$, we define the hyper-rectangle
$\rho_u\subseteq\Re^{3d}$ as
\begin{align*}
\rho_u
={}&
\prod_{j=1}^d
[u^{(j)}-R_u,u^{(j)}+R_u]
\times
\prod_{j=1}^d
(-\infty,u^{(j)}]
\times
\prod_{j=1}^d
[u^{(j)},\infty).
\end{align*}
By construction,
\[
p_v\in\rho_u
\]
if and only if, for every $j\in[d]$,
\[
u^{(j)}-R_u\leq v^{(j)}\leq u^{(j)}+R_u
\]
and
\[
v^{(j)}-R_v\leq u^{(j)}\leq v^{(j)}+R_v.
\]
These conditions are equivalent to
\[
\dist_\infty(v,u)\leq R_u
\qquad\text{and}\qquad
\dist_\infty(v,u)\leq R_v,
\]
respectively. Therefore,
\begin{equation}
\label{eq:knn-lifting}
(v,u)\in\E
\quad\Longleftrightarrow\quad
p_v\in\rho_u.
\end{equation}

We can now apply the algorithm from Section~\ref{subsec:mainalg}
in $\Re^{3d}$. 

The correctness and approximation analysis now follow from
Lemma~\ref{lem:corr}. In particular, for the biclique $S_i^*$ used
in the $i$-th iteration of that proof, let
\[
y=\bigcap_{u\in\Unodes(S_i^*)}\rho_u.
\]
Since $S_i^*$ is a biclique, Equation~\eqref{eq:knn-lifting} implies
that
\[
\{p_v\mid v\in\Vnodes(S_i^*)\}\subseteq y.
\]
Since $y$ is a hyper-rectangle in $\Re^{3d}$, it can be decomposed
into $O(\log^{3d}n_V)$ canonical hyper-rectangles. Thus, all steps
of the proof of Lemma~\ref{lem:corr} follow verbatim, with
$O(\log^{3d}n_V)$ canonical hyper-rectangles instead of
$O(\log^d n_V)$.

\begin{theorem}
\label{thm:kNNGraph}
Given a $k$-NN bipartite graph $\G(\V,\U,\E,k)$ in $\Re^d$, where
$d=O(1)$, such that $n_V=|\V|$, $n_U=|\U|$ (and without loss of
generality $n_V\leq n_U$), and $|\E|=m$, there exists an
$O\bigl(\log(n_U)\cdot\log^{3d}(n_V)\bigr)$-approximation algorithm for the $\ptwo$ (and $\pone$) problem for
any $c\geq0$. The algorithm runs in
$O\bigl(
n_V\cdot(n_U+m)\cdot m\cdot\log^{3d-1}(n_V)
\bigr)$
time.
\end{theorem}

\subsection{Extension to any $\ell_\base^d$ metric: Theoretical approach}
\label{appndx:extellalpha}

We first define the notion of an \emph{$\eps$-biclique}. Given a graph $\G$ and an arbitrarily small constant $\eps\in(0,1)$, let $\G'(\V',\U',\E')$ be a graph where $\V'=\V$, $\U'=\U$, and $\E'=\{(v,u)\mid \dist_\base(v,u)\leq (1+\eps)\delta\}$. Note that $\E'$ contains all edges of $\G$, as well as additional pairs $(v,u)$ with distance greater than $\delta$ but at most $(1+\eps)\delta$.
An $\eps$-biclique $\cover'$ of $\G$ is defined as any biclique of the graph $\G'$. Observe that every biclique of $\G$ is also an $\eps$-biclique. An \emph{$\eps$-biclique edge cover} $\edgecover$ is a collection of $\eps$-bicliques whose union covers all edges of $\G$.

Instead of a range tree, we use the construction and properties of the \emph{BBD-tree}~\cite{arya2000approximate, arya1998optimal}. A BBD-tree is a hierarchical geometric data structure designed for range aggregation queries in $\Re^d$, where $d=O(1)$. Let $\rangetree'$ denote the BBD-tree constructed over $\V$. The BBD-tree induces a collection $\canonical$ of $O(|\V|\log |\V|)$ canonical axis-aligned hyper-rectangles.\footnote{A hyper-rectangle in the canonical set of a BBD-tree may contain a hole; that is, each region $x\in\canonical$ can be represented as the difference of two hyper-rectangles. For simplicity, we treat each such region $x$ as a standard hyper-rectangle, since all properties and arguments we claim remain unchanged.}
For any convex region $\rho\subseteq \Re^d$, there exists a subset $C(\rho)\subseteq \canonical$ of $O(\eps^{-d}+\log |\V|)$ pairwise disjoint hyper-rectangles such that $\V\cap \rho \subseteq \bigcup_{x\in C(\rho)}(\V\cap x)\subseteq \V\cap (\oplus_{\eps}\rho)$. Here, $\oplus_{\eps}\rho$ denotes the Minkowski sum of $\rho$ with a ball of radius $\eps\cdot \mathsf{radius}(\rho)$, that is, the region obtained by expanding $\rho$ by $\eps\cdot \mathsf{radius}(\rho)$ in all directions, where $\mathsf{radius}(\rho)$ denotes the radius of $\rho$.
As in the range tree construction, we assume without loss of generality that each hyper-rectangle $x\in \canonical$ is minimal with respect to $\V$.

The algorithm for general $\ell_\base^d$ metrics closely follows the algorithm from Theorem~\ref{thm:approx}, with the following modifications: (i) the hash table $B$ is defined over all pairs $(v,u)$ with $v\in \V$ and $u\in \U$; (ii) we use a BBD-tree $\rangetree'$ to get $\canonical$ constructed over $\V$; and (iii) in each iteration, for every hyper-rectangle $x\in \canonical$, we define a graph $G_x(V_x,U_x,E_x)$ where $V_x=\V\cap x$, $U_x=\{u\in \U\mid x\cap \V\subseteq \oplus_{\eps}\rho_u\}$, and $E_x=\{(v,u)\mid v\in V_x,\ u\in U_x,\ \dist_\base(v,u)\leq (1+\eps)\delta,\ B(v,u)=1\}$.

Using the same arguments as in the correctness proof of Lemma~\ref{lem:corr}, while accounting for the multiplicative $(1+\eps)$ approximation and the covering guarantees of the BBD-tree, it is straightforward to show that the algorithm returns an $\eps$-biclique edge cover $\edgecover$ satisfying $\sigma(\edgecover)\leq O((\eps^{-d}+\log n_V)\log n_U)\cdot \sigma(\edgecover^*)$, where $\edgecover^*$ denotes an optimal biclique edge cover of $\G$. We emphasize that $\edgecover^*$ is not an optimal $\eps$-biclique cover.
Skipping the remaining details, we conclude with the following theorem.
\begin{theorem}
\label{thm:approxanybase}
Given a \pbg\ $\G(\V,\U,\E)$ in $\ell_\base^d$, where $d=O(1)$ and $\base\geq 1$, $n_V=|\V|$, $n_U=|\U|$ (with $n_V\leq n_U$), and $|\E|=m$, there exists an algorithm that returns an $\eps$-biclique edge cover $\edgecover$ such that $\sigma(\edgecover)\leq O((\eps^{-d}+\log n_V)\log n_U)\cdot \sigma(\edgecover^*)$, where $\edgecover^*$ is an optimal solution to the $\ptwo$ problem for any $c\geq 0$. The algorithm runs in $O(n_V^3n_U^2\log n_V)$ time and uses $O(n_Vn_U+ n_V \log n_V)$ space.
\end{theorem}

\section{Additional experiments}
\label{appndx:exp}

\begin{figure*}[h]
    \centering
    \begin{tabular}{cccc}
        Adults & Credits & Gamma & POPSIM\\
        \includegraphics[width=0.22\textwidth]{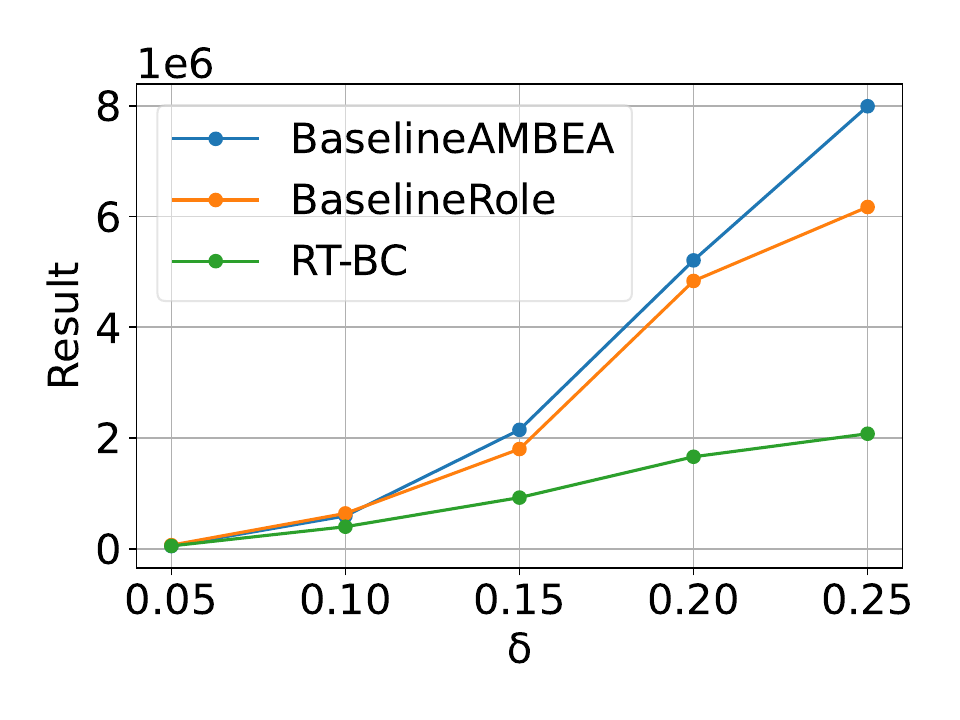} & 
        \includegraphics[width=0.22\textwidth]{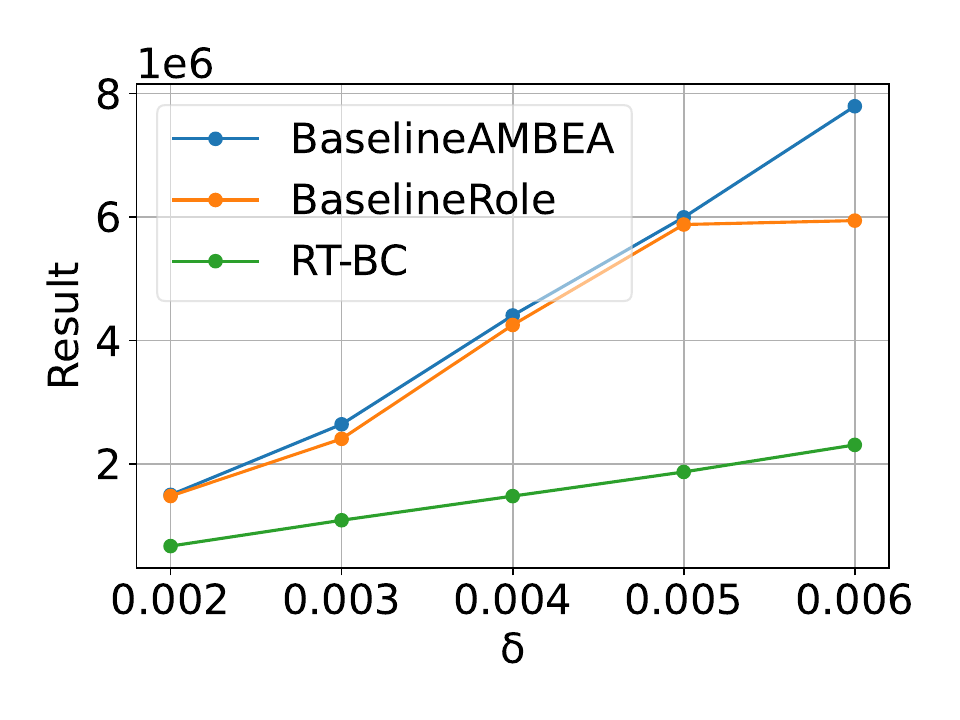} & 
        \includegraphics[width=0.22\textwidth]{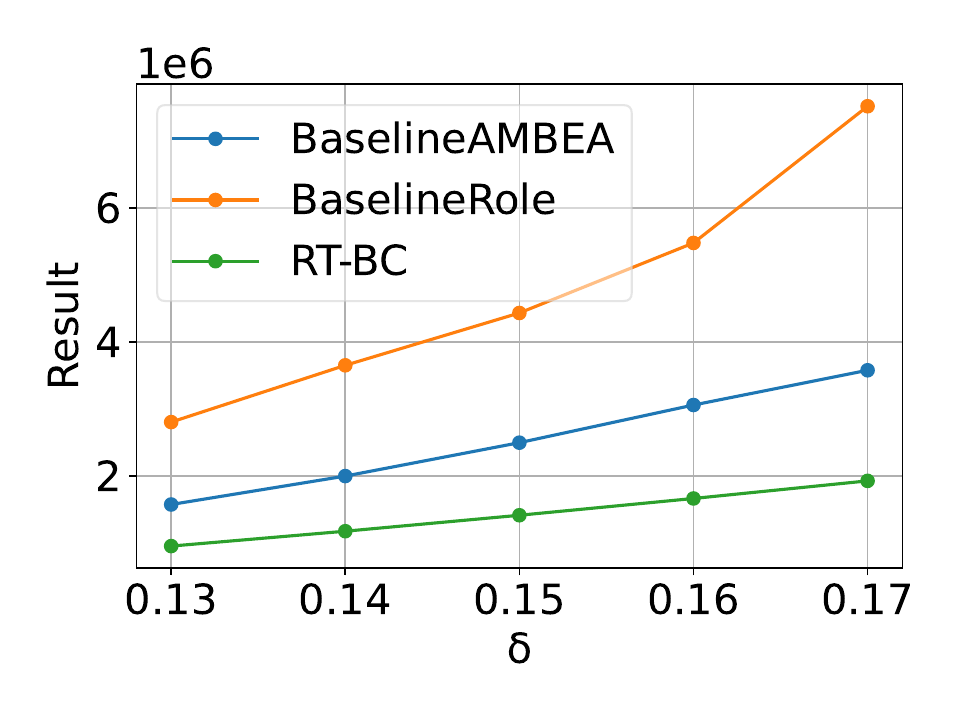} & 
        \includegraphics[width=0.22\textwidth]{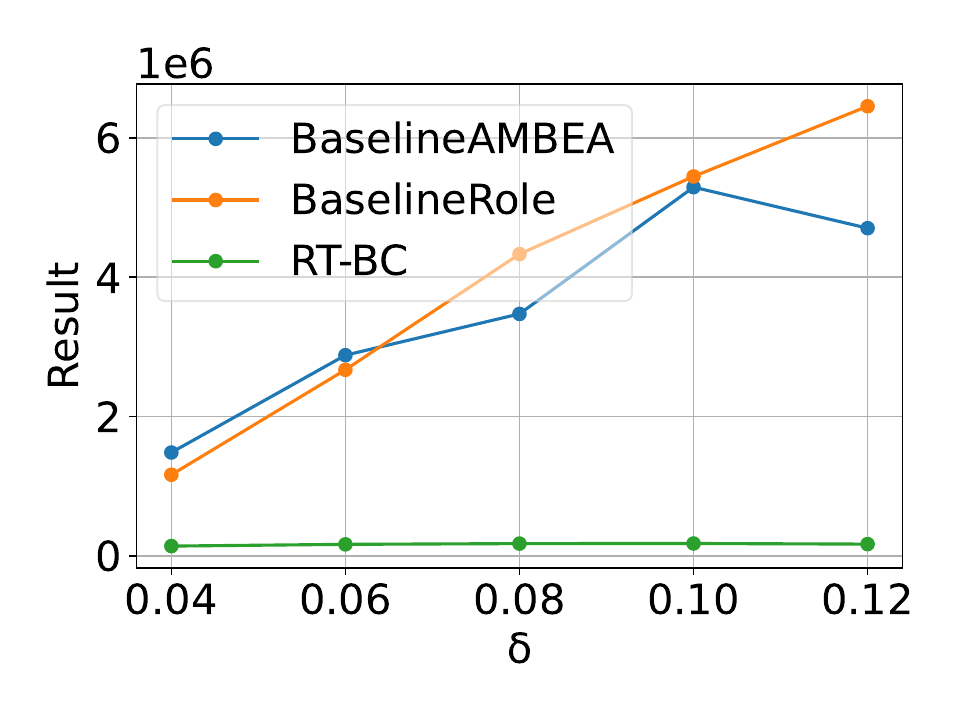} \\
        \multicolumn{4}{c}{(a) Size of representation}\\[1em]
        
        \includegraphics[width=0.22\textwidth]{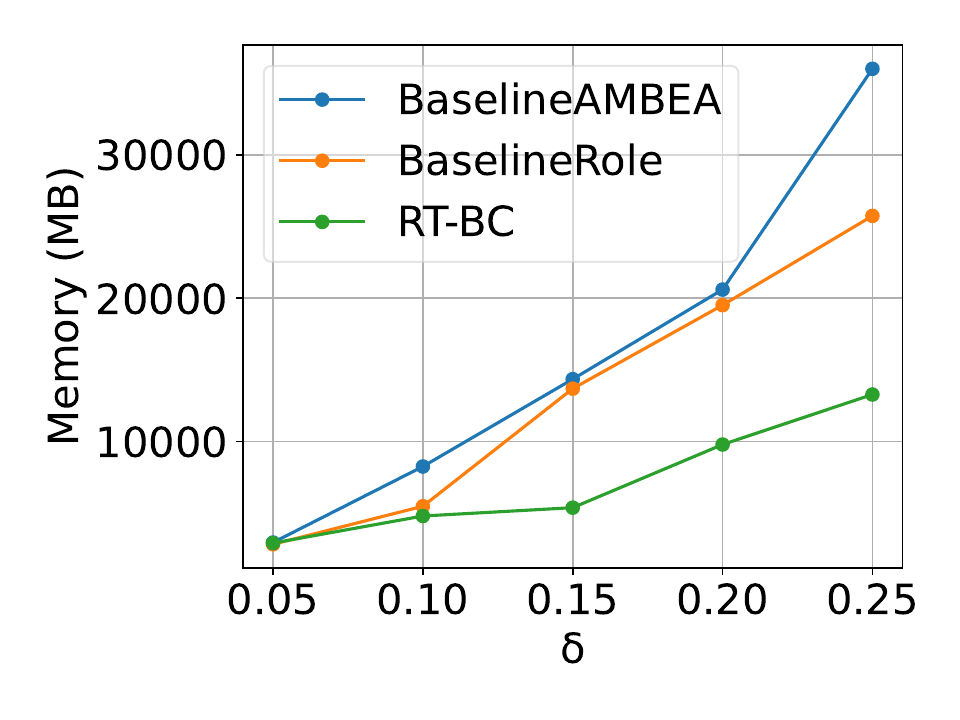} & 
        \includegraphics[width=0.22\textwidth]{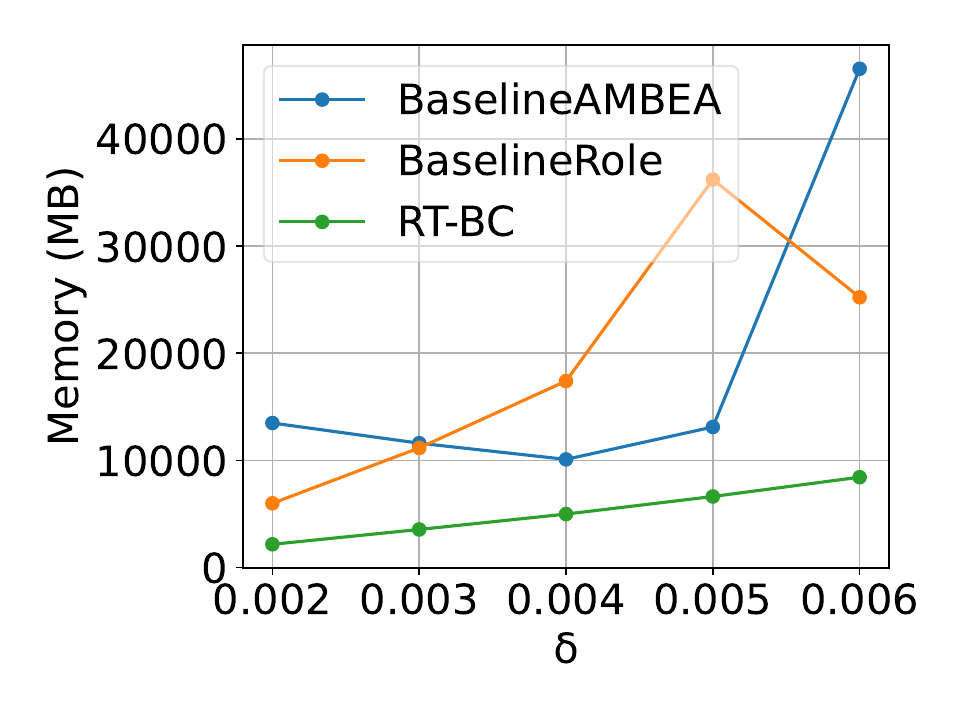} & 
        \includegraphics[width=0.22\textwidth]{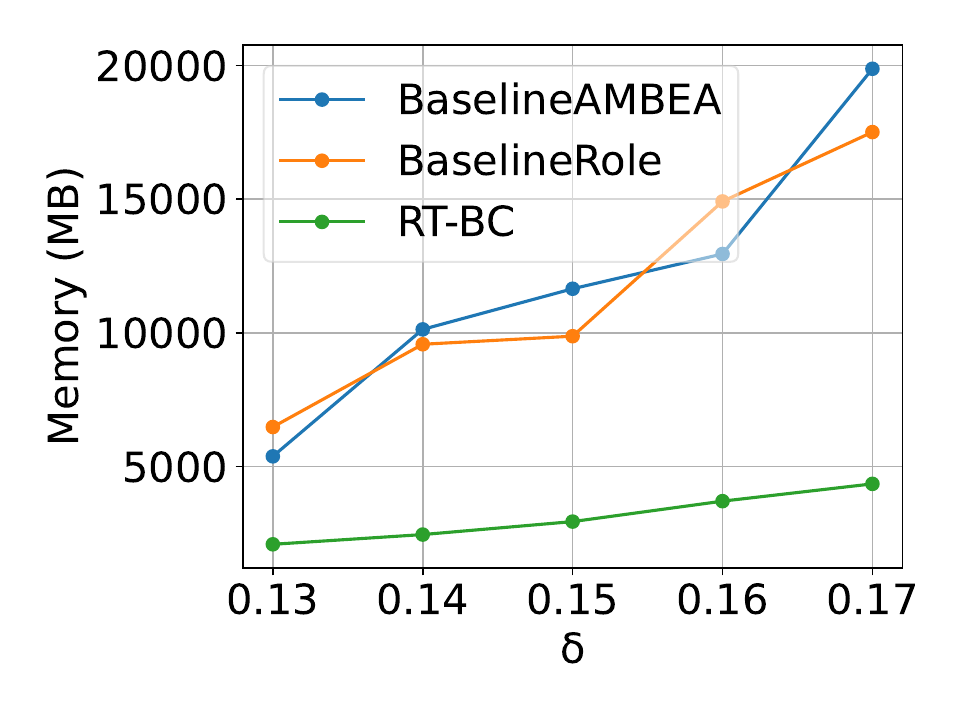} & 
        \includegraphics[width=0.22\textwidth]{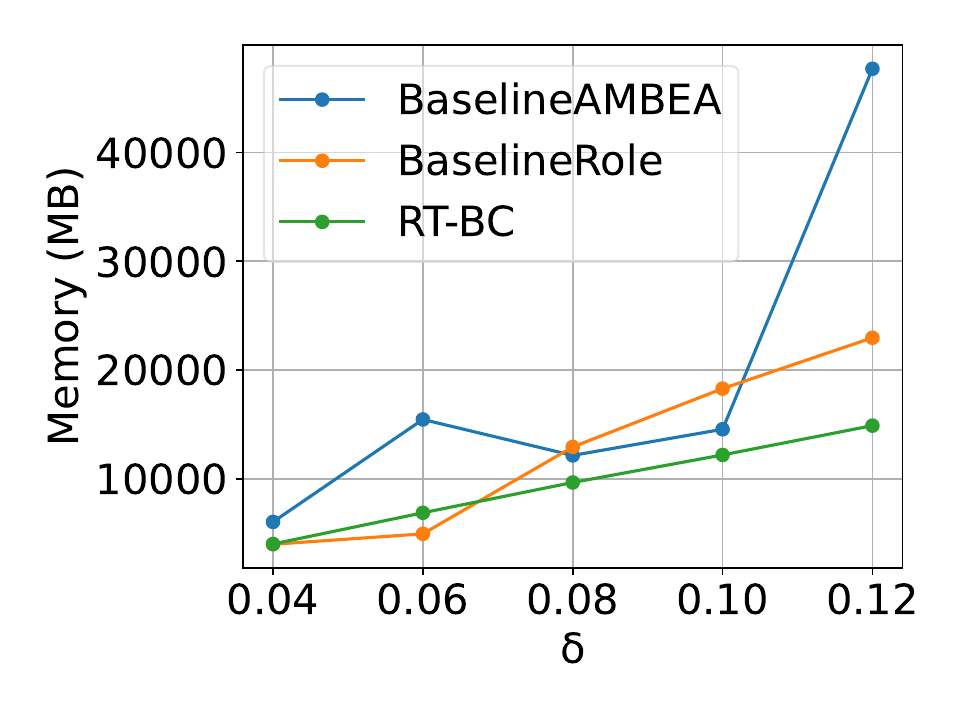} \\
        \multicolumn{4}{c}{(b) Peak memory in MB}\\[1em]
        
        \includegraphics[width=0.22\textwidth]{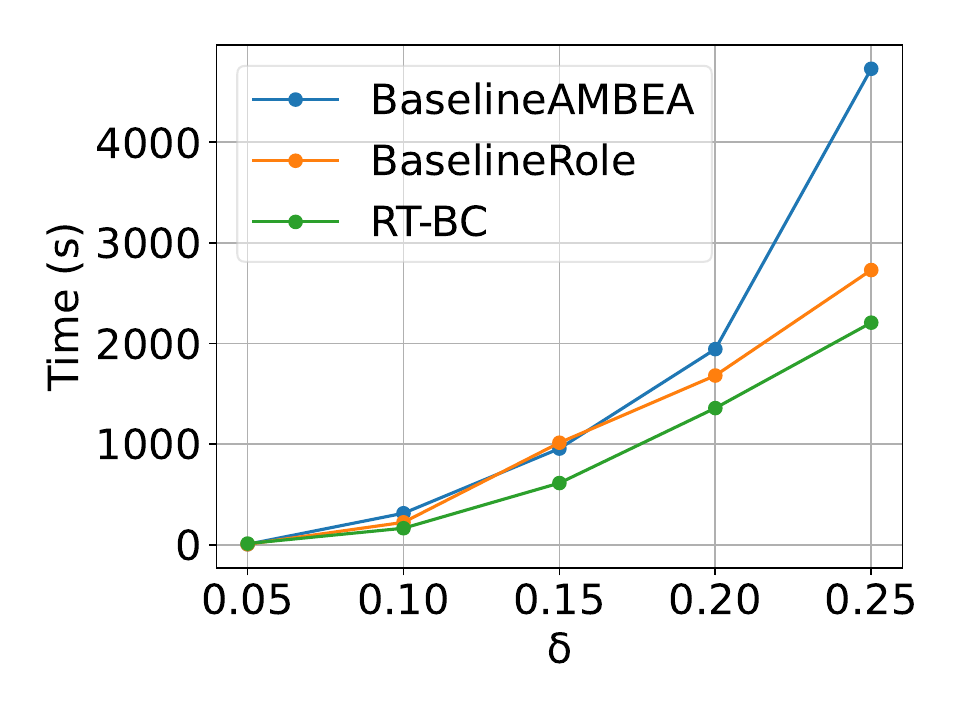} & 
        \includegraphics[width=0.22\textwidth]{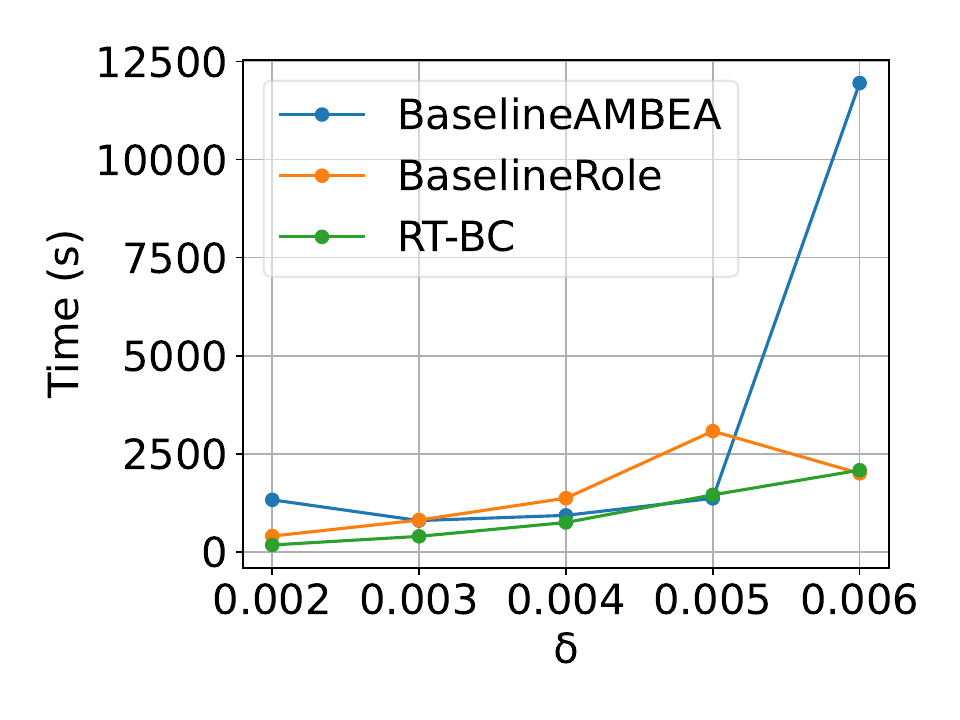} & 
        \includegraphics[width=0.22\textwidth]{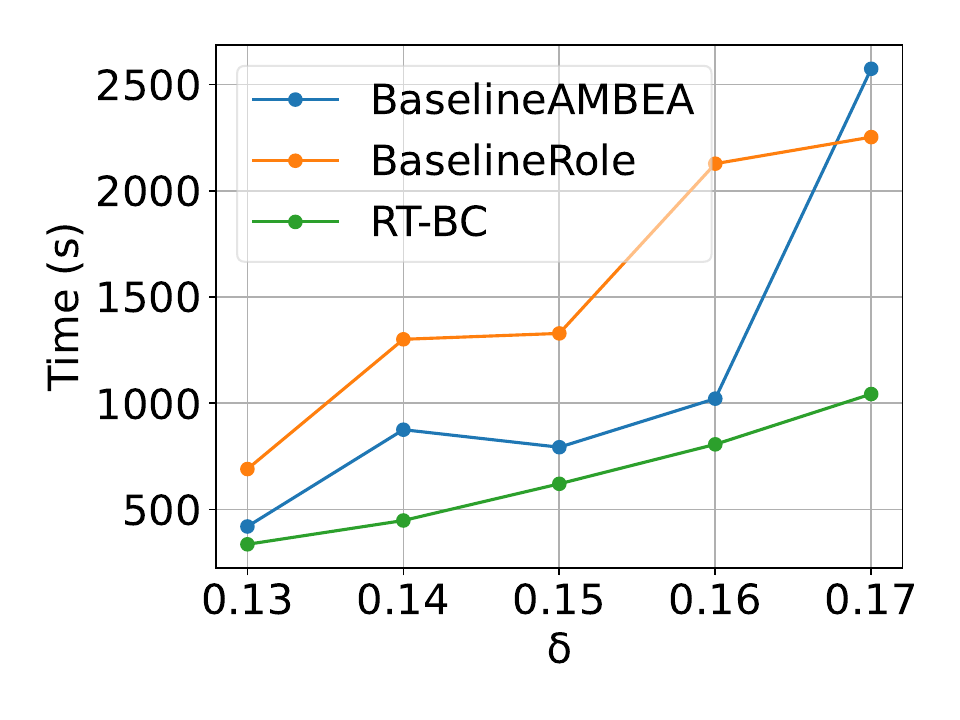} & 
        \includegraphics[width=0.22\textwidth]{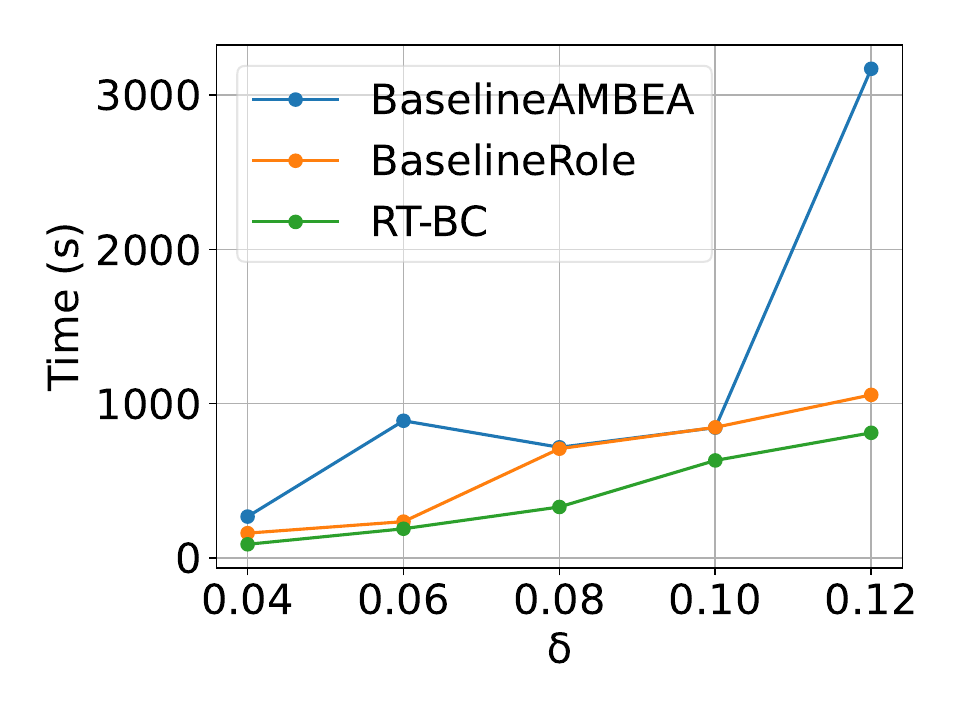} \\
        \multicolumn{4}{c}{(c) Running time in seconds}\\[1em]
    \end{tabular}

    \caption{Experimental results for $\pone$ in $\ell_\infty^d$ across all datasets}
    \label{fig:wbclinf}
\end{figure*}

\begin{figure*}[t]
    \centering
    \begin{tabular}{cccc}
        Adults & Credits & Gamma & POPSIM\\

        \includegraphics[width=0.22\textwidth]{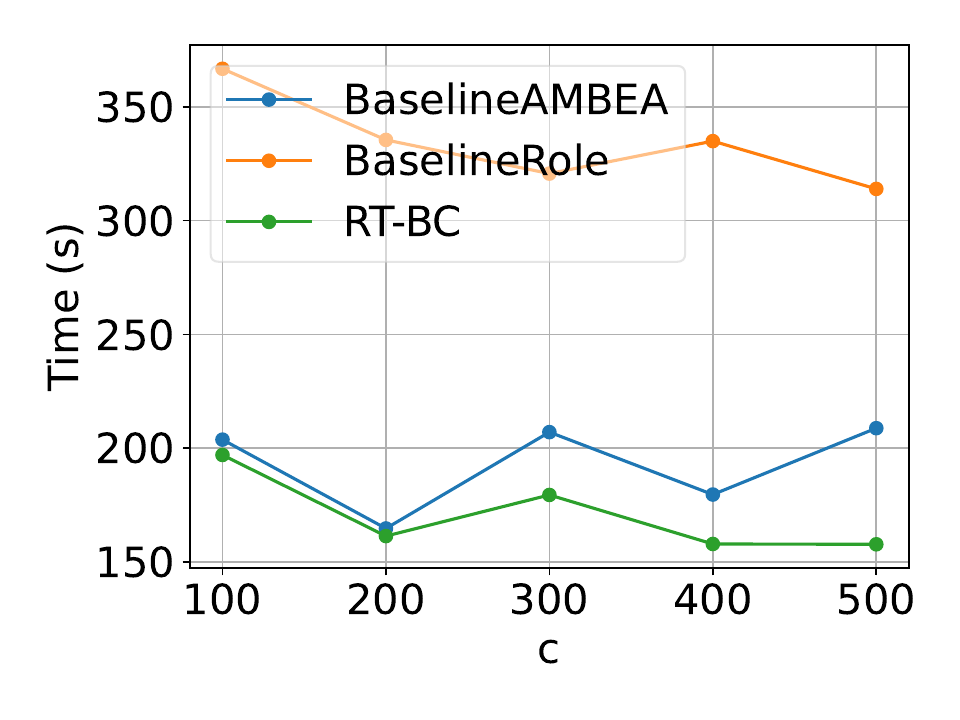} & 
        \includegraphics[width=0.22\textwidth]{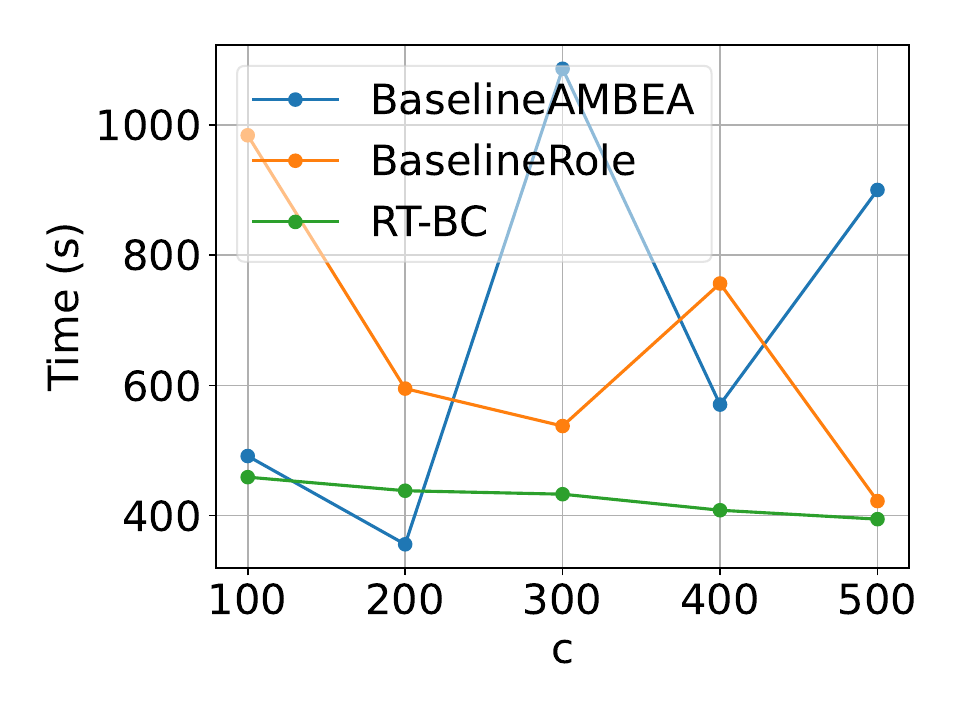} & 
        \includegraphics[width=0.22\textwidth]{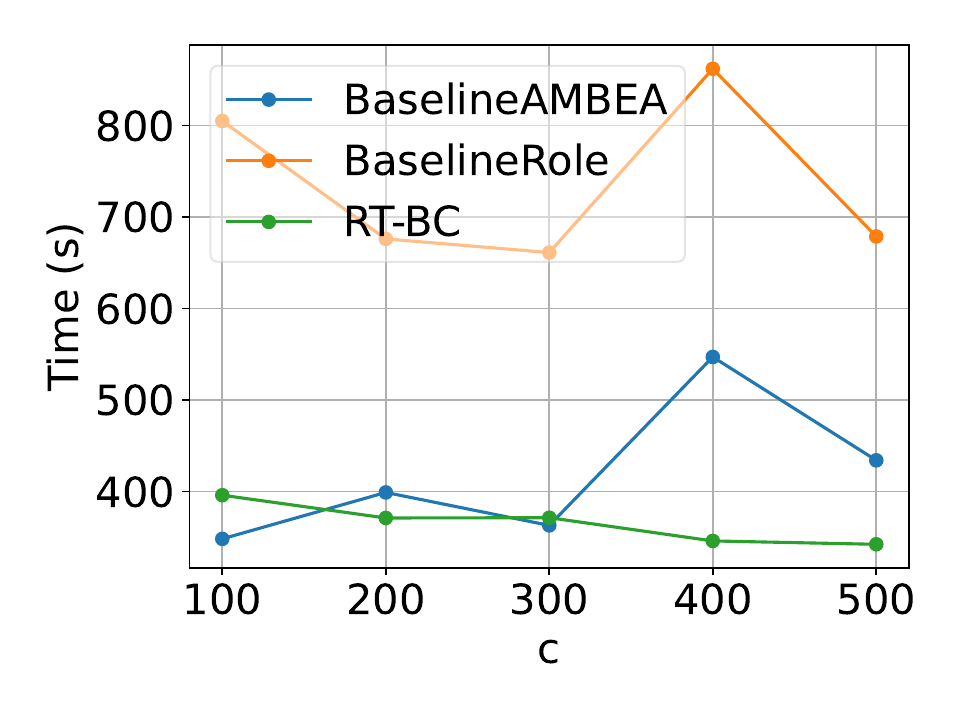} & 
        \includegraphics[width=0.22\textwidth]{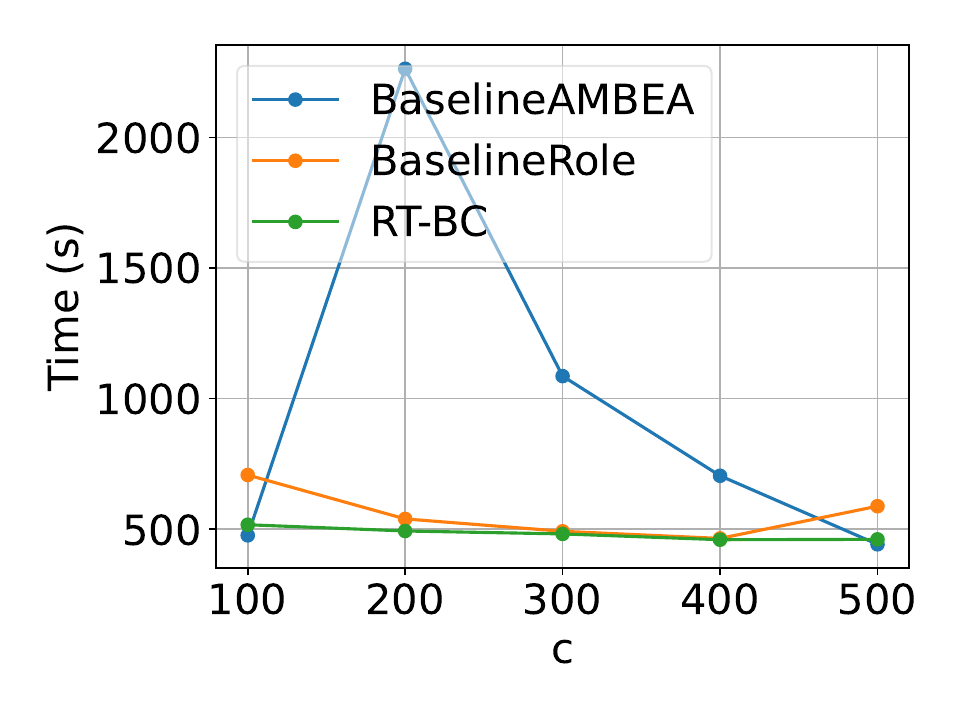} \\
        \multicolumn{4}{c}{(a) Running time in seconds}\\[1em]
    \end{tabular}

\vspace{-0.5em}    \caption{Experimental results for $\ptwo$ in $\ell_\infty^d$ across all datasets}
    \label{fig:wbclinfTime}
    \vspace{-1em}
\end{figure*}

\begin{table*}[h]
\centering
\begin{tabular}{|c|c|c|c|c|c|}
\hline
Dataset & $\delta$ & $m$ & $n$ & $d$ & Grouping Attribute \\
\hline
\multirow{5}{*}{Adults} & $0.05$ & $64,280$ & \multirow{5}{*}{\num{14846}} & \multirow{5}{*}{5} & \multirow{5}{*}{Sex} \\ \cline{2-3}
                     & $0.1$ & \num{854381} &                     &                    & \\ \cline{2-3}
                     & $0.15$ & \num{2589297} &                     &                  &   \\ \cline{2-3}
                     & $0.2$ & \num{5990399} &                     &                   &  \\ \cline{2-3}
                     & $0.25$ & \num{8714044} &                     &                  &   \\ \hline
\multirow{5}{*}{Credit} & $0.002$ & \num{1538057} & \multirow{5}{*}{\num{28726}} & \multirow{5}{*}{6} & \multirow{5}{*}{Sex} \\ \cline{2-3}
                     & $0.003$ & \num{2912661} &                     &                 &    \\ \cline{2-3}
                     & $0.004$ & \num{4437255} &                     &                 &    \\ \cline{2-3}
                     & $0.005$ & \num{6034989} &                     &                 &    \\ \cline{2-3}
                     & $0.006$ & \num{7805033} &                     &                 &    \\ \hline
\multirow{5}{*}{Gamma} & $0.13$ & \num{1511959} & \multirow{5}{*}{\num{18905}} & \multirow{5}{*}{10} & \multirow{5}{*}{Class} \\ \cline{2-3}
                     & $0.14$ & \num{1961614} &                     &                 &    \\ \cline{2-3}
                     & $0.15$ & \num{2477054} &                     &                 &    \\ \cline{2-3}
                     & $0.16$ & \num{3057096} &                     &                 &    \\ \cline{2-3}
                     & $0.17$ & \num{3702052} &                     &                 &    \\ \hline
\multirow{5}{*}{POPSIM} & $0.04$ & \num{1753508} & \multirow{5}{*}{8,591} & \multirow{5}{*}{2} & \multirow{5}{*}{Race} \\ \cline{2-3}
                     & $0.06$ & \num{3090592} &                     &                 &    \\ \cline{2-3}
                     & $0.08$ & \num{4328986} &                     &                 &    \\ \cline{2-3}
                     & $0.1$ & \num{5442100} &                     &                  &   \\ \cline{2-3}
                     & $0.12$ & \num{6449440} &                     &                 &    \\ \hline
\end{tabular}
\caption{\label{tab:datasetsellinfty}Statistics of the $\delta$-disk graphs used in our experiments for $\ell_\infty^d$.}
\end{table*}

%% file: bigraph.tex
\usetikzlibrary{shapes, fit, positioning, calc, backgrounds}

\begin{figure}[t]
    \centering
    \resizebox{\columnwidth}{!}{%
    \begin{tikzpicture}[
        scale=0.85, transform shape,
        vertex/.style={circle, draw, thick, fill=white, minimum size=16pt, inner sep=0pt, font=\small},
        copyvertex/.style={circle, draw, fill=gray!10, minimum size=13pt, inner sep=0pt, font=\scriptsize},
        group/.style={draw=black!40, dashed, rounded corners, inner sep=3pt},
        edge/.style={thick},
        samecopy/.style={thick, black},
        diffcopy/.style={thick, red, dashed, opacity=0.8},
        partlabel/.style={font=\small\bfseries, gray}
    ]

    \node[vertex] (a) at (0, 1.5) {$a$};
    \node[vertex] (b) at (0, 0) {$b$};
    \node[partlabel, above=0.15cm of a] {Part $\V$};

    \node[vertex] (c) at (2.5, 1.5) {$c$};
    \node[vertex] (d) at (2.5, 0) {$d$};
    \node[partlabel, above=0.15cm of c] {Part $\U$};
    
    \draw[edge] (a) -- (c);
    \draw[edge] (a) -- (d);
    \draw[edge] (b) -- (d);
    
    \node at (1.25, -0.8) {\textbf{(a) Graph $\G$}};

    \draw[->, thick, gray!60] (3.1, 0.75) -- 
        node[midway, above, black] {\scriptsize $k$-blow-up} 
        node[midway, below, black] {\scriptsize $(k=2)$} 
        (4.2, 0.75);

    \begin{scope}[xshift=5cm]
        
        \node[copyvertex] (a1) at (0, 1.8) {$a_1$};
        \node[copyvertex] (a2) at (0, 1.2) {$a_2$};
        \begin{pgfonlayer}{background}
            \node[fit=(a1)(a2), group, label={[gray, font=\scriptsize]left:Copies of $a$}] {};
        \end{pgfonlayer}

        \node[copyvertex] (b1) at (0, 0.3) {$b_1$};
        \node[copyvertex] (b2) at (0, -0.3) {$b_2$};
        \begin{pgfonlayer}{background}
            \node[fit=(b1)(b2), group, label={[gray, font=\scriptsize]left:Copies of $b$}] {};
        \end{pgfonlayer}

        \node[copyvertex] (c1) at (3.5, 1.8) {$c_1$};
        \node[copyvertex] (c2) at (3.5, 1.2) {$c_2$};
        \begin{pgfonlayer}{background}
            \node[fit=(c1)(c2), group, label={[gray, font=\scriptsize]right:Copies of $c$}] {};
        \end{pgfonlayer}

        \node[copyvertex] (d1) at (3.5, 0.3) {$d_1$};
        \node[copyvertex] (d2) at (3.5, -0.3) {$d_2$};
        \begin{pgfonlayer}{background}
            \node[fit=(d1)(d2), group, label={[gray, font=\scriptsize]right:Copies of $d$}] {};
        \end{pgfonlayer}

        \foreach \i in {1,2} { \foreach \j in {1,2} {
            \pgfmathparse{\i==\j ? "samecopy" : "diffcopy"} \edef\style{\pgfmathresult}
            \draw[\style] (a\i) -- (c\j);
        }}

        \foreach \i in {1,2} { \foreach \j in {1,2} {
            \pgfmathparse{\i==\j ? "samecopy" : "diffcopy"} \edef\style{\pgfmathresult}
            \draw[\style] (a\i) -- (d\j);
        }}

        \foreach \i in {1,2} { \foreach \j in {1,2} {
            \pgfmathparse{\i==\j ? "samecopy" : "diffcopy"} \edef\style{\pgfmathresult}
            \draw[\style] (b\i) -- (d\j);
        }}

        \node at (1.75, -0.8) {\textbf{(b) Graph $k\G$}};
    \end{scope}
    \end{tikzpicture}%
    }

    \caption{Illustration of the blow-up graph construction with $k=2$. (a) A bipartite graph $\G$ with vertices $\V=\{a,b\}$ and $\U=\{c,d\}$. (b) The constructed graph $k\G$. Solid black edges connect corresponding copies, while dashed red edges connect non-corresponding copies.}
    \label{fig:gwbc_reduction_example}
    \vspace{-1.2em}
\end{figure}